\documentclass{llncs}
\usepackage[hidelinks]{hyperref}

\usepackage{geometry}
\usepackage{caption}

\usepackage[title]{appendix}
\usepackage{graphicx}
\usepackage{amssymb,amsmath,bm}
\usepackage{algorithm}
\usepackage[noend]{algcompatible}
\usepackage{appendix}
\usepackage{multicol}
\usepackage[dvipsnames]{xcolor}
\usepackage{tablefootnote}
\usepackage[T1]{fontenc}

\usepackage{enumitem}
\usepackage{marvosym}

\usepackage{multirow,booktabs,arydshln}
\newcommand{\tabincell}[2]{\begin{tabular}{@{}#1@{}}#2\end{tabular}}

\definecolor{camel}{rgb}{0.76, 0.6, 0.42}

\algdef{SE}[PROCEDURE]{Procedure}{EndProcedure}{}{\algorithmicend\ }
\algtext*{Procedure}
\algtext*{EndProcedure}
\newcommand{\If}{\IF}
\newcommand{\Else}{\ELSE}
\newcommand{\EndIf}{\ENDIF}

\algrenewcommand\algorithmicwhile{\textbf{upon}}
\newcommand{\Upon}{\WHILE}
\newcommand{\EndUpon}{\ENDWHILE}

\algrenewcommand\algorithmicloop{\textbf{wait}}
\newcommand{\Wait}{\LOOP}
\newcommand{\EndWait}{\ENDLOOP}

\newcommand{\Repeat}{\LOOP}
\newcommand{\EndRepeat}{\ENDLOOP}

\algrenewcommand\algorithmicindent{1em}
\algrenewcommand{\algorithmiccomment}[1]{\hskip0.0em{\color{camel}   $\rhd$ #1}}

\newcommand{\argmax}{\mathop{\mathrm{argmax}}\nolimits} 

\newcommand{\ID}{\mathsf{ID}}

\newcommand{\node}{\mathcal{P}}
\newcommand{\hash}{\mathcal{H}}
\newcommand{\adv}{\mathcal{A}}
\newcommand{\cha}{\mathcal{C}}
\newcommand{\negl}{\mathsf{negl}}
\newcommand{\bigO}{\mathcal{O}}

\newcommand{\AVSS}{\mathsf{AVSS}}
\newcommand{\sssh}{\mathsf{AVSS\textrm{-}Sh}}
\newcommand{\ssrec}{\mathsf{AVSS\textrm{-}Rec}}
\newcommand{\wcs}{\mathsf{WCS}}
\newcommand{\wcrb}{\mathsf{\gamma\textrm{-}coin}}
\newcommand{\coin}{\mathsf{Coin}}
\newcommand{\elect}{\mathsf{Election}}
\newcommand{\seedgen}{\mathsf{Seeding}}
\newcommand{\scrape}{\mathsf{Scrape}}
\newcommand{\PVSS}{\mathsf{PVSS}}

\newcommand{\Gen}{\mathsf{KenGen}}

\newcommand{\Vrfy}{\mathsf{Vrfy}}
\newcommand{\Sigforge}{\mathsf{Sig\textrm{-}forge}}

\newcommand{\Adv}{\mathcal{A}}
\newcommand{\Q}{\mathcal{Q}}
\newcommand{\RBC}{\mathsf{RBC}}
\newcommand{\CBC}{\mathsf{RBC}}

\newcommand{\C}{\mathit{C}}
\newcommand{\Proof}{\mathit{\Pi}}
\newcommand{\A}{\mathit{A}}
\newcommand{\G}{\mathit{G}}

\newcommand{\R}{\mathit{R}}
\newcommand{\Si}{\mathit{S}}
\newcommand{\K}{\mathit{K}}

\newcommand{\VRF}{\mathsf{VRF}}
 
\newcommand{\KEYSHARE}{\textsc{KeyShare}}
\newcommand{\STORED}{\textsc{KeyStored}}

\newcommand{\CIPHER}{\textsc{Cipher}}
\newcommand{\VAL}{\textsc{Val}}
\newcommand{\ECHO}{\textsc{Echo}}
\newcommand{\READY}{\textsc{Ready}}
\newcommand{\KREC}{\textsc{KeyRec}}

\newcommand{\RECREQ}{\textsc{RecRequest}}
\newcommand{\LOCK}{\textsc{Lock}}
\newcommand{\COMMIT}{\textsc{Commit}}
\newcommand{\CONFIRM}{\textsc{Confirm}}
\newcommand{\CANDIDATE}{\textsc{Candidate}}

\newcommand{\KEY}{\textsc{Key}}

\newcommand{\SCRIPT}{\textsc{PvssScript}}
\newcommand{\LOCKPVSS}{\textsc{AggPvss}}
\newcommand{\COMMITPVSS}{\textsc{AggPvssCommit}}
\newcommand{\CONFIRMPVSS}{\textsc{AggPvssStored}}
\newcommand{\SEEDSHARE}{\textsc{SeedShare}}
\newcommand{\SEED}{\textsc{Seed}}
\newcommand{\SEEDECHO}{\textsc{SeedEcho}}
\newcommand{\SEEDREADY}{\textsc{SeedReady}}

\newcommand{\AUX}{\textsc{Aux}}
\newcommand{\CONF}{\textsc{Conf}}

\newcommand{\Sign}{\mathsf{Sign}}
\newcommand{\Eval}{\mathsf{VRF.Eval}}

\newcommand{\Verify}{\mathsf{SigVerify}}
\newcommand{\VerifyVRF}{\mathsf{VRF.Verify}}
\newcommand{\GenVRF}{\mathsf{VRF.Gen}}

\newcommand{\ABA}{\mathsf{ABA}}
\newcommand{\ADKG}{\mathsf{ADKG}}
\newcommand{\VBA}{\mathsf{VBA}}
\newcommand{\aba}{\mathsf{\textrm{$b$-}biased\textrm{-}ABA}}
\newcommand{\zeroaba}{\mathsf{\textrm{$0$-}biased\textrm{-}ABA}}

\newcommand{\sh}{\mathsf{sh}}
\newcommand{\cmt}{\mathsf{cmt}}
\newcommand{\cmax}{\mathsf{vrf_{max}}}

\newcommand{\deal}{\mathsf{Deal}}
\newcommand{\weights}{\mathsf{Weights}}
\newcommand{\vrfyscript}{\mathsf{VrfyScript}}
\newcommand{\vrfysecret}{\mathsf{VrfySecret}}
\newcommand{\agg}{\mathsf{AggScripts}}
\newcommand{\share}{\mathsf{GetShare}}
\newcommand{\vrfyshare}{\mathsf{VrfyShare}}
\newcommand{\aggshares}{\mathsf{AggShares}}
\newcommand{\pvss}{\mathsf{pvss}}
\newcommand{\param}{\mathsf{param}}
\newcommand{\ek}{{ek}}
\newcommand{\dk}{{dk}}
\newcommand{\allek}{\mathsf{ek}}
\newcommand{\sk}{{sk}}
\newcommand{\vk}{{vk}}
\newcommand{\allvk}{\mathsf{vk}}
\newcommand{\pvsstag}{\mathsf{tag}}

\usepackage{multirow}
\usepackage{arydshln}

\newcommand{\ignore}[1]{}

\newcommand{\yuan}[1]{\textcolor{blue}{#1}}

\begin{document}
	\title{Efficient Asynchronous Byzantine Agreement\\ without Private Setups}
	%\title{Asynchronous VSS and Its Applications to Common Randomness and Consensuses\thanks{Supported by organization x.}}
	%
	%\titlerunning{Efficient Asynchronous BFT   without Private Setup}
	\author{}
	\institute{ }

	\author{Yingzi Gao\inst{1,3}  \and
		Yuan Lu\inst{1}%(\Letter)  
		\and
		Zhenliang Lu\inst{2}%(\Letter) 
		\and  
		Qiang Tang\inst{2} 
		\and 
		Jing Xu\inst{1} 
		\and  
		Zhenfeng Zhang\inst{1}   %\orcidID{2222--3333-4444-5555}
	}

	%
	%\authorrunning{Anonymous Authors}
	\institute{ 
		Institute of Software, Chinese Academy of Sciences\\
		\email{\{yingzi2019,luyuan,xujing,zhenfeng\}@iscas.ac.cn} 
		\and
		School of Computer Science, The University of Sydney\\
		\email{zhlu9620@uni.sydney.edu.au,qiang.tang@sydney.edu.au} \\ 
		\and
				University of Chinese Academy of Sciences\\
	}

	\maketitle              % typeset the header of the contribution
	\begin{abstract}
		
		Though recent breakthroughs greatly improved the efficiency of asynchronous Byzantine agreement (BA) protocols, they mainly focused on the setting with private setups, e.g., assuming established non-interactive threshold cryptosystems. Challenges remain to reduce the large communication complexities in the absence of such setups. For example, Abraham et al. (PODC'21)  recently gave the first private-setup free construction for asynchronous   validated BA (VBA) with expected $\mathcal{O}(n^3)$ messages and $\mathcal{O}(1)$ rounds, relying on only public key infrastructure (PKI), but the design still costs $\mathcal{O}(\lambda n^3 \log n)$ bits. Here $n$ is the number of  parties, and $\lambda$ means the cryptographic security parameter capturing the length of hash, digital signature, etc. %, and  a term $|m|n^2$ related to the size of VBA input $|m|$ is ignored.
		
		We reduce  the   communication of private-setup free asynchronous BA  to expected $\mathcal{O}(\lambda n^3)$ bits. At the core of our design, we present a systematic treatment of reasonably fair common randomness protocols in the asynchronous network, and proceed as:
		\begin{itemize}
			\item We give an efficient  reasonably fair common coin  protocol  in the asynchronous setting  with only PKI setup. It costs only $\mathcal{O}(\lambda n^3)$ bit and $\mathcal{O}(1)$ asynchronous rounds, and ensures that  with at least $1/3$ probability, all honest parties can output a common bit that is as if randomly flipped. This directly renders more efficient private-setup free asynchronous binary agreement (ABA)  with expected $\mathcal{O}(\lambda n^3)$ bits  and expected constant rounds. 
			\item Then, we lift our common coin   to attain perfect agreement by using a single ABA.
			This gives us a reasonably fair random leader election protocol with expected $\mathcal{O}(\lambda n^3)$ communication and expected constant rounds. It is  pluggable in all existing VBA protocols (e.g., Cachin et al., CRYPTO'01;  Abraham et al., PODC'19; Lu et al., PODC'20) to remove the needed private setup or distributed key generation (DKG). As such, the  communication of private-setup free VBA is reduced to expected $\mathcal{O}(\lambda n^3)$ bits while preserving fast termination in  expected $\mathcal{O}(1)$ rounds. Moreover, our result paves a generic path to   private-setup free asynchronous BA protocols, as it is not restricted to merely improve Abraham et al.'s specific VBA protocol in PODC'21.
			%The resulting ABAs  reduce the communication   of existing   constant-round asynchronous distributed key generation of Abraham et al. (PODC'19) to  $\mathcal{O}(\lambda n^3)$.
			%This leader election primitive and its construction might be of independent interest.
			%\item Along the way, we improve an important building block, asynchronous verifiable secret sharing (Canetti and Rabin, STOC'93), and give a private-setup free implementation with only $\mathcal{O}(\lambda n^2)$ bits in the PKI setting. By contrast, prior art with the same communication complexity (Backes et al., CT-RSA'13) relies on  a private setup. 
		\end{itemize}
		Our results and techniques could be found  useful and interesting for a broad array of  applications such as asynchronous DKG and DKG-free asynchronous random beacon.
	\end{abstract}

    \setcounter{footnote}{0} 

	%!TEX root = main.tex

\section{Introduction}
\label{sec:intro}
 
Recently, following the unprecedented demand of deploying BFT protocols on the Internet for robust and highly available decentralized applications, renewed attentions are gathered to implement more efficient asynchronous Byzantine agreements \cite{abraham2018validated,keidar2021all,miller2016honey,guo2020dumbo,lu2020dumbo,yang2021dispersedledger}.
%
%\medskip
%\noindent
%{\bf Challenges of removing private setup}.
Nevertheless, asynchronous  protocols have to rely on randomized executions to circumvent the seminal FLP ``impossibility'' result \cite{fischer1982impossibility}.
In particular, to  quickly decide the    output in expected constant rounds of interactions, many asynchronous   protocols \cite{canetti1993fast,mostefaoui2015signature,miller2016honey,cachin2000random,cachin2001secure,abraham2018validated,lu2020dumbo,guo2020dumbo,keidar2021all,cohen2020not,blum2020asynchronous,abraham2010fast,abraham2021reaching} essentially need   {\em common randomness}, given which for ``costless'',  one at least can construct optimally resilient asynchronous Byzantine agreement (BA) protocols that   cost   expected $\bigO(n^2)$ messages and expected    $\bigO(1)$ rounds \cite{abraham2018validated,lu2020dumbo,mostefaoui2015signature,cachin2000random,cachin2001secure,crain2020two}. %, if assuming   common randomness for free. 

However,  % most efficient asynchronous BFT protocols  rely on non-interactive threshold cryptosystems \cite{cachin2000random} for  asymptotic and realistic efficient , which in practice are usually achieved through heavy private setup assumptions such as a trusted dealer.
 efficient ways to implement asynchronous common randomness in practice mostly rely on different varieties of private setups.
For example, initiated by M. Rabin \cite{rabin1983randomized}, it     assumes that a trusted dealer directly uses secret sharing to   distribute a large number of random secrets among the participating parties before the   protocol starts, so the parties can collectively reconstruct a sequence of common randomness while running the protocol.
Later, Cachin et al. \cite{cachin2000random} presented how to set up a non-interactive threshold pseudorandom function (tPRF) 
by assuming that a trusted dealer can faithfully share a short tPRF key, which now is  widely used by   existing practical asynchronous BFT protocols including \cite{abraham2018validated,keidar2021all,miller2016honey,lu2020dumbo}.
%
%
%

%In particular, at their heart, an efficient common randomness primitive for constant running time from non-interactive threshold cryptosystem .
%Many of them realize beside optimal resilience of tolerating $n/3$ corruptions (where $n$ represents the number of parties throughout the paper).
%
%Nevertheless, for    asymptotic and realistic efficiencies, most of these recent breakthroughs heavily  rely on non-interactive threshold cryptosystems \cite{cachin2000random}, which in practice are usually achieved through heavy private setup assumptions such as a trusted dealer, 
%or an (asynchronous) distributed key generation protocol emulating a trusted deader,
%
These private setups  might cause unpleasant  deployment hurdles, preventing asynchronous protocols  from  being widely used in  broader settings. %, as such heavy assumptions might not always hold in practice.
%For example, five years have passed into proof-of-stake style ``permissionless'' blockchains \cite{miller2016honey}
%and it remains unclear how to adapt  asynchronous BFT ,
%since Miller et al. planned to apply there is still on the run.
Hence, it becomes critical to reduce the   setup assumptions for easier real-world deployment. % and trust reduction.

%In particular, efficient asynchronous BFT without the private setup of threshold cryptosystems  might bring us new hope to fulfill the ambitious plan of Miller et al. to apply asynchronous BFT in   proof-of-stake style ``permissionless'' blockchains \cite{miller2016honey} and also provide the key building block for distributed key generation   in the fully asynchronous network \cite{abraham2021reaching,kokoris2020asynchronous}.

\smallskip
\noindent{\bf Existing efforts on reducing setups}.
%It is   challenging  to efficiently generate common randomness  on the fly without synchrony and private setup assumptions.
There are a few known approaches to construct private-setup free  asynchronous BA, but 
most   are costly  or even  prohibitively expensive. 
 
Back to 1993, Canetti and Rabin  \cite{canetti1993fast} gave a beautiful common coin construction (CR93) centering around  asynchronous verifiable secret sharing ($\AVSS$), from which a fast and optimally resilient asynchronous binary agreement ($\ABA$) can be realized. 
Here, $\AVSS$ is a two-phase protocol that allows a dealer to confidentially ``commit'' a secret   across $n$ participating parties  during a sharing phase, 
after which   a reconstructing phase can be invoke   to let the honest parties collectively  recover the earlier committed secret.
The resulting {\em common coin} is {\em reasonably fair}, as it ensures all honest parties to output either 0 or 1 with some constant probability.
%
%For the honest sender, this primitive ensures that its input   can always be successfully ``committed'' and later recovered, and more importantly, be confidentially until some honest party enters the reconstructing phase,
%despite of the adversary that fully controls $n/3$ Byzantine parties and underlying communication network.
Though this reasonably fair  common coin attains constant asynchronous rounds and can be directly plugged into many binary agreement constructions \cite{mostefaoui2015signature,canetti1993fast,crain2020two}, 
it incurs tremendous $\bigO(n^6)$ messages and $\bigO(\lambda n^8 \log n)$ bits, where $\lambda$ is the  cryptographic security parameter.
The huge complexities of CR93 are   dominated by its expensive $\AVSS$.
Since then,  many more  efficient private-setup free $\AVSS$ protocols \cite{cachin2002asynchronous,backes2011computational,alhaddad2021high,bangalore2018almost} were proposed and can directly  improve it.
For example, Cachin et al. \cite{cachin2002asynchronous} gave an $\AVSS$    to share $n$ secrets with only $\bigO(n^2)$ messages and $\bigO(\lambda n^3)$ bits, but the resulting common coin and $\ABA$ protocols (CKLS02)  still incur $\bigO(n^3)$ messages and $\bigO(\lambda n^4)$ bits, which remains expensive and exists an $\bigO(n)$ gap between the message and the communication complexities.

%Recently, Kokoris-Kogias et al. \cite{kokoris2020asynchronous} presented a new   path to reducing common coin to $\AVSS$, which can save the amortized communication cost.
%To generate a single random coin, their protocol (KMS20)   incurs $\bigO(\lambda n^4)$ bits and $\bigO(n)$ asynchronous rounds, which is seemingly worse than CKLS02,
%but essentially has better amortized   complexities. This is because it has at most can  continually  generate  perfect fair random coins after running $\bigO(n)$ rounds (at a price of constant rounds and $\bigO(\lambda n^2)$  bits), so it can be cheaper in an amortized way.
%Noticeably, KMS20 requires a strengthened  $\AVSS$ notion with high-threshold secrecy, which specifies that the adversary shall not learn the secret ``committed'' during the sharing phase, unless at least $n-2f$ honest parties tried to reconstruct the secret (where $f$ is the number of probably corrupted parties).\footnote{Different from high-threshold $\AVSS$,   classic   $\AVSS$  in CR93 \cite{canetti1993fast} only preserves secrecy before the first honest party activates reconstruction. Moreover, classic  $\AVSS$   is even weaker in another aspect, as it  might only ensure $f+1$ honest parties  to receive their corresponding shares. 
%Instead, many advanced notions \cite{kokoris2020asynchronous,cachin2002asynchronous,alhaddad2021high,backes2011computational} require strong commitment property (i.e., asynchronous complete secure sharing) that implies all honest parties to receive the corresponding shares.}

Recently, Kokoris-Kogias et al. \cite{kokoris2020asynchronous} (KMS20)  presented a new   path to reducing the common coin primitive to $\AVSS$.\footnote{KMS20 requires a strengthened  $\AVSS$   with high-threshold secrecy: the adversary cannot learn the   secret  before $n-2f$ honest parties start to reconstruct  (where $f$ is the number of   corrupted parties). In contrast,  the classic   $\AVSS$  notion \cite{canetti1993fast} only preserves secrecy, before the first honest party activates reconstruction.}
%To generate a single random coin, 
In particular, KMS20  incurs $\bigO(\lambda n^4)$ bits and $\bigO(n)$ asynchronous rounds to generate a single random coin, which   is seemingly worse than CKLS02. Nonetheless, once being bootstrapped, it  can  continually  generate      coins at a lower per-coin cost of $\bigO(\lambda n^2)$ bits and $\bigO(1)$ rounds, thus being cheaper in an amortized way.
 In another recent  breakthrough, Abraham et al. \cite{abraham2021reaching} presented an elegant asynchronous validated  Byzantine agreement ($\VBA$) protocol  (AJM+21) without private setup.
 %\footnote{Note that AJM+21 \cite{abraham2021reaching} $\VBA$ also implies asynchronous binary agreement ($\ABA$) with same complexities, because  there is a simple complexity-preserving reduction from $\ABA$ to $\VBA$ in the PKI setting, cf. \cite{cachin2001secure}.}  
 %
 It costs expected $\bigO(n^3)$ messages,  constant rounds, and  $\bigO(|m|n^2 + \lambda n^3 \log n)$ bits for $|m|$-bit input,\footnote{Through the paper, we   consider the input size $|m|$ of $\VBA$ to be at most $\lambda n$ bits, so     the $|m| n^2$ term does not dominate the communication complexity, thus ignored.    
 	For   larger input, it can be an   orthogonal problem to push the $|m| n^2$ term to $|m| n$, as discussed by many  ``extension'' protocols \cite{lu2020dumbo,patra2011communication,nayak2020improved,ganesh2021optimal} for multi-valued BA.} and only assumes the presence of a bulletin PKI that can facilitate the management of public keys.
 At the core of AJM+21 $\VBA$, it lifts reasonably fair common coin to a new  {\em random proposal election} primitive, 
 such that with a   constant probability, the honest parties can randomly decide a common value proposed by some non-corrupted party.
 %with a purpose of letting the honest parties to decide a common value proposed by some honest party with some reasonably constant probability.
 As such, a   certain $\VBA$  protocol  called {\em No-Waitin' HotStuff} (NWH)   was   tailored to cater for this special proposal  election primitive.
 However,  this election notion is too specific to be used in other existing  $\VBA$ constructions \cite{abraham2018validated,lu2020dumbo,cachin2001secure}, due to the  imperfect of necessary agreement. %, and   unclear how to be applied for asynchronous BFT agreements other than $\VBA$.
 Still, AJM+21 $\VBA$ costs   $\bigO(\lambda n^3 \log n)$ bits, and leaves room for further reducing the  communication cost asymptotically by removing the  $\log n$ factor.

\smallskip
Bearing the state-of-the-art,  it calls 
%for  more general  approaches to   asynchronous BA protocols with easily achievable setup assumptions, in  particular, 
to systematically treat the fundamental basis of more efficient (reasonably fair) common randomness protocols, such that we can make them private-setup free and   pluggable in the existing  asynchronous Byzantine agreement designs, thus overcoming the current deployment hurdles of   asynchronous protocols. That said, the following question remains open:

\noindent
\begin{center}
{\em Can we design   efficient  asynchronous      common randomness protocols with \\  fewer setup assumptions, thus   reducing the expected communication cost of \\asynchronous Byzantine   agreements (e.g., $\ABA$ and $\VBA$) to    $\bigO(\lambda n^3)$ bits?} %Will advanced high-threshold $\AVSS$  be necessary for us to realize these efficient private-setup free asynchronous BFT protocols?  
\end{center}
% .

%Asynchronous verifiable secret sharing (AVSS) was introduced by Canetti and Rabin in 1993 to instantiate asynchronous Byzantine agreement without private setup that is also fast (i.e.,   output in expected constant ``running time'') and optimally resilient (i.e.,   tolerate up to $n/3$ Byzantine parties). The primitive was originally minted to allow $n$ parties collectively ``commit'' a confidential value taken as input by a designated sender during a sharing phase, such that upon any honest party outputs in the phase, a reconstructing phase can be then executed to make all honest parties collectively recover a common value; 
%for the honest sender, the primitive ensures that its input   can always be successfully ``committed'' and later recovered, and more importantly, be confidentially until some honest party enters the reconstructing phase,
%despite of the adversary that fully controls $n/3$ Byzantine parties and underlying communication network.
%, and more importantly, the secret would be confidential until the first honest party enters the 
%The despite of Byzantine parties and asynchronous network

%After that, AVSS has become a pivotal primitive

\subsection{Our contribution} 
We give an affirmative answer to the above question. At the core of our solution, we develop a set of new techniques to design an efficient private-setup free construction for reasonably fair common coin that are   pluggable in many existing $\ABA$ protocols \cite{canetti1993fast,mostefaoui2015signature,crain2020two}; more interestingly, we formalize and construct an efficient (reasonably fair) leader election notion with perfect agreement, by lifting our common coin protocol to  be always agreed. This leader election primitive can be directly plugged in all existing $\VBA$ protocols \cite{cachin2001secure,abraham2018validated,lu2020dumbo,guo2022speeding,gelashvili2021jolteon} to remove their reliance on private setups.  %All proposed protocols can work only with a bulletin public key infrastructure (PKI), and attain the maximal $n/3$ resilience, expected $\bigO(n^3)$ message complexity, expected  $\bigO(\lambda n^3)$ communication complexity, and expected constant running time.
%by presenting a   novel path  to more efficient asynchronous binary agreement and asynchronous validated agreement without private setup. Our protocols can work with   a bulletin public key infrastructure (bulletin PKI), and attain the maximal $n/3$ resilience, expected $\bigO(n^3)$ message complexity, expected  $\bigO(\lambda n^3)$ communication complexity, and expected constant running time.

\vspace{-0.8cm}
\begin{table}[]
	\centering	
	\caption{\bf Comparison of private-setup free asynchronous BA protocols}\label{tab:comp}
	\begin{tabular}{c|cc|cc|ccc}
		\hline\rule{0pt}{8pt}
		\multirow{2}{*}{}           & \multicolumn{2}{c|}{$\ABA$/Coin}                               & \multicolumn{2}{c|}{$\VBA$/$\elect$}                               & \multirow{2}{*}{\begin{tabular}[c]{@{}c@{}}Adaptive \\ Security?\end{tabular}} & \multirow{2}{*}{\begin{tabular}[c]{@{}c@{}}Cryptographic\\ Assumption\end{tabular}} & \multirow{2}{*}{\begin{tabular}[c]{@{}c@{}}Setup\\ Assumption\end{tabular}} \\ \cline{2-5}\rule{0pt}{8pt}
		& Comm.                                 & Round                       & Comm.                                 & Round                       &                                                                               &                                                                                     &                                                                             \\ \hline\rule{0pt}{8pt}
		CKLS02 \cite{cachin2002asynchronous}  $^\S$                    & $\bigO(\lambda n^4)$                  & $\bigO(1)$                  & -                                     & -                           & Yes                                                                           & Dlog+hash                                                                           & global param $^*$                                                                \\\rule{0pt}{8pt}
		KMS20 \cite{kokoris2020asynchronous} $^\dagger$                      & $\bigO(\lambda n^4)$                  & $\bigO(n)$                  & $\bigO(\lambda n^4)$                  & $\bigO(n)$                  & No$^\star$                                                                           & RO+DDH$^\star$                                                                              & PKI  $^\sharp$                                                                        \\\rule{0pt}{8pt}
		AJM+21 \cite{abraham2021reaching} $^{\dagger\ddagger}$                     & $\bigO(\lambda n^3 \log n)$           & $\bigO(1)$                  & $\bigO(\lambda n^3 \log n)$$^\P$           & $\bigO(1)$                  & No                                                                            & RO+SXDH                                                                             & PKI                                                                         \\ \hline\rule{0pt}{8pt}
		\multirow{2}{*}{This paper} & \multirow{2}{*}{$\bigO(\lambda n^3)$} & \multirow{2}{*}{$\bigO(1)$} & \multirow{2}{*}{$\bigO(\lambda n^3)$} & \multirow{2}{*}{$\bigO(1)$} & No                                                                           & RO+SXDH                                                                             & PKI                                                             \\ \cline{6-8}\rule{0pt}{9pt}
		&                                       &                             &                                       &                             & Yes                                                                            & RO+DDH                                                                            & PKI, 1-time rnd $^{**}$                                                                         \\ \hline
	\end{tabular}
	\vspace{-0.2cm}
	\begin{scriptsize}
		{   
			\begin{itemize}
				\item[$*$] Global parameters   capture some minimal setups such as an agreed group description and   group generators. For some  schemes relying on   collision-resistant     hash    \cite{cachin2002asynchronous,kokoris2020asynchronous}, a  one-time common random string   is   needed to key  the  hash functions.
				\item[$\S$] CKLS02 \cite{cachin2002asynchronous} did not construct $\VBA$ or   leader election. We also  do not realize any complexity-preserving reductions    to it. 
				\item[$\star$] KMS20 states that it might be adaptively secure by using the pairing-based adaptively secure threshold signature \cite{libert2016born}, and this might cause it rely on SXDH assumption instead of only DDH assumption.
				\item[$\sharp$] The PKI setup in KMS20 can be removed by recent high-threshold $\AVSS$ presented in \cite{alhaddad2021high}.
				\item[$\dagger$]
				Note that KMS20  and AJM+21 did not present an explicit construction  for random   leader election ($\elect$). Nevertheless, they gave   asynchronous distributed key generation protocols ($\ADKG$) that   can   bootstrap    threshold verifiable random function and   thus can  set up $\elect$ (and also common coin) schemes via $\ADKG$. %However, this unnecessarily long path to constructing common randomness protocols is  essentially improvable   to our results.
				\item[$\ddagger$] AJM+21 only presents an explicit  $\VBA$  construction but does not construct $\ABA$. However, $\VBA$ implies $\ABA$ with same complexities, because  there is a simple complexity-preserving reduction from $\ABA$ to $\VBA$ in the PKI setting, cf. \cite{cachin2001secure}.
				\item[$\P$] The communication   of AJM+21 can be reduced to $\bigO(\lambda n^3)$ by a recent     reliable broadcast protocol \cite{dasasynchronous}, but   this   only applies to the specific AJM+21 $\VBA$ construction, while our result is generic and can be adapted to all existing $\VBA$ protocols.
				\item[$**$] 1-time rnd means a    one-time common random string can be published   after PKI registration but before protocol execution.

				%reasonably fair common coin and leader election. Nevertheless, this very recent work gave an asynchronous distributed key generation protocol with    $\bigO(\lambda n^3 \log n)$ bits and expected constant running time, which   at least bootstraps    threshold pseudorandom function and  can faithfully set up common coin and leader election schemes. However, this long path to constructing common randomness is  essentially improvable and unnecessarily complex to our results.  
			\end{itemize}
		}
	\end{scriptsize}
\end{table}
\vspace{-0.5cm}
 
In greater details, our technical contribution is three-fold:
\begin{itemize}

	\item We  implement $\bigO(\lambda n^3)$-bit and $\bigO(1)$-round {\bf common coin and   $\ABA$ with only   PKI setup} in the asynchronous network, conditioned on SXDH assumption and random oracle.
	
	The crux of our design is a novel efficient construction for the reasonably fair common coin in the bulletin PKI setting  (in  the random oracle model).  
	Different from   CR93 (that used $n^2$ $\AVSS$ instances), we use verifiable random function in  combination with more efficient   $\AVSS$ construction to reduce the number of necessary $\AVSS$ instances by an $\bigO(n)$ order. This private-setup free common coin costs only  $\bigO(\lambda n^3)$ bits and constant asynchronous rounds.
	With our common coin protocol at hand, we can implement private-setup free $\ABA$s with expected $\bigO(n^3)$ message complexity and  $\bigO(\lambda n^3)$ communication complexity with only bulletin PKI. 
	As illustrated in Table \ref{tab:comp}, it closes the $\bigO(n)$ gap between the message and the communication complexities in the earlier private-setup free $\ABA$ protocols such as CKLS02 \cite{cachin2002asynchronous}, while preserving other benefits such as the maximal $n/3$ resilience and the optimal expected constant  rounds. Even comparing with a  recent work due to Abraham et al. \cite{abraham2021reaching} that presents a more efficient $\VBA$ construction and improves $\ABA$ as a by-product,\footnote{Remark that there might exist efficient reduction from  $\ABA$ to $\VBA$  in the   public key infrastructure setting, which was discussed in \cite{cachin2001secure}. 
	%The idea is simple: every party signs and multicasts its input bit, then each one solicits a vector of $n-f$ input-signature pairs from distinct parties, and feeds the vector into $\VBA$, such that $\VBA$ would return to everyone the common vector of $n-f$ signed bits, the majority of which becomes the $\ABA$ output.
	Therefore, the recent private-setup free $\VBA$ protocol in \cite{abraham2021reaching} also improves $\ABA$.} our approach still realizes a $\log n$ factor improvement.

	%Besides the exemplary $\ABA$ protocol,   it is also pluggable in many existing $\ABA$ constructions \cite{mostefaoui2015signature,crain2020two} to attain the same asymptotic performance.
	\smallskip
	\item We also realize $\bigO(\lambda n^3)$-bit and $\bigO(1)$-round   {\bf random leader election and   $\VBA$ with only    PKI}  in the asynchronous setting, assuming SXDH assumption and random oracle. 
	
	At the core of this contribution, 
	%it is an efficient random leader election protocol with reasonable fairness and also {\em perfect agreement} in the absence of private setups.
	we   use one single $\ABA$ protocol to lift  our   common coin and clean up the possible disagreement among honest parties,
	and then obtain an efficient random leader election protocol with reasonable fairness and also {\em perfect agreement} in the absence of private setups.
	%
	%To implement more efficient $\VBA$ without private setup, we construct a  leader election primitive , assuming  the random oracle model.
	The leader election protocol costs   expected $\bigO(\lambda n^3)$ bits and expected constant asynchronous rounds, and can directly be plugged in all existing $\VBA$ protocols (i.e., multi-valued Byzantine agreement with external validity) \cite{cachin2001secure,abraham2018validated,lu2020dumbo,guo2022speeding,gelashvili2021jolteon} to replace its   counterpart relying on private setups. 
	%
	%This results in private-setup free $\VBA$s  in the asynchronous setting. 
	The resulting $\VBA$ protocols can realize   the maximal $n/3$ resilience and optimal expected constant rounds, with costing expected $\bigO(n^3)$ messages and  $\bigO(\lambda n^3)$ bits. As shown in Table \ref{tab:comp}, this construction closes the $\bigO(\log n)$ gap between the message and the communication complexities of $\VBA$ protocols. %In addition, as a by-product, our $\VBA$ construction can be directly plugged in the asynchronous distributed key generation protocol in  AJM+21 and reduces its communication cost by an $\bigO(\log n)$ order to attin $\bigO(\lambda n^3)$-bit communication complexity.

	\smallskip
	\item   Along the way, we develop a set of crucial techniques that could be of independent interests.
	
	We set forth a new  primitive called weak core-set selection ($\wcs$) to simplify the cumbersome  component of information gather in CR93 \cite{canetti1993fast} and AJM+21 \cite{abraham2021reaching}. Recall that information gather is a multi-sender extension of reliable broadcast \cite{bracha1987asynchronous}, such that each   party reliably broadcasts a value and then outputs a set of values that is a superset of some $(n-f)$-sized core-set. Selecting a core-set out of $n$   broadcasted values requires another $2n$   reliable broadcasts in \cite{canetti1993fast}.
	We conceptually weaken the primitive in a way that     $f+1$ honest parties (instead of all honest parties) are ensured to output a superset of the core-set. 
	This appropriate weakening   significantly simplifies the protocol   (i.e., replace a couple of reliable broadcasts by only two multicast rounds), and still  it is a powerful building block, from which we can implement efficient   common coin in the PKI setting.
	
	We also give an efficient $\AVSS$ construction (satisfying the classic CR93 notion \cite{canetti1993fast}) with only bulletin PKI setup (under the discrete logarithm assumption). The $\AVSS$ protocol is adaptively secure, and costs only  $\bigO(n^2)$ messages and $\bigO(\lambda n^2)$ bits when sharing a $\lambda$-bit secret. 
	%To our knowledge, this is the first private-setup free $\AVSS$ that attains $\bigO(\lambda n^2)$ communication complexity, and 
	Prior art with the same communication complexity either relies on private setup \cite{backes2013asynchronous,kate2019brief,kate2010constant} or incurs at least $\bigO(\lambda n^3)$ bits \cite{cachin2002asynchronous,backes2011computational} (except two  recent work \cite{alhaddad2021high,yurek2021hbacss}, yet  they still have an extra $\log n$ factor).
	%Moreover, 
	%Also, we set forth a weaker core-set selection p
	
\end{itemize}

%%%%%%%%%%%%%%%%%%%%%%%%%%%%%%%%%%%%%%%%%%%%%%%%%%%%%%%%%%%%%%%%%%%%%%%%%%%%%%%%%%%%%%%%

\subsection{Challenges and our techniques}

\smallskip
\noindent
{\bf Remaining efficiency hurdles}. 
To flip coins, both CR93 and AJM+21 let each party commit an unbiased   secret gathered from enough parties.  
%To have such an unbiased secret, one usually needs to  verifiably share or reliably broadcast $\bigO(\lambda n)$ bits. 
In CR93, everyone plays a role of ``delegate'' for each party to share a secret through $\AVSS$,	
then everyone picks and commits $n-f$ secrets from distinct ``delegates'',  the aggregation of which is uniformly distributed. 
%So every party  in CR93 has to verifiably share $n$ secrets.
In AJM+21, each party again plays a role of ``delegate'' to choose a random secret for everyone, which is hidden in form of aggregatable public verifiable secret sharing  ($\PVSS$).
Every party now combines $n-f$ $\PVSS$ scripts from distinct ``delegates'' to obtain an aggregated  $\PVSS$, which also  hides an unbiased secret.
Then, the aggregated $\PVSS$ script, which  has $\bigO(\lambda n)$ bits,  can be committed via a reliable broadcast (the broadcast of $\PVSS$ can be analog to   $\AVSS$, though no explicit $\AVSS$ invocation). %, thus bootstrapping the setup for threshold verifiable random function (tVRF) and enabling a proposal election primitive. 
%Clearly, in both CR93 and AJM+21, each party might commit an unbiased secret.

After enough parties commit their unbiased secrets,   it invokes a procedure to select a core-set of these secrets. That means,   all parties would choose a set of indices corresponding to some indeed committed (yet unknown) secrets, and more importantly, the honest parties' choices  shall have a large enough intersection (called core-set).  
Hence, a simple trick to flip the coin can be imagined: each party just reconstructs the secrets, then the lowest bit of the largest reconstructed secret   becomes the coin. This works because the largest secret has a constant probability to appear in the  core-set. 
Usually,    such a core-set  is obtained via several reliable  broadcasts \cite{attiya2004distributed,abraham2010fast,canetti1993fast}, whose input is a  set of $\bigO(n)$ indices.  

As such, further reducing the communication  of the  CR93 and AJM+21 frameworks  seems challenging, because every party    reliably broadcasts and/or verifiably shares at least $\bigO(\lambda n)$ bits. 
While, on the other side, although KMS20 gives a workaround to the above steps %, from which an eventually perfect common coin framework can be implemented to 
and reduces the number of secrets to share by an $\bigO(n)$ order, it on the contrary causes  %(even if it can be eventually amortized to $\bigO(\lambda n^2)$ bits after being invoked sufficiently many times to generate a large sequence of coins), 
   slow termination of  $\bigO(n)$    rounds.

\begin{figure}[h] 
	\begin{center}
		\vspace{-0.5cm}
		\captionsetup{font={footnotesize}}
		\includegraphics[width=14cm]{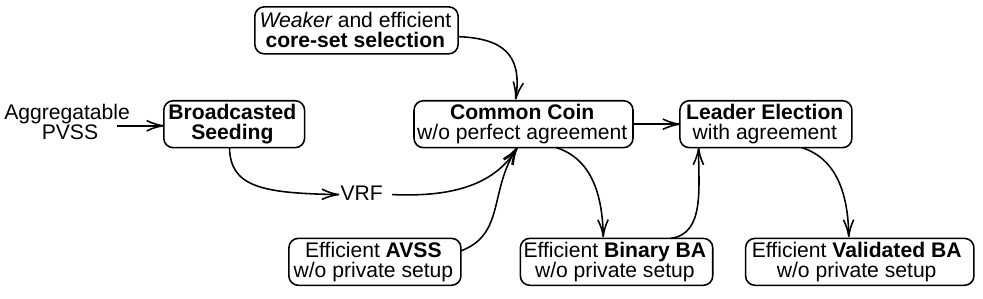}\\
		\vspace{-0.2cm}
		\caption{Our new path to constructing efficient asynchronous common randomness and   Byzantine agreement protocols in the absence of private setups.}	\label{fig:technique}
	\end{center}
	\vspace{-0.8cm}
\end{figure}

\smallskip
\noindent
{\bf Techniques to efficient common coin from PKI}.
We present {\em a collection of new techniques} to circumvent  the efficiency hurdles lying in the phases of committing secrets and selecting core-set, as illustrated in Figure \ref{fig:technique}. %In addition, they naturally close the remaining  efficiency  gap between the message and communication complexities of private-setup free asynchronous Byzantine agreement.
%with present a systematic study on efficient    ,
%and  .
%As shown in Fig. \ref{fig:reduction}, our new method can be described as follows:
In greater detail, we proceed as:

	\smallskip
	\noindent
	\underline{\smash{\em An efficient construction of private-setup free $\AVSS$}}.
	%The number of bits costed by the $n$ $\AVSS$ instances in our $\coin$ construction, we for the first time design an $\AVSS$ protocol attaining $\bigO(\lambda n^2)$ communication complexity with only the bulletin PKI assumption. 
	The underlying $\AVSS$    shall cost at most quadratic communication, if we aim at  common coin with cubic communication.  
	%Noticing that the classic $\AVSS$ notion in CR93 (e.g., without the fancy high-threshold secrecy)  is already enough in our new $\coin$ framework, 
	We realize this by lifting Pedersen's verifiable secret sharing \cite{pedersen1991non} to  asynchronous network, through exploiting the wisdom of hybrid secret sharing \cite{krawczyk1994secret}. First, the dealer just collects $n-f$ signatures from distinct parties for the same Pedersen's polynomial commitment \cite{pedersen1991non}. 
	Such that these signatures can ensure that at least $f+1$ honest parties receive the  same commitment binding a unique encryption key.  
	Then, the dealer uses the key to encrypt its actual secret, and leverages the $n-f$ solicited signatures to convince the honest parties to participate into a reliable broadcast for the ciphertext. By the design, it  avoids reliably broadcasting the large polynomial commitment to the whole network.   
	%An extra challenge might happen in the reconstruction phase as probably only $f+1$ honest parties receive the polynomial commitment and can recover the encryption key during the sharing phase; nevertheless, we let these parties to broadcast the recovered key such that all parties can eventually receive the same key to decrypt their ciphertext.
	
 \smallskip
\noindent
\underline{\smash{\em Weakening   Core-Set selection for more efficient construction}}. 
	As earlier pointed out, the selection of core-set is important to harvest some reasonable probability to make the honest parties output some common randomness. And we further tackle the problem of how to efficiently attain core-set.
	%Recall that the core-set is usually used to commit a number of unbiased secrets that are confidentially shared across the network.
	Our main observation is that if the core-set primitive is used to pick some confidentially shared  secrets   
	that come with some ``proofs'' attesting the unbiased generation, e.g., VRFs,
	we might slightly weaken the primitive. I.e.,  only $f+1$ honest parties obtain the core-set, instead of all honest parties. 
	This works because if  the largest VRF appears in this weaker core-set, $f+1$ parties can reconstruct this largest VRF, and then multicast it to let all parties accept it. 
	The weaker notion can be easily designed in the PKI setting: each party multicasts its input set, and then waits for $n-f$ signatures returned by distinct parties, 
	indicating that at least $f+1$ honest parties   have a superset. 
	%As such, we  implement the weakened core-set selection primitive, i.e., 
	This becomes an efficient workaround to avoid $\bigO(n)$ reliable broadcasts in the conventional  core-set selection \cite{canetti1993fast,abraham2021reaching,abraham2010fast,attiya2004distributed}. 
	
	%So our new $\coin$ construction starts by letting each party  use $\AVSS$ to share its own VRF evaluation-proof pair; then, until at least $n-f$ parties receive the outputs from a set of $\AVSS$ instances covering the same $n-f$ $\AVSS$es, i.e., a core set called by literature \cite{canetti1993fast,abraham2021reaching}, but our core set is slightly weaker, because probably only $f+1$ honest parties (instead of all honest parties) share the same core. 

	%
	%After fixing a quorum proof for any core set, the parties then start reconstruct all $\AVSS$ instances to reveal all hidden VRF evaluations, and use the most significant bits of the VRF evaluations to pick  the largest VRF, the lowest bit of which naturally becomes the final output bit.  

	 \smallskip
	\noindent
	\underline{\smash{\em A broadcast version ``coin'' to patch VRF}}. 	
	With efficient  $\AVSS$ and core-set selection at hand,
	it is   enticing to let every party confidentially share its own $\VRF$ evaluation with proof via $\AVSS$ (instead of gathering an  unbiased secret from many parties). 
	%This might simplify the procedure of gathering the unbiased secrets from at least $f+1$ parties. 
	Then, the weak core-set selection can ensure that at least  $f+1$ honest parties get the core of $n-f$ shared VRFs. 
	As such,   the largest VRF seemingly can appear in the core-set with a constant probability,  
	and  its   least significant bits would naturally become the common coin. 
	%The introduce of $\VRF$ might simplify the  procedure of gathering the unbiased secrets from at least $f+1$ parties. 
	%so $f+1$ honest parties can multicast the VRF to convince all honest parties to accept the same largest VRF evaluation, 
	%and hence the least significant bits of VRF evaluation can be the common coins.
	
	However, the above seemingly appealing idea  does not   work in the bulletin PKI setting,
	because if the VRFs are evaluated on deterministic seeds.
	Corrupted parties can simply exploit malicious key generation to bias its own VRF, for example: just run key generation for polynomial times, choose the most favorable VRF key pair,
	and register    the verification key at  PKI. As such, VRFs evaluated by corrupted parties gain great advantage to be the largest. Thus, the adversary at least can nearly always know the flipped coin in advance.
	
	To patch VRF with providing an unpredictable seed,
	we set forth a notion called reliable broadcasted seeding ($\seedgen$) and construct it from aggregatable  $\PVSS$ \cite{gurkan2021aggregatable}. 
	To some extent, $\seedgen$  
	can be viewed as a broadcast version   common coin, which 
	``broadcasts'' an unpredictable VRF seed and is led by the party who evaluates the VRF. 
	%Though the malicious leading party of $\seedgen$ protocol might
	Moreover, as long as an honest party gets the VRF seed from the $\seedgen$ protocol, all honest parties would   do so, and thus all parties can  verify the leader's VRF evaluation on the seed. 
	On the other side, if a corrupted $\seedgen$ leader intentionally blocks the protocol, no honest party can verify its VRF evaluation, 
	and this actually ``harm'' the corrupted leader itself, because no honest party would support to solicit it into the core-set.
	
	Putting everything together, we   design a novel $\coin$ framework,  
	Then a weak core-set can gather $n-f$  $\AVSS$es that share  VRFs (patched by $\seedgen$). 
	With 1/3 probability,  the largest VRF appears in the core and   is   also evaluated by some honest party, and its lowest bits   become  the flipped coin. 

	\smallskip
	\noindent
	{\bf Lifting   agreement for   Leader $\elect$}.   
	Although   common coin   is a   powerful tool  to enable asynchronous Binary agreement, it   faces a few barriers for implementing the interesting class of validated Byzantine agreements ($\VBA$), thus missing the key step to reduce the communication complexity of expected constant-round $\ADKG$   \cite{abraham2021reaching}. 
	The main issue stems from the fact that   most existing $\VBA$s \cite{cachin2001secure,abraham2018validated,lu2020dumbo} require a leader election  with at least one significant strengthening  relative to common coin, i.e., 
	ensuring agreement all the time.

	%at least two strengthenings relative to common coin: first, there shall be a reasonable constant probability that all parties' outputs are as if being uniformly chosen from the indices of all parties instead of from only 0 and 1; second, all parties' outputs need to be consistent, otherwise, the agreement property in \cite{cachin2001secure,abraham2018validated,lu2020dumbo} can be directly broken. 
	
	%Fortunately, the first strengthening can be straightly obtained by our $\coin$ construction. 
	Noticing that our  $\coin$ protocol only lacks agreement in some unlucky cases (when the largest VRF does not appears in the core-set), 
	we introduce a single $\ABA$ instance along with a  set of voting rules to ``detect''  the possible disagreement, thus lifting $\coin$  to attain perfect agreement.
	
	In particular, everyone reliably broadcasts the speculative largest VRF heard at the end of   $\coin$ execution, and then waits for $n-f$ such broadcasted VRFs.
	If there exists a majority VRF (out of the $n-f$ received) that is also the largest, they vote 1 to $\ABA$; otherwise, they vote 0.
	%These voting rules can ensure that
	%the honest parties either to vote 1 with seeing the same largest as well as majority VRF or to vote 0, 
	%such that 
	When $\ABA$ returns 0, it is okay for the honest parties to elect a default leader, e.g., the first party;
	while if $\ABA$ outputs 1, 
	%at least one honest party can announce its $n-f$ received VRF evaluations to convince all honest parties to choose the same VRF as output,
	%and no malicious party can make any 
	any two honest parties shall not have two distinct VRF evaluations that is both largest and majority,
	because the VRF evaluation satisfying the rules must be unique, so they can just select the leader according to the lower bits of this VRF evaluation.
	%and then it becomes natural to invoke a binary agreement   to decide if there exists such a uniquely selected VRF or not, 
	This fixes the lack of agreement in common coin without incurring extra factors in the asymptotic complexities.

%\medskip
%\noindent
%{\bf The challenges}. 

\subsection{Application scenarios}
%Efficient asynchronous BFT without the private setup of threshold cryptosystems  might bring us new hope to fulfill the ambitious plan of Miller et al. to apply asynchronous BFT in   proof-of-stake style ``permissionless'' blockchains \cite{miller2016honey} and also provide the key building block for distributed key generation   in the fully asynchronous network \cite{abraham2021reaching,kokoris2020asynchronous}.

$\ABA$ and $\VBA$ protocols are at the core of many   asynchronous fault-tolerant   systems (e.g., BFT state-machine replication and robust MPC service) deployed over the unstable wide-area network \cite{miller2016honey,guo2020dumbo,lu2019honeybadgermpc}. Our techniques can remove unpleasant setup assumptions of them, and therefore have a broad array of applications.
Here are a couple of typical examples:

\smallskip
\noindent
{\bf Fast  asynchronous DKG with cubic communication}.  Our new path to   efficient $\VBA$ protocols can   be used to replace the $\VBA$ instantiation in AJM+21  \cite{abraham2021reaching}  to improve   the same paper's asynchronous distributed key generation ($\ADKG$) protocol, thus reducing the communication cost from $\bigO(\lambda n^3 \log n)$ to $\bigO(\lambda n^3)$ with preserving all other benefits such as fast termination in expected constant   rounds and optimal   resilience, cf.   Section \ref{sec:application} for details.

\smallskip
\noindent
{\bf Asynchronous random beacon without DKG}. 
%In addition, the techniques we proposed can be flexibly applied to achieve more functions. 
Our random leader election protocol can be easily adapted to realize  randomness beacon, which can continually output a sequence of unbiased random strings in the asynchronous network.
Here an unbiased random string means that its distribution  is uniform (taken over all possible executions) despite the adversary.
Conceptually, the construction only has to sequentially run many leader election protocols, and thus preserves expected $O(\lambda  n^3)$ communication   and $O(1)$ rounds.
Different from prior asynchronous randomness beacon \cite{cachin2000random} that has to run $\ADKG$ \cite{abraham2021reaching,kokoris2020asynchronous,dasasynchronous} to bootstrap, our implementation does not go to $\ADKG$, thus being more friendly for dynamic join and leave, cf.   Section \ref{sec:application} for more details.

%We can construct an asynchronous random beacon protocol by continually running our $\elect$ with $O(\lambda  n^3)$ communication cost, cf.   Section \ref{sec:application} for details.

%!TEX root = main.tex

\section{Other Related Work}
Byzantine agreement was   introduced by Shostak, Pease and Lamport \cite{lamport1982byzantine}. Since then, it has been extensively studied in various settings with different flavors \cite{Fitzi2003EfficientPP,Fischer2005EasyIP,Dwork1988ConsensusIT,abraham2019communication,dolev1985bounds,bitcoin,david2018ouroboros,Katz2006OnEC,blum2020asynchronous}. 
In the asynchronous network, the Byzantine agreement problem has to be solved via randomized protocols, and was initially studied by Ben-Or \cite{ben1983another} and Rabin \cite{rabin1983randomized}, respectively.

Kokoris-Kogias et al. \cite{kokoris2020asynchronous} and Abraham et al. \cite{abraham2021reaching} 
also lifted their private-setup free Byzantine agreement protocols towards asynchronous distributed key generation ($\ADKG$), thus emulating the trusted dealer of threshold cryptosystems against asynchronous adversary. The rationale behind $\ADKG$ is that it might amortize the cost of asynchronous   protocols, if the   threshold cryptosystem is used for many times. 
The latter study (AJM+21) is the state-of-the-art $\ADKG$ approach with using private-setup free $\VBA$ as the core building block. 
Our new path to $\VBA$    can directly reduce  its communication cost by a $\log n$ factor.

Cohen et al. \cite{cohen2020not} constructed an asynchronous common coin   with sub-optimal resilience using VRFs, but it uses pre-determined nonce for  VRFs. So it actually takes an implicit assumption that either a trusted third-party performs honest VRF key generations on behalf of all parties or a common random string can be provided after all parties register at PKI.
%, otherwise, malicious parties can easily run  a polynomial times of VRF key generations to choose the most preferable VRF keys and break their design. 
G{\k{a}}gol et al. \cite{gkagol2019aleph} presented how to efficiently     re-configuring     asynchronous random beacon  by running over   permissioned directed acyclic graph (DAG). However, their bootstrapping might similarly need a common random string (generated after the registration of   parties). %, which is unavailable in our bulletin PKI setting. 
We insist on more  stringent setting without on-line common random strings.
In addition, our leader election  can be plugged in such  DAG-based     protocols \cite{gkagol2019aleph,keidar2021all} to make them efficient and   private-setup free.

There are many $\AVSS$ protocols \cite{backes2013asynchronous,kate2019brief} that can realize optimal communication with relying on private setups.
Prior to this paper and a concurrent work \cite{dasasynchronous}, the best known result are a couple of  recent studies \cite{alhaddad2021high,yurek2021hbacss} that incurs $\bigO(n^2)$ messages and $\bigO(\lambda n^2 \log n)$ bits for $\lambda$-bit input secret. 
%To our knowledge, we for the first time close the   $\log n$ gap between the message and communication complexities without private setup assumption.
Some $\AVSS$ protocols \cite{kate2019brief,alhaddad2021high} also focus on linear amortized  communication for sufficiently large input secret, 
but they remain to exchange quadratic messages and bits while sharing a short secret. Our $\AVSS$ can easily combine the information dispersal technique \cite{cachin2005asynchronous} to realize the same linear amortized  communication.

In the informational theoretic setting, a few interesting studies explored asynchronous leader election and byzantine agreement \cite{kapron2010fast,king2013byzantine}, but they realized sub-optimal resilience   and non-constant asynchronous rounds with a security   guarantee of only $1-1/poly(n)$. Our results and the prior art \cite{abraham2021reaching,cachin2005asynchronous} in the computational setting can attain optimal $n/3$ resilience and (expected) constant asynchronous rounds, with overwhelming probability $1-1/superploy(n, \lambda)$, conditioned on the hardness of underlying cryptographic assumptions.

Common coin  can also  circumvent the Dolev-Strong bound \cite{dolev1983authenticated} to fasten the termination of (partially) synchronous protocols \cite{dinsdale2020afgjort,abraham2019synchronous,Katz2006OnEC,feldman1997optimal}.
Technique-wise, Afgjort \cite{dinsdale2020afgjort} can be thought of our counterpart in the (partially) synchronous setting with the extra help from on-line common random strings,
it also gathers VRFs in a core-set, and then uses the least significant bits of the maximal VRF to toss coins.
However, Afgjort explicitly relies on the (partial) synchrony assumption to wait for that the largest VRF appears in the core-set.
In contrast, we make two significant enhancements to adapt our private-setup free asynchronous setting. First, unpredictable seeds are generated on the fly to patch VRFs due to the lack of on-line common random string. Second, $\AVSS$ and a special  asynchronous core-set selection protocols are designed to ensure that honest parties' VRFs are not leaked until a large enough core-set is fixed.

\medskip
\noindent
{\bf Note on results in a concurrent work \cite{dasasynchronous}}. A concurrent work from Das et al. \cite{dasasynchronous} presented the technique of asynchronous data dissemination (ADD) to improve the efficiency of reliable broadcast and relevant protocols such as $\AVSS$.
%It can be leveraged to reduce the communication cost of reliable broadcast and private-setup free $\AVSS$. 
It can reduce the communication    of the specific AJM+21 $\VBA$ protocol  to $\bigO(n^3)$ by   replacing the reliable broadcast building block.
It also has applications to reduce the communication   of AJM+21 ADKG to cubic. Though Das et al.'s reliable broadcast and some of their proposed $\AVSS$ protocols can be adaptively secure,  their applications to   $\VBA$ and ADKG are also in the random oracle model against only static corruptions. 

Different from Das et al. that focused on  improving the broadcast components, we present a set of very different techniques to simplify  the protocol structures of common coin and random leader election, which are the basis protocols of fast-terminating Byzantine agreement.  
In particular, our common coin and leader election can be directly plugged into any Byzantine agreement protocols that requires such a building block to improve the efficiency of their private-setup free variants, while ADD only explicitly helps the specific AJM+21 $\VBA$   (e.g., if without our results or future studies on asynchronous common randomness). %Our leader election is even   applicable for all existing $\VBA$ protocols.
In addition, our leader election protocol can be adapted into a reconfiguration-friendly random beacon protocol with DKG, while  Das et al.'s results can only  bootstrap random beacon protocol through DKG. 
Moreover, it might be interesting to explore the new design space provided by the combination of ADD and our techniques towards practical private-setup free asynchronous protocols.

%So \yuan{TODO: ...}

%There are also a few nice result on implementing random beacons. \yuan{...}

	%!TEX root = main.tex

\section{Models}
\label{sec:model}

{\bf Fully asynchronous system without private setup}. 
%In our system model, 
There are $n$ designated parties, each of which has a unique identity (i.e., $\node_1$ through $\node_n$)  known by everyone.
Moreover, we consider the fully-meshed asynchronous message-passing model with   Byzantine corruptions and bulletin public key infrastructure (PKI). 
%the $n$ parties are connected in  with fully-meshed p2p secure channels, 
In particular, our system and threat models can be detailed as:
\begin{itemize}
	\item {\em Bulletin PKI.}  There exists a PKI functionality that can be viewed as a bulletin board, such that each party $\node_i \in \{\node_j\}_{j \in [n]}$ can register some public keys (e.g., the verification key of digital signature) bounded to its identity via the PKI before the start of protocol. 
	Once a public key is registered, we assume all parties can receive them immediately from the PKI.
	
	\item {\em  Computing model.} 
	Following \cite{cachin2001secure,abraham2018validated} and modern cryptographic practices, we 
	let the $n$ parties and the adversary $\adv$ be probabilistic polynomial-time interactive Turing machines (ITMs).
	A party $\node_i$ is an ITM defined by the given protocol:
	it is activated upon receiving an incoming message to carry  out  some polynomial steps of computations,  update  its   states,  possibly  generate  some  outgoing  messages, and wait for the next  activation.
	%run in polynomial time and probably generate some outgoing messages.
	%to bound the adversary’s running time in the asynchronous setting, 
	%The adversary $\adv$ is a probabilistic ITM that runs in polynomial time (in the number of message bits generated by honest parties), it is the first ITM  activated  in the system (and then messages other parties to activate and run the protocol on a session identifier).
	%The above model does not rule out the probably ``infinite'' protocol execution and might give the adversary non-negligible probability to break cryptosystems, so
	Moreover, we explicitly require the bits of the messages generated by honest parties to be probabilistic uniformly bounded by a polynomial in the security parameter $\lambda$, %which was formulated as {\em efficiency} in \cite{cachin2001secure,abraham2019asymptotically} 
	which naturally rules out infinite protocol executions and thus restrict the running time of the adversary through the entire protocol. %Same to \cite{cachin2001secure} and \cite{abraham2019asymptotically},  all system parameters (e.g.,  $n$) are bounded by polynomials in  $\lambda$.

	\item {\em Up to $n/3$   Byzantine corruptions.}  %The adversary can fully control $f$   out of  $n$ parties. In addition, we restrict to static corruptions, namely, 
	The adversary can  choose up to $f$  out of $n$ parties to corrupt and fully control. %, before the course of a protocol    execution.  
	No asynchronous BA protocols can tolerate more than $f = \lfloor (n-1)/3 \rfloor$ such Byzantine corruptions, so this is the optimal resilience.
	%$Through  the paper, we stick with this optimal resilience. 
	We also consider that the adversary can   control  the corrupted parties to  generate their key materials maliciously, which captures that the compromised parties might exploit advantages while registering public keys at   PKI.
	%Therefore,   $f+1$ and $n-2f$ are interchangeable, so are $2f+1$ and $n-f$.

	\item {\em Fully asynchronous network.} 
	We assume that there exists an established     p2p channel between any two   parties.
	The  channels are considered as {\em secure}, which means the adversary cannot modify or drop the messages sent between honest parties 
	and   it is computationally infeasible for the adversary to learn any information of the messages except their lengths.
	Moreover, the adversary must be consulted to approve the delivery of     messages, namely,  
	it can arbitrarily delay and reorder   messages. % deliveries,
	Here we assume asynchronous secure channels (instead of merely  reliable asynchronous channels) for presentation simplicity, and they   can be obtained from the bulletin PKI through authenticated key exchange, and therefore are not extra assumptions.
	%are not extra assumptions as can be   obtained from the PKI assumption. 
	
	%\item {\em Computationally-bounded adversary.} 
	%We   consider  computationally bounded adversary.
	%In the asynchronous network,
	%it means the adversary  can   perform  some  probabilistic computing steps
	%bounded by   polynomials	in the number of message bits generated by honest parties.	

	\item  {\em Miscellany.} All system parameters, such as $n$, are (probably unfixed) polynomials in  the security parameter  $\lambda$ \cite{cachin2001secure,abraham2018validated,lu2020dumbo}.
	
\end{itemize}

%\smallskip
\noindent
{\bf Quantitative performance  metrics}. 
Since we are particularly interested in constructing efficient asynchronous  protocols, e.g., for generating common randomness or reaching consensus, without private setup, it becomes needed to introduce quantitative metrics to define the term ``efficiency'' in the context.
To this end, we consider the following widely adopted notions to quantify the performance of  protocols in the asynchronous network:

\begin{itemize}
	\item {\em Communication complexity} is defined as the   bits of all messages exchanged among honest parties  during  a protocol execution. Sometimes, an asynchronous protocol  might have randomized executions, so we  might consider  the upper bound of expected communication complexity (averaged over all possible executions) under the influence of adversary.
	%in the presence of any  adversary that fully controls the network and corrupts up to $n/3$ parties.
	
	\item {\em Message complexity} captures the  number of messages exchanged among honest parties in a protocol  execution. 
	Similar to communication's, we sometimes might consider the upper bound of expected message complexity.
	
	\item {\em Asynchronous  rounds}. 
	The eventual delivery of   asynchronous network   might  cause that the protocol execution is somehow independent to ``real time''.
	Nevertheless, it is  needed to characterize the running time of asynchronous protocols. 
	A standard way to do so  is: for each message the adversary  assigns a virtual round number $r$, subject to the condition that any $(r-1)$-th round messages between any two correct parties must be delivered before any $(r+1)$-th round message is sent \cite{canetti1993fast}.
	We then can measure the running time by counting such  asynchronous ``rounds''. 
	
\end{itemize}

\noindent
{\bf Note on ``private-setup free''}. More precisely, our private-setup free model   admits bulletin PKI  with some system parameters that are just group descriptions and random group generators. These  parameters are indeed one-time setup (belonging to a global system  instead of just for ours), and are usually ignored in the literature \cite{abraham2021reaching,dasasynchronous}. 
Except that, we   consider     other structured common reference strings such as that of KZG polynomial commitment \cite{kate2010constant} fall  into the category of private setups. We  might consider that  no one can  stay  on-line to provide a trusted common random string after  PKI registration, which is the subtle reason why we have to patch VRF with a broadcasted seeding protocol. % after the completeness of PKI registration.

\medskip
\noindent
{\bf Note on static/adaptive adversary}. Some of our results (e.g., our $\AVSS$ protocol) can be secure against an adaptive adversary that can corrupt up to $f$ parties while the protocol is running.
While our common coin and random leader election protocols are secure in the static model, in which the adversary is restricted to corrupt parties before the protocol starts.
However, this is only because the existing aggregatable $\PVSS$ scheme is not proven to be adaptively secure (which actually is the same reason why AJM+21 \cite{abraham2021reaching} is statically secure).
To demonstrate that, we can introduce a one-time online common random string  assumption,   thus avoid the broadcasted seeding protocol that relies on $\PVSS$, and then show that our common coin (as well as random leader election)  become adaptively secure.
Namely, we can assume a trusted one-time randomness that is announced after PKI registration but before protocol execution, and adaptive security can be realized by our protocols in the setting, as 
$\PVSS$ is no longer needed, cf. more detailed discussions in Section \ref{sec:coin}.

Moreover, the assumption of adaptively secure private channels can be easily realized by existing techniques (e.g., in the erasure model \cite{canetti1999adaptive}).

%%%%%%%%%%%%%%%%%%%%%%%%%%%%%%%%%%%%%%%%%%%%%%%%%%%%%%%%%%%%%%%%%%%%%%%%%

%!TEX root = main.tex

%\subsection{Notations and Preliminaries}
%\label{sec:pre}

\section{Preliminaries}

%\noindent
%{\bf Cryptographic hash function.}
%A   hash function $\hash: \{0,1\} ^ {*} \rightarrow \{0,1\} ^ \lambda$ is said to be collision-resistant, if no probabilistic polynomial-time $\adv$ can generate two distinct values $x_1$ and $x_2$ s.t. $\hash(x_1)=\hash(x_2)$ except with negligible probability. %Through the paper, we consider such a collision-resistance hash   $\hash$.

\noindent{\bf Reliable broadcast ($\mathsf{RBC}$)} \cite{bracha1987asynchronous} is a protocol   among a set of $n$ parties, 
in which a party called sender aims to send a value to all. It satisfies the next properties:
%(i) {\em Agreement}: the honest parties output the same value.
%(ii) {\em Totality}: if an honest party outputs, then all honest parties would output.
%(iii) {\em Validity}: if an honest sender inputs $ v $,   all honest parties would output $ v $.
\begin{itemize}
	\item {\em Agreement}. If any two honest parties output $ v $ and $v'$ respectively,   $ v = v'$.
	
	\item {\em Totality}. If an honest party outputs $ v $, then all honest parties output $ v $.
	
	\item {\em Validity}. If an honest sender inputs $ v $,   all honest parties would output $ v $.
\end{itemize}

\noindent
{\bf Digital signature}. A digital signature scheme consists of a tuple of  algorithms $(\Gen, \Sign, \Verify)$:
\begin{itemize}
	\item $\Gen(1^\lambda) \rightarrow (sk, pk)$ is a probabilistic algorithm   generating the signing and verification keys.

	\item $\Sign(sk, m) \rightarrow \sigma$ takes   a  signing key $sk$ and a message  $m$ as input to compute a signature $\sigma$.

	\item $\Verify(pk, m, \sigma) \rightarrow 0/1$   verifies whether $\sigma$ is a valid signature produced by a certain party with verification key $pk$ for the message $m$ or not.

\end{itemize}

%The signature scheme shall satisfy {\em correctness} and {\em security}: the {\em correctness} requirement is trivial, i.e., for any   signing-verification key pair, the honestly generated signature for any message shall pass the verification; for {\em security}, 

We require the digital signature scheme to be existentially unforgeable under an adaptive chosen-message attack (i.e., \textsf{EUF-CMA} secure). 
In the bulletin PKI setting, %an honest party $\node_i$  runs the $\Gen$ algorithm honestly to generate its private signing key $sk_i$ and public verification key $pk_i$, and 
every party is bounded to a unique   verification key for  signature. 
For presentation brevity, in a protocol   with an explicit   identifier $\ID$, 
we might let $\Sign_i^\ID(m)$   denote $\Sign(sk_i, \langle \ID, m \rangle)$, which means a specific party $\node_i$ signs a message $m$ with using its private key, and also  let $\Verify_i^\ID(m, \sigma)$ to denote $\Verify(pk_i, \langle \ID, m \rangle, \sigma)$, where $pk_i$ is the public key of a certain party $\node_i$. % and $\ID$ also represents the protocol's identifier.

%%%%%%%%%%%%%%%%%%%%%%%%%%%%%%%%%%%%%%%%%%%%%%%%%%%%%%%%%%%%%%%%%%%%%%%%%%%%%%%%%
\ignore{
Formally speaking, EUF-CMA can be formalized the following experiment $\Sigforge_{\Adv,\Pi}(n)$ for an Adversary $\Adv$ and parameter $n$:
\begin{itemize}
	\item 
\begin{itemize}
	\item $\Gen(1^n)$ is run to obtain keys $(pk,sk)$.
	\item Adversary $\Adv$ is given $pk$ and access to an oracle $\Sign_{sk}(\cdot)$.The adversary then outputs $(m \sigma)$. Let $\Q$ denote the set of all queries that $\Adv$ asked its oracle.
	\item$\Adv$ succeed if and only if (1)$\Vrfy_{pk}(m, \sigma)=1$ and (2)$m \notin \Q$.In this case the output of the experiment is defined to be 1.
\end{itemize}
if for all probabilistic polynomial adversary $\Adv$ there is  a negligible function $\negl$ such that:
$$
\Pr[\Sigforge_{\A,\Pi}(n)=1]\leq \negl(n).
$$
The signature scheme $\mathsf{\Pi}=(\Gen, \Sign,\Vrfy)$ is \textsf{existentially unforgeable under an adaptive chosen-message attack}.
\end{itemize}
}%%%%%%%%%%%%%%%%%%%%%%%%%%%%%%%%%%%%%%%%%%%%%%%%%%%%%%%%%%%%%%%%%%%%%%%%%%%%%%%%%

\medskip
\noindent
{\bf Verifiable random function}. A verifiable random function ($\VRF$) \cite{micali1999verifiable} is a pseudorandom function that also returns a proof to attest that the correctness of its evaluation result. It consists  of three algorithms $(\GenVRF, \Eval, \VerifyVRF)$:
\begin{itemize}
	\item $\GenVRF(1^\lambda) \rightarrow (sk, pk)$ is a probabilistic algorithm that generates a pair of private key and public verification key for verifiable random function.

	\item $\Eval(sk, x) \rightarrow (r,\pi) $ takes a secret key $sk$ and a value $x$ as input  and outputs a pseudorandom value $r$  with a proof $\pi$.

	\item $\VerifyVRF(pk, x, r,\pi) \rightarrow 0/1$ verifies  whether $r$ is correctly computed from $x$ and $sk$ using $\pi$ and the corresponding $pk$.

\end{itemize}

$\VRF$  shall satisfy {\em unpredictability}, {\em verifiability} and  {\em uniqueness}.  Here  {\em verifiability}    conventionally    means $\Pr[~ \VerifyVRF(pk, x, \Eval(sk, x)) = 1 ~|~  (sk, pk) \leftarrow \GenVRF(1^\lambda) ~] = 1$. 
{\em Unpredictability} requires that for any input $x$, it is computationally  infeasible to distinguish the value $r = \Eval(sk, x)$ from another uniformly sampled value $r'$ without access to $sk$. 
{\em Uniqueness} requires that it is computationally infeasible to find $x$, $r_1$, $r_2$, $\pi_1$, $\pi_2$ such that $r_1 \neq r_2$ but $\VerifyVRF(pk, x,r_1, \pi_1) = \VerifyVRF(pk, x,r_2, \pi_2) = 1$.

Taking the bulletin PKI for granted, everyone can be associated to a unique $\VRF$ public   key. For example, an honest party $\node_i$  runs   $\GenVRF$     to generate a unique pair of private key $sk_i$ and public key $pk_i$, and then registers its $pk_i$ via the PKI.
However, %a corrupted party might exploit malicious key generate   in order to bias its $\VRF$ distribution before registering its public key. 
  traditional VRF notion \cite{micali1999verifiable} does not    malicious key generation done by corrupted parties.
To capture such threat, we actually require a stronger {\em unpredictability} property called {\em unpredictability under malicious key generation} due to David et al. \cite{david2018ouroboros} throughout the paper, which means that even if the adversary is allowed to corrupt some parties to conduct malicious key generation,   $\VRF$ remains to perform like a random oracle. Such $\VRF$  ideal functionality  can be achieved in the random oracle (RO)  model under   CDH assumption \cite{david2018ouroboros}. 

Notation-wise, we   let $\Eval_i^\ID(x)$ be short for $\Eval(sk_i, \langle\ID, x \rangle)$, where $sk_i$ represents the private key of a party $\node_i$, and $\ID$ in our context is an explicit session identifier of a protocol instance. Similarly,  $\VerifyVRF_i^\ID(x, r,\pi)$  then denotes $\VerifyVRF(pk_i, \langle\ID, x \rangle, r,\pi)$.

\medskip
\noindent
{\bf (Aggregatable) public verifiable secret sharing}.  A $(n,t)$ non-interactive $\PVSS$ scheme can be described as a tuple of  non-interactive  algorithms as follows (with taking  $\param$ as an  implicit input):
\begin{itemize}
	%\item $\setup(1^\lambda) \leftarrow $. It returns the common parameters $\param$ that is implicitly used by all algorithms down below.
	\item $\deal(\allek, s) \rightarrow \pvss$ is an algorithm that takes a secret $s$ as input  and outputs   a script $\pvss$.
	
	\item $\vrfyscript(\allek, \pvss) \rightarrow 0/1$ is a deterministic algorithm that takes all encryption keys as input, and can verify whether a $\PVSS$ script $\pvss$ is valid in the sense that $\pvss$ commits a fixed polynomial that can be   reconstructed collectively by $n$ parties (i.e., output 1) or not (i.e., 0).

	\item $\share(\dk_i, \pvss) \rightarrow \sh_i$ is executed by the party $\node_i$,   takes a valid $\pvss$ script and $\node_i$'s decryption key $\dk_i$ as input, and outputs the secret share $\sh_i$ of the secret committed to $\pvss$.
	
	\item $\vrfyshare(j, \sh_j, \pvss) \rightarrow 0/1$ takes the $\PVSS$ script  $\pvss$ and party $\node_j$'s secret share $\sh_j$    as input, and verifies whether $\sh_j$ is the correct $j^{th}$ share of the polynomial committed to $\pvss$   or not.
	
	\item $\aggshares(\{(j, \sh_j)\}_{t}) \rightarrow a$  takes $t$ valid secret shares from distinct parties regarding an implicit $\PVSS$ script  $\pvss$,   and computes the secret $a$ committed to the   $\pvss$.
	
	\item $\vrfysecret(s, \pvss) \rightarrow 0/1$  verifies whether a secret $s$ is indeed committed to $\pvss$   or not.
	
\end{itemize}

Gurkan et al. \cite{gurkan2021aggregatable} recently proposed to lift      $\PVSS$ scheme to further enjoy aggregability, which need to slightly adapt the syntax. Here we only highlight the small adaptions to these algorithmic interfaces: 

\begin{itemize}
	
	\item $\deal(\allek,  \sk_i, s) \rightarrow \pvss$. Now the algorithm takes an extra secret signing key $\sk_i$ as input, which is needed to make the $\pvss$ script to carry an unforgeable weight tag bounded to the identity $\node_i$. 
	% Remark that sometimes, we might let $\deal$ to share a randomly chosen element in $\mathbb{Z}_q$ and thus omit the input field  of $s$  to rewrite it as $\deal(\allek,  \sk_i) \rightarrow \pvss$ for short.
	
	\item $\vrfyscript(\allek, \allvk, \pvss) \rightarrow 0/1$. It takes some verification keys $\allvk$ besides $\allek$ and $\pvss$ as input. The output   still represents whether $\pvss$ is valid or not.%; and when $b=1$, it outputs a non-empty set $\phi$ to cover the indices of the parties whose contribute in the  $\pvss$ script.

	%\item $\agg((\pvss_1, \pvsstag_1), (\pvss_2,\pvsstag_2)) \rightarrow (\pvss, \pvsstag)$. Besides aggregating two $\PVSS$ scripts $\pvss_1$ and $\pvss_2$ to return $\pvss$, the algorithm is lifted to also aggregate their tagging strings to output $\pvsstag$.

	\item $\agg(\pvss_1, \pvss_2) \rightarrow \pvss$. This is a newly introduced algorithm that takes two valid $\PVSS$ scripts $\pvss_1$ and $\pvss_2$ as input and outputs a valid $\PVSS$ script $\pvss$. 
	% We might let $\agg(\{\pvss_i,\dots, \pvss_j\})$ to be short for iteratively aggregating the $\pvss$ scripts in the set.

	\item $\weights(\pvss) \rightarrow \vec{w}$. This is another new algorithm. It takes a valid $\pvss$ script as input and outputs an $n$-sized vector $\vec{w}$,  every $j^{th}$ element in which belongs to $\mathbb{N}^0$ and represents that the $\pvss$ script indeed aggregates a certain $\pvss$ script from the party $\node_j$.

\end{itemize}

The aggregatable $\PVSS$ scheme due to Gurkan et al. \cite{gurkan2021aggregatable} satisfies   a few nice security properties such as {\em   verifiable commitment}, {\em   verifiable aggregation} and {\em secrecy}. 
%The aggregatable $\PVSS$ brings  two new properties called {\em commitment} and {\em unpredictability}. 
Informally, {\em verifiable commitment} means that any party can verify that a $\PVSS$ script $\pvss$ indeed commits a fixed secret $s$ that can later be collectively reconstructed by the participating parties; {\em secrecy} means that it  is infeasible for an adversary to compute the committed secret from the $\PVSS$ script;
{\em verifiable aggregation} means   if $\weights(\pvss)$ returns $(w_1, w_2, \cdots, w_n)$, then the secret $s$ committed to $\pvss$ indeed equals $\sum_{i=1}^{n}w_i s_i$, where $s_i$ is the secret committed to some $\PVSS$ script $\pvss_i$ that is solely generated (and signed) by the party $\node_i$.
%two valid  $\PVSS$ scripts $\pvss_1$ and $\pvss_2$ commit two secrets $s_1$ and $s_2$, respectively, then $\agg(\pvss_1, \pvss_2)$ always outputs a valid $\pvss$  committing a  secret $F^*(0) = F^*_1(0)+F^*_2(0)$;
%{\em unpredictability} states that if each honest party $\node_i$ randomly chooses the input secret $s_j$ committed to its $\PVSS$ script $\pvss_j$, then the adversary cannot compute the aggregated secret $s$ committed to a valid $\pvss$, as long as the script $\pvss$  has a non-zero weight to reflect some honest party's contribution in it.
We defer the detailed descriptions of these properties to  Appendix \ref{app:seed}.
	
	%!TEX root = main.tex

\section{Warm-up: AVSS and Weaker Core Set from   PKI}
\label{sec:avss}

%As briefly mentioned,  asynchronous verifiable secret sharing ($\AVSS$) is an important building block used by us for more efficient common coin and leader election. 
As briefly mentioned in Introduction, 
our common coin and leader election protocols require a more efficient private-setup free $\AVSS$ instantiation and an efficient construction for a weaker core-set selection ($\wcs$) notion.
When construing our common coin, $\AVSS$ is used to let everyone confidentially share an unbiased VRF evaluation, 
and $\wcs$ can be used let enough honest parties hold an intersecting core-set containing at least $n-f$ completed $\AVSS$es.
Thus, with a constant probability, the largest VRF evaluated by some honest party can appear in the core-set, and its lowest bits are ensured to be the common   coin.

%section makes a couple of needed preparations for our  common coin and leader election constructions.
This Section   focuses on the  needed preparing building blocks --- $\AVSS$ and core-set selection.
The $\AVSS$ protocol to present attains $\bigO(\lambda n^2)$   bits in the PKI setting.\footnote{Through the paper, the input secret to $\AVSS$ is assumed small, e.g., $\bigO(\lambda)$ bits.}
Then, we put forth to and construct a weak core-set selection, which can  ensure  $f+1$ honest parties (instead of all) to get some superset of a $(n-f)$-sized  core-set.

%a constant probability that
%used to construct efficient common randomness, as it can enable   honest parties to harvest a constant probability to get the common largest $\VRF$ evaluation (the least significant  bits of which are naturally the random output). %Our design realizes the conventional $\AVSS$ notion .
 
\subsection{Efficient Private-Setup Free AVSS}

%Somewhat surprisingly, the design of $\AVSS$ is very simple, and it can be readily instantiated in the asynchronous message-passing model with the bulletin PKI assumption.
Instead of varieties of strengthened $\AVSS$es, we focus on  the hereinbelow classic $\AVSS$  notion  defined by Canetti and Rabin in 1993 \cite{canetti1993fast}.

\begin{definition}	[Asynchronous Verifiable Secret Sharing \cite{canetti1993fast}]	\label{def:avss}
	An $\AVSS$ consisting of a tuple of protocols $(\sssh,\ssrec)$   can be defined as follows.
	
	\smallskip
{\textsc{Syntax}}. In each $\sssh$ instance with a session identifier $\ID$, 
a designated dealer $\node_D$ inputs a secret 
and each party outputs a string (e.g., a share of  the input secret). In the corresponding $\ssrec$ instance, the parties   input their outputs of    $\sssh$ to  collectively reconstruct the shared secret.
	
	\smallskip
{\textsc{Properties}}.   $\AVSS$      satisfies   next properties except with negligible probability:
	\begin{itemize}

		\item {\bf Totality}.  If some honest party outputs in the $\sssh$ instance associated to $\ID$, 
		then every honest party  activated to execute the $\sssh$ instance would complete the execution and output.
		
		\item {\bf Commitment}.  When an honest party outputs in the $\sssh$ instance for $\ID$, there exists a fixed value $m^*$, such that when all honest parties  are 
		activated to run the corresponding $\ssrec$ instance, all of them can reconstruct the same value $m^*$. 
		
		\item {\bf Correctness}. If the dealer is honest and inputs  secret $m$  in $\sssh$, then:
		\begin{itemize}
			\item If all honest parties are activated to run $\sssh$ on $\ID$, all honest parties would  output in the $\sssh$ instance;
			\item If any honest party reconstructs some value $m^*$ in the corresponding $\ssrec$ instance, $m^*=m$.
		\end{itemize}
		
		%\item {\bf Correctness}.  If the designed dealer is honest and input secret $m$  in $\sssh$, then the  value $m^*$ reconstructed by any honest party in the corresponding $\ssrec$ instance must be equal to $m$, for all $\ID$.
		
		\item {\bf Secrecy}. In   any $\sssh$ instance, if the dealer is honest, the adversary shall not learn any information about the input secret from its view (which includes all internal states of corrupted parties and all messages sent to the corrupted parties), before the first honest party starts the corresponding $\sssh$ instance. This can be formalized as that the adversary has negligible advantage in the Secrecy game (deferred to Appendix \ref{app:sec_game}).
	\end{itemize}
\end{definition}

\noindent
{\bf High-level rationale}. The intuitions behind our $\AVSS$ construction is simple.
Inspired by the hybrid approach of secret sharing \cite{krawczyk1994secret}, our sharing sub-protocol is split into two steps: 
(i) it first takes the advantage of PKI model to let the dealer multicast a polynomial commitment to an encryption key and then 
collect enough signatures (e.g., $n-f$) on the commitment, thus ensuing at least $f+1$ honest parties commit the same encryption key, 
and (ii) then the dealer multicasts the $n-f$ signatures solicited from the first step to convince the whole network to participate into a reliable broadcast to disseminate the ciphertext encrypting the actual secret. 

While in the reconstruction phase, probably only $f+1$ honest parties might output the polynomial commitment to the decryption/encryption key, though all honest parties have the ciphertext.
%and we cannot make them to multicast the commitment because this incurs cubic communication. 
This seemingly causes some parties fail to reconstruct the secret (because they cannot decrypt).
Nevertheless,  the $f+1$ honest parties holding the correct commitment can   recover the decryption key and then multicast it to all parties.
So all honest parties can count whether the same key is  from $f+1$ distinct parties, and finally use it to decrypt the secret.

\begin{algorithm}[ht]
		%\captionsetup{font={footnotesize}}
	\begin{footnotesize}
		
		\caption{$\sssh$ protocol with identifier $\ID$ and dealer $\node_D$}
		\label{alg:sh}
		\begin{algorithmic}[1]
			\vspace*{0.1 cm}
			\Statex {\color{blue}  /* Protocol for {the dealer $\node_D$} */ }
			\vspace*{0.1 cm}
			\Upon {receiving input secret $m \in \mathbb{Z}_q$}
			\State choose two random polynomials $A(x)$ and $B(x)$   from $\mathbb{Z}_q[x]$ of degree at most $f$ 
			\State let $a_j$ to be the $j^{th}$ coefficient of $A(x)$ and $b_j$ to be that of $B(x)$ for {$j \in [0,f]$},
			\State let $key \leftarrow a_{0}=A(0)$, i.e., $A(0)$ is also called $key$
			\State compute $c_j\leftarrow g_1^{a_j}g_2^{b_j}$ for each $j \in [0,f]$, and let $\C \leftarrow \{c_j\}_{j\in[0, f]}$ 
			
			\State	\textbf{send} $\KEYSHARE(\ID, \C, A(j), B(j))$ to $\node_j$ for {each $j \in [n]$}
			
			\EndUpon
			\vspace*{0.1 cm}
			\Upon {receiving $\STORED(\ID, \sigma_j)$ from $\node_j$ s.t. $\Verify^\ID_j( \C, \sigma_j)=1$}
			\State $\Proof \leftarrow \Proof \cup \{(j, \sigma_j)\}$
			\IF {$\vert \Proof \vert = n-f$}
			\State  $cipher \leftarrow key \oplus m$
			and \textbf{multicast} ${\CIPHER(\ID, \Proof,  \C,cipher)}$ to all parties
			\EndIf      
			\EndUpon
			
			\vspace*{0.15 cm}
			\Statex {\color{blue}  /* Protocol for {each party $\node_i$} */ }
			\vspace*{0.1 cm}
			
			\State
			$\sh_A \leftarrow \bot$, $\sh_B \leftarrow \bot$, $\cmt \leftarrow \bot$%,  $flag \leftarrow 0$
			\vspace*{0.1 cm}
			\Upon {receiving $\KEYSHARE(\ID, \C', A'(i), B'(i))$ from $\node_D$ for the first time}
			\State parse $\C'$ as $\{c'_0, c'_1,\dots, c'_f\}$
			\IF	{$g_1^{A'(i)}g_2^{B'(i)} =\prod_{k=0}^{f} {c'_k}^{i^k}$}
			\State record $A'(i)$, $B'(i)$ and $\C'$,  $\sigma \leftarrow \Sign^\ID_i{( \C')}$,
			\textbf{send} $\STORED(\ID, \sigma)$ to $\node_D$
			\EndIf
			\EndUpon
			
			\vspace*{0.1 cm}
			\Upon {receiving $\CIPHER(\ID, \Proof, C, cipher)$ from $\node_D$ for the first time}
			\Wait { for  a valid $\KEYSHARE$ message s.t. $A'(i)$, $B'(i)$ and $\C'$ are recorded} \EndWait
			\IF {$ C'  = C$   and $\Proof$ has $n-f$ valid signatures for $C$ from distinct parties} %$|\{j ~|~ (j, \cdot) \in \Proof\} |= n-f$ $\wedge$ $\forall (j,\sigma_j) \in \Pi$, $\Verify^\ID_j(h, \sigma_j)=true$}
			\State $\sh_A \leftarrow A'(i)$,
			$\sh_B \leftarrow B'(i)$
			and $\cmt \leftarrow C'$
			\State \textbf{multicast} ${\ECHO(\ID, cipher)}$ to all parties
			\EndIf

			\EndUpon

			\vspace*{0.1 cm}
			\Upon {receiving $2f+1$ $\ECHO(\ID, cipher)$ from distinct parties}
			\State \textbf{multicast} $\READY(\ID,cipher)$ to all parties if $\READY$ not sent yet
			\EndUpon
			
			\vspace*{0.1 cm}
			\Upon {receiving $f+1$ $\READY(\ID, cipher)$ from distinct parties}
			\State \textbf{multicast} $\READY(\ID, cipher)$ to all parties if $\READY$ not sent yet
			\EndUpon
			
			\vspace*{0.1 cm}
			\Upon {receiving $2f+1$ $\READY(\ID, c)$ from distinct parties}
			\State \textbf{output} $(cipher, \sh_A, \sh_B, \cmt)$
			\EndUpon
		\end{algorithmic}

	\end{footnotesize}
	
\end{algorithm}

\medskip
\noindent
{\bf Constructing $\AVSS$ without private setups}.
%Here  we  explain our new $\AVSS$ protocol in details.
 The rationale behind our $\AVSS$ construction is straightforward.
The sharing phase splits the hybrid approach of secret sharing \cite{krawczyk1994secret} into two steps: 
(i) it first takes the advantage of PKI model to let the dealer
collect enough signatures (e.g., $n-f$) on a polynomial commitment \cite{pedersen1991non} to an encryption key, thus ensuing at least $f+1$ honest parties commit the same encryption key, 
and (ii) then the dealer multicasts the $n-f$ signatures solicited from the first step to convince the whole network to participate into a reliable broadcast to disseminate the ciphertext encrypting the actual secret. 

In particular, the sharing protocol $\sssh$    proceeds in the following steps: 
%
%The sharing protocol $\sssh$    proceeds in the following steps: 
\begin{enumerate}%[start=0]
	
	%\item {\em Public group parameters}. 
	
	\item {\em Key sharing} (Line 1-6, 11-15). In this phase, the dealer distributes the key shares to all parties using Pedersen's VSS scheme \cite{pedersen1991non}. The dealer   randomly constructs two polynomials $A(x)$ and $B(x)$ of degree at most $f$. Let $A(0)=key$. Then, the dealer computes a commitment $\C=\{c_j\}$ to the polynomial $A(x)$ with using  $B(x)$ for hiding, where each element $c_j=g_1^{a_j}g_2^{b_j}$, and $a_j$ and $b_j$ represent the $j^{th}$ coefficients of $A(x)$ and $B(x)$, respectively. The dealer   sends a $\KEYSHARE$ message to each party $\node_j$ containing the commitment $\C$ as well as $A(j)$ and  $B(j)$.
	
	Once a party $\node_i$ receives $\KEYSHARE$ message from the dealer, 
	it checks that $\C$ indeed commits $A(i)$ with using $B(i)$ for blinding, and 
	returns a signature for $\C$ to the dealer via a $\STORED$ message.
	
	\item {\em Cipher broadcast} (Line 7-10, 16-26). 
	%In this phase, the dealer      encrypts its actual input by the key  shared in the earlier phase, i.e., $A(0)$, and then broadcasts the resulting ciphertext. 
	After receiving $n-f$ valid $\STORED$ messages from distinct parties, the dealer sends a $\CIPHER$ message to all parties containing a ciphertext $cipher$ encrypting its input $m$, the   commitment $\C$, and a quorum proof $\Pi$ containing $n-f$ valid signatures for $\C$. 
	
	The remaining process of the phase is similar to a Bracha's reliable broadcast for $cipher$ \cite{bracha1987asynchronous}, except that 
	a party   would not ``echo'' $cipher$  if not yet receiving valid $\Pi$ for $\C$.
	%has to ``validate'' $h$ due to $\Pi$  carried by the $\CIPHER$ message to  further proceed. 
	%After receiving the $\CIPHER(\ID, \Pi, h,c)$ from the dealer, parties who have recorded $\C'$, $A'(i)$ and $B'(i)$ in the last phase check if $\hash(C')=h$ and the signatures in $\Pi$ are valid for $h$. Then they $\ECHO$ $h$ and $c$ to all parties. Upon receiving $2f+1$ $\ECHO$ messages with the same $h$ and $c$, send a $\READY$ message to all parties if it is not sent yet. After receiving $f+1$ $\READY$ messages, send $\READY$ if it is not sent yet. When receiving $2f+1$ $\READY$ messages, every party $\node_i$ outputs $(h, c, A'(i), B'(i), \C')$. 
	At the end of the phase, each party can output $(cipher, A(i), B(i), \C)$, where
	$A(i)$, $B(i)$ and $\C$ can be $\bot$.
	
\end{enumerate}

\begin{algorithm}[ht]
	%\captionsetup{font={footnotesize}}
	\begin{footnotesize}
		
		\caption{$\ssrec$ protocol with identifier $\ID$, {\color{blue}   for {each party $\node_i$}  }}
		\label{alg:rec}
		\begin{algorithmic}[1]

			\Statex   \textbf{Initialization:} $\Phi \leftarrow \emptyset$
			
			\vspace*{0.1 cm}
			\Upon {being activated with input $(cipher, \sh_A, \sh_B, \cmt)$}
			\IF {$\cmt \neq \bot$, $\sh_A \neq \bot$ and $\sh_B \neq \bot$}
			\State \textbf{multicast} $\KREC(\sh_A, \sh_B)$ to all parties
			\EndIf
			\EndUpon
			
			\vspace*{0.1 cm}
			\Upon {receiving ${\KREC(\sh_{A,j},\sh_{B,j})}$ from $\node_j$ for the first time}
			\IF {$\cmt \neq \bot$  }
			\State parse $\cmt$ as $\{c_0, c_1,\dots, c_f\}$
			\IF	{$g_1^{\sh_{A,j}}g_2^{\sh_{B,j}} =\prod_{k=0}^{f} {c_k}^{j^k}$}
			%\State record $F'(i)$ and $\C'$, $\sigma \leftarrow \Sign^\ID_i{(\hash(\C))}$,
			%\textbf{send} $\STORED(\ID, \sigma)$ to $\node_D$
			\State $\Phi \leftarrow \Phi \cup {(j, \sh_{A,j})}$
			\IF {$|\Phi| = f+1$}
			\State interpolate polynomial $A(x)$ from $\Phi$ and compute $key \leftarrow A(0)$
			\State $\textbf{multicast}$ $\KEY(\ID, key)$ to all parties
			\EndIf
			\EndIf
			\EndIf
			\EndUpon
			
			\vspace*{0.1 cm}
			\Upon {receiving $f+1$ $\KEY$ messages  containing the same $key$}
			\State  $m \leftarrow key \oplus cipher$
			and $\textbf{output}$ $m$
			\EndUpon
			
		\end{algorithmic}

	\end{footnotesize}
	
\end{algorithm}

%\medskip
Then, the  $\ssrec$  phase can be activated to  proceed in the next two   phases: 
\begin{enumerate}
	\item {\em Key recovery} (Line 1-10). For each party $\node_i$ that activates $\ssrec$ with taking the output of $\sssh$ as input, it sends $\KREC$ message containing  $A(i)$ and $B(i)$, in case these variables are not $\bot$. 
	%For each party  $\node_i$ whose $\C'$ is not $\bot$, it can check each $h_j$ and  shares receiving from $\node_j$'s $\RECREQ$ message. After receiving $f+1$ valid $\RECREQ$ messages, interpolate $A(x)$ and compute the $key$. Since some parties did not record $\C'$ in $\sssh$, they can not check the shares receiving from other parties. So parties who got the $key$ send a $\KEY$ message to broadcast it.
	At least $f+1$ honest parties   have already received the commitment $\C$, so they can eventually solicit $f+1$ valid shares of the polynomial committed to $\C$ through  $\KREC$ messages, and then interpolate the shares to reconstruct $A(x)$ and compute $key=A(0)$.
	\item {\em Key amplification}. After a party obtains decryption $key$, it multicasts $key$  via a  $\KEY$ message. So all honest parties can receive $f+1$ $\KEY$ messages containing the same $key$,  and then compute and output $m=key \oplus cipher$.  
\end{enumerate}

For the sake of completeness, we also present   formal pseudocode descriptions for $\sssh$ and $\ssrec$ in Alg. \ref{alg:sh} and Alg. \ref{alg:rec}, respectively. 
%In practice, $g_1$ and $g_2$ can be derived from hash-to-group, so even the protocol designer cannot compute $\log_{g_1}g_2$.

\ignore{
The intuition of proving our simple $\AVSS$ protocol is clear: the totality is mainly because our construction employs Bracha broadcast's message pattern for distributing the ciphertext encrypting the input secret and the hash of polynomial commitment; the secrecy follows the information theoretic argument due to Pedersen \cite{pedersen1991non} about his verifiable secret sharing; the commitment is ensured by Pedersen commitment   and the unforgeability of   signatures. %; while the correctness property is a direct result of the correctness of all underlying cryptographic primitives.
For the sake of completeness, we defer the proof to Appendix \ref{app:avss_proof}
}

\medskip
\noindent
{\bf Security analysis of $\AVSS$}. The intuition of proving our simple $\AVSS$ protocol is clear: the totality is mainly because our construction employs Bracha broadcast's message pattern for distributing the ciphertext encrypting the input secret and the hash of polynomial commitment; the secrecy follows the information theoretic argument due to Pedersen \cite{pedersen1991non} about his verifiable secret sharing; the commitment is ensured by Pedersen commitment   and the unforgeability of   signatures. More formally, the security of our $\AVSS$ protocol can be proved as follows.

\begin{lemma}\label{lemma:sh-consistent}
	If any two honest parties $\node_i$ and $\node_j$ output $(cipher, \cdot,\cdot,\cdot)$ and $(cipher', \cdot,\cdot,\cdot)$ in $\sssh[\ID]$, respectively, then $cipher=cipher'$ 
	%and $\cmt=\cmt'$(or $\cmt = \bot$ or $\cmt'=\bot$ ) 
	except  with negligible probability.
\end{lemma}
\begin{proof}
	Suppose that $cipher \neq cipher'$, $\node_i$ receives $2f+1$ $\READY$ messages containing $cipher$, the senders of which include at least $f+1$ honest parties; in the same way, $\node_j$ must have received at least $f+1$ $\READY$ messages containing $cipher'$ from honest parties; so it induces that at least one honest party sent two different messages, which is impossible. So there is a contradiction if $cipher \neq cipher'$, implying $cipher=cipher'$.
	
	%Honest party will record a $\cmt$  only if it receives a valid $\Pi$ containing  $n-f$ valid signatures for $\cmt$. Suppose that $\cmt \neq \cmt'$, $\node_i$ receives a $\Pi$, which means at least $f+1$ honest parties have signed for $\cmt$; in the same way, $\node_j$ receives a $\Pi'$ and at least $f+1$ honest parties have signed for $\cmt'$. So at least one honest parties signed for both $\cmt$ and $\cmt'$, which is impossible. Thus, $\cmt = \cmt'$. 
\end{proof}
\begin{lemma}\label{lemma:totality}
	If some honest party outputs in the $\sssh$ instance associated to $\ID$, 
	then every honest party  activated to execute the $\sssh$ instance would complete the execution and output.
\end{lemma}
\begin{proof}
	Assume that an honest party outputs in the $\sssh$, it must have received $2f+1$ $\READY$ messages. At least $f+1$ of the messages are sent from honest parties. Therefore, all parties will eventually receive $f+1$ $\READY$ messages from these $f+1$ honest parties and send a $\READY$ message as well. Then, all honest parties will eventually receive $2f+1$ valid $\READY$ messages and output. So the totality property always holds.
\end{proof}

\begin{lemma}\label{lemma:commit}
	When some honest party outputs in the $\sssh$ instance for $\ID$, there exists a value $m^*$ that is fixed associated to  $\ID$.
\end{lemma}
\begin{proof}
	Firstly, we prove that if any two honest parties record $\cmt$ and $\cmt'$ respectively, then $\cmt=\cmt'$. Note that honest party will record a $\cmt$  only if it receives a valid $\Pi$ containing  $n-f$ valid signatures for $\cmt$. Suppose that $\cmt \neq \cmt'$, $\node_i$ receives a $\Pi$, which means at least $f+1$ honest parties have signed for $\cmt$; in the same way, $\node_j$ receives a $\Pi'$ and at least $f+1$ honest parties have signed for $\cmt'$. So according to the unforgeability of digital signatures,  at least one honest party signed for both $\cmt$ and $\cmt'$, which is impossible. Thus, $\cmt = \cmt'$. 
	Moreover, $\C$ is computationally binding conditioned on DLog assumption, so all honest parties agree on the same polynomial $A^*(x)$ committed to $\C$, which fixes a unique $key^*$.
	From Lemma \ref{lemma:sh-consistent} and the totality of $\AVSS$, when some honest party outputs in the $\sssh$ instance for $\ID$, all honest parties receive the same cipher $cipher$. So there exists a unique $m^*=cipher \oplus key^*$, which can be fixed once some honest party outputs in $\sssh$.
\end{proof}

\begin{lemma}\label{lemma:commit-rec}
	When all honest parties activate  $\ssrec$ on $\ID$, each of them can reconstruct the same value $m^*$
\end{lemma}
\begin{proof} 
	Any honest party outputs in the $\sssh$ subprotocol must receive $2f+1$ $\READY$ messages from distinct parties, at least $f+1$ of which are from honest parties. Thus, at least one honest party has received $2f+1$ $\ECHO$ messages from distinct parties. This ensures that at least $f+1$ honest parties get the same commitment $\C$ and a valid quorum proof $\Pi$. 
	Since the signatures in $\Pi$ are unforgeable, at least $f+1$ honest parties did store valid shares of $A^*(x)$ and $B^*(x)$ along with the corresponding commitment $\C$ except with negligible probability. 
	So after all honest parties start $\ssrec$, there are at least $f+1$ honest parties would broadcast $\KREC$ messages with valid shares of $A^*(x)$ and $B^*(x)$. These messages can be received by all parties and can be verified by at least $f+1$ honest parties who record $\C$. With overwhelming probability, at least $f+1$ parties can interpolate $A^*(x)$ to compute $A^*(0)$  as $key$ and broadcast it, and all parties can receive at least $f+1$ same $key^*$ and then output the same $m^* = cipher \oplus key^*$ as they obtain the same ciphertext $cipher$ from $\sssh$. 
\end{proof}

\begin{lemma}\label{lemma:correct-1}
	If the dealer is honest and all honest parties are activated to run $\sssh$ on $\ID$, all honest parties would  output in the $\sssh$ instance.
\end{lemma}
\begin{proof}
	If the dealer is honest and all honest parties are activated, it is clear that (i) all honest parties can eventually wait the shares of $A(x)$ and $B(x)$ as well as the same commitment $\C$ so all honest parties will sign for $\C$. Thus, the honest dealer must can collect at least $n-f$ valid digital signature for $\C$ from distinct parties to form valid $\Pi$ and (ii) all honest parties can eventually broadcast the same $\ECHO$ messages and the same $\READY$ messages after receiving the shares of $A(x)$ and $B(x)$ as well as the same $\C$, thus finally outputting in the $\sssh$ instance.
\end{proof}

\begin{lemma}\label{lemma:correct-2}
	If the dealer is honest and inputs secret $m$, the value $m^*$ reconstructed by any honest party in the corresponding $\ssrec$ instance must be equal to $m$, for all $\ID$.
\end{lemma}
\begin{proof}
	From Lemma \ref{lemma:commit-rec}, we have proved that all honest parties will reconstruct the $m^*$ which is fixed when some honest party completes the $\sssh$. So all we need is to prove that the fixed $m^*$ is equal to the $m$ that the honest dealer inputs in the $\sssh$. 
	It is easy to see that (i) any honest party must output a ciphertext $cipher$ same to the ciphertext computed by the honest sender and (ii) due to the correctness and binding of commitment scheme honest parties must receive the same $C$ to $A(x)$ and $B(x)$, where $A(x)$ and $B(x)$ are chosen by the honest deader. So $m^* = cipher \oplus A(0) = m$.
\end{proof}

\begin{lemma}\label{lemma:secure}
	In   any $\sssh$ instance, if the dealer is honest, the adversary shall not learn any information about the key shared by the dealer from its view.
\end{lemma}
\begin{proof}
	The adversary's $view$ in an $\sssh$ execution with an honest dealer would include the commitment $\C$, the ciphertext $cipher$, some signatures for $\C$, the secret shares received by up to $f$ corrupted parties as well as all public keys and corrupted parties secret keys. The signatures leak nothing related to the shared $key$ (even if the adversary can fully break digital signature to learn the private signing keys). 
	Thus, following the information-theoretic argument in \cite{pedersen1991non}, since the commitment $\C$ is perfectly hiding and $f$ shares of Shamir's secret sharing scheme also leaks nothing about the key, the adversary can  learn nothing about $key$.
\end{proof}

%\medskip
%\noindent
%The security properties of our $\AVSS$ construction can be formalized as follows:

\begin{theorem}\label{thm:avss}
	The algorithms shown in   Alg. \ref{alg:sh} and Alg. \ref{alg:rec}  realize $\AVSS$ as defined in Definition \ref{def:avss}, in the asynchronous message-passing model with $n/3$ adaptive byzantine corruption and bulletin PKI assumption without private setups, conditioned on the hardness of Discrete Log problem and   \textsf{EUF-CMA}  security of digital signature.
\end{theorem}
\begin{proof}
	Here prove that Alg. \ref{alg:sh} and Alg. \ref{alg:rec} satisfy   $\AVSS$'s properties one by one:
	\begin{itemize}
		\item {\em Totality}. Totality can be proved immediately from Lemma \ref{lemma:totality}.
		
		\item {\em Commitment}.	Commitment can be proved from Lemma \ref{lemma:commit} and Lemma \ref{lemma:commit-rec}.
		
		\item {\em Correctness}. Correctness can be proved from Lemma \ref{lemma:correct-1} and Lemma \ref{lemma:correct-2}
		
		\item {\em Secrecy}. From Lemma \ref{lemma:secure}, if the dealer is honest, the adversary can learn nothing about $key$. So the adversary cannot distinguish the distribution of $cipher = key \oplus m_b$ and a uniform  distribution (otherwise, the adversary can be invoked to break Lemma \ref{lemma:secure}). Therefore, the adversary's advantage in the Secrecy game ${\bf Adv_{sec}}$ is negligible.
	\end{itemize}
\end{proof}
%\smallskip
%{\textsc{Remark on the Common Reference String}}.
%The $\sssh$ and $\ssrec$ protocols are parameterized by two generators $g_1$ and $g_2$ of a cyclic group $\mathbb{G}_q$ s.t. the underlying discrete logarithm problem is  intractable and no one knows the discrete logarithm of $g_1$ to basis $g_2$, or that of $g_2$ to basis $g_1$.
%Essentially, assuming the random oracle (RO) model, one can further simplify the above common reference strings to just a common group generator $g_1$ in addition to the group description.
%The key idea is   to replace ``$c \leftarrow key \oplus m$'' (the line 10 of  Alg. \ref{alg:sh}) and ``$m \leftarrow key \oplus c$'' (the line 13 of  Alg. \ref{alg:rec}) by ``$c \leftarrow \hash(key) \oplus m$'' and ``$m \leftarrow \hash(key) \oplus c$'', respectively, where $\hash$ is a random oracle. After the replacement, the polynomial $B(x)$ is no longer needed to hide $key$, and thus all operations related to $B(x)$ can be removed. The proof for  secrecy can be similarly argued, because gaining non-negligible advantage in the $\AVSS$'s secrecy-game   indicates that the adversary can  query RO with the right $key$ in a polynomial number of queries, causing the break of Dlog assumption. 
{\textsc{Remark on adaptive security}}. Our $\AVSS$ protocol is adaptively secure.
This is because our $\sssh$ sub-protocol is similar to that of \cite{cachin2002asynchronous}, which use Pedersen's polynomial commitment in combination of Shamir's secret sharing to consistently distribute the secret shares  (of an encryption key) to the participating parties, which avoids the shortage of using static cryptographic primitive such as non-interactive $\PVSS$.
While the major difference between \cite{cachin2002asynchronous} and our $\sssh$ is that we concatenate $n-f$ \textsf{EUF-CMA} secure digital signatures to form a quorum proof attesting that enough honest parties have received the consistent secret shares. Nevertheless, this doesn't sacrifice adaptive security, because the quorum proof ensures that there must exist $f+1$ honest parties that  can never be corrupted and also received the consistent secret shares, and hence these $f+1$ forever honest parties can help all other honest parties to recover the same encryption key during the $\ssrec$ protocol.

\medskip
\noindent
{\bf Complexities of $\AVSS$}. The complexities of the $\AVSS$  protocol shown   in   Alg. \ref{alg:sh} and Alg. \ref{alg:rec}  can be easily seen:
Both $\sssh$ and $\ssrec$ protocols cost at most a constant number of asynchronous rounds to terminate. Each round at most exchanges $n^2$ messages, indicating $O(n^2)$ message complexity.
Moreover, there are $O(n)$ messages having $\bigO(\lambda n)$ bits and $O(n^2)$ messages having $\bigO(\lambda)$ bits, thus the communication complexity of the protocol is of overall $O(\lambda n^2)$ bits. Recall that $\lambda$ captures the size of cryptographic objects.

\subsection{Weak Core Set Selection}
%To flip a coin in the asynchronous setting, a few earlier studies \cite{canetti1993fast,cachin2002asynchronous,abraham2021reaching} rely on    core set selection  . 
Core-set selection is a critical component while flipping a coin in the asynchronous network \cite{canetti1993fast,abraham2021reaching}.
It allows each party to output a set of indices representing some completed reliable broadcasts \cite{abraham2021reaching} or completed $\AVSS$es \cite{canetti1993fast}, and more importantly, the  the intersection of all honest parties' outputs corresponds to an always large enough core-set (e.g., $n-f$).
%Essentially, each element in the output set can correspond to an unbiased random number that was earlier committed but not yet revealed (for example, in CR93, this might mean a secret aggregating $f+1$ secrets shared by distinct parties via $\AVSS$). As such, there would be a constant probability to make the largest random number appearing in the core set, and thus a common coin can be derived after reconstructing the committed random numbers.
%This notion allows all honest parties to share a core of selecting which $\AVSS$es.
%
%\yuan{pick up here}

Instead of this widely known approach, here we introduce a weakened core set selection primitive in which probably only $f+1$ honest parties can receive the core.
In the presence of PKI, it can be constructed very efficiently, and remains to be an expressive notion while flipping a coin.
The idea is to use it select a core-set of $\AVSS$es that hide some $\VRF$s.
%a weaker core-set corresponding to some $\VRF$s   that are evaluated by distinct parties and then confidentially shared via .
With a constant probability, the largest $\VRF$ can luckily appear  in the   core-set, so at least $f+1$ honest parties can reconstruct this largest  $\VRF$ and multicast it to the whole network, thus still ensuring all honest parties to get the largest $\VRF$.
%We will elaborate more details on using the weak core set for flipping a coin in later sections, and 

Let us focus on this weakened notion   and its efficient construction.
%
%which is a subset of all honest parties output set. Using this property of common core, with a constant probability, all honest parties can choose the same value which is not only the largest but also proposed by honset parties.
%In this section, we claim that the common core can be weaken to that only $f+1$ honest parties are able to receive the core set, which will be used to construct a more efficiency coin. 
%
More formally, we can define the weak core-set notion as follows. %executed among $n$ parties with up to $f$ static byzantine corruptions. The syntax and properties are as follows:
\begin{definition}[Weak Core-Set Selection] \label{def:wcs}
	%An instance of the $(n,f, k,\alpha)$-common coin protocol is among .
	A protocol    among $n$ parties with up to $f$ Byzantine corruptions realizes a weak core-set selection, if  has syntax and properties as follows.
	
	\smallskip
	{\textsc{Syntax}}.
	For each protocol instance with session identifier $\ID$,  
	every party  $\node_i$ inputs a set of indices $\Si_i$ s.t. $|\Si_i|\geq n-f$. Note that each index  belongs to $[n]$, and each honest party's input set $\Si_i$ can   monotone increase over the protocol execution. Then, every honest party $\node_i$ outputs a set of indices $\hat{\Si}_i$. 

	\smallskip
	{\textsc{Properties}}. 
	It satisfies the next properties except with negligible probability:
	\begin{itemize}
		\item {\bf Termination}. If any index  in any honest party's input set can  eventually appear in all honest parties' input sets,  then every honest party   would output. 
		\item {\bf $(f+1)$-Supporting Core-Set}. Once the first honest party outputs, there exists a core-set $\Si^*$   consisting of at least $n-f$ distinct indices,
		and $\Si^*$ must be the intersection of at least $f+1$ honest parties' output sets.
		% Moreover, for at least $k$ honest parties, if some of which can output a set $\hat{\Si}$, then $\Si^* \subseteq \hat{\Si}$.
		\item {\bf Validity}. Any index in the honest parties' outputs can be found in some honest party's input set.
	\end{itemize}
\end{definition}

Intuitively, the above definition captures our purpose that after each party conducts reliable broadcast or verifiable secret sharing, 
each party can invoke the primitive to output a set of indices  representing which reliable broadcasts or $\AVSS$es are indeed completed. 
More importantly, the output sets of at least $f+1$ honest parties share a $(n-f)$-sized  intersection, 
representing that all these honest parties have output in these reliable broadcasts or $\AVSS$es.

\begin{algorithm}[ht]
	%\captionsetup{font={scriptsize}}
	\begin{footnotesize}
		
		\caption{$\wcs$ protocol with identifier $\ID$, {\color{blue} for {each party $\node_i$}}}
		\label{alg:wcs}
		\begin{algorithmic}[1]
			\State $\widetilde{\Si} \leftarrow \bot$, $\hat{\Si} \leftarrow \bot$
			\vspace*{0.15 cm}
			\Upon {receiving the input set $\Si$ and $|\Si| \geq n-f$}
			\State $\widetilde{\Si} \leftarrow \Si$
			and \textbf{multicast} $\LOCK(\ID, \widetilde{\Si})$ to all parties
			\EndUpon
			
			\vspace*{0.15 cm}
			\Upon {receiving $\LOCK(\ID, \widetilde{\Si}_j)$ from $\node_j$ for the first time}
			\If {$|\widetilde{\Si}_j| \geq n-f$}
			\Wait { for $\widetilde{\Si}_j \subseteq \Si$}
			\State $\sigma^{j}_i \leftarrow \Sign^\ID_i(\widetilde{\Si}_j)$
			and \textbf{send} $\CONFIRM(\ID, \sigma^{j}_i )$ to $\node_j$
			\EndWait
			\EndIf
			\EndUpon
			
			\vspace*{0.15 cm}
			\Upon {receiving $\CONFIRM(\ID, \sigma^{i}_j )$ from $\node_j$ s.t. $\Verify^\ID_j(\widetilde{\Si}, \sigma^i_{j})=1$}
			\State $\Sigma \leftarrow \Sigma \cup \{j, \sigma^i_{j}\}$ 
			\If {$|\Sigma| = n-f$} 
			\State \textbf{multicast} $\COMMIT(\ID, \Sigma, \widetilde{\Si})$ to all parties
			\EndIf
			\EndUpon
			
			\vspace*{0.15 cm}
			\Upon {receiving  $\COMMIT(\ID, \Sigma_j, \widetilde{\Si}_j)$ message from $\node_j$ for the first time}
			%\If {$|\{k ~|~ (k, \sigma^j_{k}) \in \Sigma_j\} |=n-f$ $\wedge$ $\forall (k, \sigma^j_{k})$ $\in$ $\Sigma_j$, $\Verify^\ID_k(h_j, \sigma^j_{k})=1$}
			\If {$\Sigma_j$ contains $n-f$ valid signatures  for $\widetilde{\Si}_j$ from distinct parties}
			
			\Comment{I.e., check   $|\{k ~|~ (k, \sigma^j_{k}) \in \Sigma_j\} |=n-f$ $\wedge$ $\forall (k, \sigma^j_{k})$ $\in$ $\Sigma_j$, $\Verify^\ID_k(\widetilde{\Si}_j, \sigma^j_{k})=1$}
			\State $\widehat{\Si} \leftarrow \Si$ and  \textbf{output} $ \widehat{\Si}$
			\EndIf
			\EndUpon
			
		\end{algorithmic}

	\end{footnotesize}
	
\end{algorithm}

\medskip
\noindent
{\bf Constructing $\wcs$ without private setup}. Here we present a concise construction of weak core set    $\wcs$  (formally shown in Alg. \ref{alg:wcs}):
\begin{enumerate}
	\item Once an honest party $\node_i$ receives an input local set $\Si$ which contains $n-f$ values, it takes a ``snapshot'' $\widetilde{\Si}$ of $\Si$ and
	multicasts $\widetilde{\Si}$ to all parties. Note that $\node_i$'s local $\Si$ can increase monotonically after the multicast, as new indices might be added to $\Si$.
	Then, if receiving some $\widetilde{\Si}_j$ sent from some party $\node_j$, 
	the party $\node_i$ checks $|\widetilde{\Si}_j|=n-f$, and   waits for that its local $\Si$ eventually becomes a superset of $\widetilde{\Si_j}$, after which, it returns a signature for $\widetilde{\Si_j}$ to $\node_j$.
	
	\item Eventually, $\node_i$ might collect  $n-f$ distinct signatures for its multicasted ``snapshot'' $\widetilde{\Si}$, which corresponds to a  quorum proof $\Sigma$  for  $\widetilde{\Si}$. Finally, $\node_i$ multicasts  $\Sigma$ and $\widetilde{\Si}$ to all parties.
	After receiving  a valid  quorum proof $\Sigma_j$ for  $\widetilde{\Si_j}$   from some party $\node_j$, 
	the party $\node_i$ can immediately output its current local set $\Si$ (without halt).
\end{enumerate}

\medskip
\noindent
{\bf Security analysis of $\wcs$}. The security intuition of our  $\wcs$ construction  is that when the first honest party outputs, it must receive a valid quorum proof attesting that 
at least $f+1$ honest parties have signed the same set consisting of $(n-f)$  indices. Thus, these $f+1$ honest parties must have a $(n-f)$-sized intersection in their outputs. More formally, we can prove the following security theorem.

\begin{theorem}\label{thm:wcs} 
The algorithm shown in   Alg. \ref{alg:wcs} realizes  $\wcs$ against $n/3$  adaptive byzantine corruptions in the asynchronous message-passing model, conditioned on that the underlying digital signature scheme is \textsf{EUF-CMA}  secure.
\end{theorem}

\begin{proof}
	We prove that Alg.  \ref{alg:wcs} realizes the  properties of $\wcs$ in  Def. \ref{def:wcs} one by one:
	\begin{itemize}
		\item {\em Termination}. If any value $v$ in some honest party's input set will eventually be included into all honest parties' input sets, any honest $\node_i$'s $\widetilde{\Si_i}$ will be included in all honest parties' local sets. So any honest $\node_i$ can collect a set $\Sigma_i$ containing at least $n-f$ signatures for its $\widetilde{\Si_i}$ and multicast it to all parties via a $\COMMIT$ message. 
		For any honest party $\node_i$, once it  receives a valid $\Sigma_j$ for the first time, it will fix a $\hat{\Si}$ and output. 
		\item {\em (f+1)-Supporting Core-Set}. When the first honest party outputs from the protocol, it has received a valid $\Sigma_j$ with $n-f$ signatures for $\Si_j$ and at least $f+1$ signatures are signed by honest parties. Note that an honest party $\node_i$ will sign for some $\Si_j$ only if $\Si_j \geq n-f$ and $\Si_j \subseteq \Si_i$. Thus, trivially from the unforgeability of digital signatures, with all but negligible probability, once the first honest party receives a valid $\Sigma_j$ for $\Si_j$ and outputs, there exists a core set $\Si^*=\Si_j$ which is subset of at least $f+1$  forever honest parties' local $\Si$. After that, if some of the $f+1$ forever honest parties can output a set $\hat{\Si}$, then $\Si^* \subseteq \hat{\Si}$.  
		\item {\em Validity}. The validity of $\wcs$ is trivial since each honest party outputs its local set $\Si$ which is the input of itself.
		
	\end{itemize}
\end{proof}
{\textsc{Remark on adaptive security}}. We concatenate $n-f$ \textsf{EUF-CMA}  secure signatures from distinct parties to attest the existence of a core set,
and therefore, whenever any so-far honest party outputs, 
there must exist $f+1$  honest parties that have signed the core set and can never be corrupted by the adaptive adversary, 
and at least these $f+1$ forever honest parties would share a $(n-f)$-sized intersection in their output sets. This facts prevent the adaptive adversary from corrupting parties posteriorly to make the core-set is received by less than $f+1$ honest parties.
%as their local is a superset

\medskip
\noindent
{\bf Complexities of $\wcs$}. The complexities of the $\wcs$ protocol can be easily seen. All parties can terminate after three asynchronous rounds (i.e., $\LOCK$, $\CONFIRM$ and $\COMMIT$). The over message complexity is   $O(n^2)$, because each honest party sends at most $3n$  messages. Each message contains at most $O(\lambda n)$ bits, so the overall communication cost is $O(\lambda n^3)$.
%\begin{itemize}
%	\item {\em Message}. Each party sends $n$ $\LOCK$, $\CONFIRM$ and $\COMMIT$ messages, So the overall message complexity of $\wcs$ is $O(n^2)$.
%	\item {\em Communication}. %The message size of one $\AVSS$ is $O(\lambda)$.The message size of one $\AVSS$ is $O(\lambda)$. 
%	The message size of $\LOCK$ $O(n)$, the message size of $\CONFIRM$ is $O(\lambda)$ and the message size of $\COMMIT$ is $O(\lambda n)$. 
%	Thus, the overall communication of $\wcs$ is $O(\lambda n^3)$.
	%$O(n^2 \cdot \lambda n + n^2 \cdot \lambda + n^3 \cdot \lambda)=O(\lambda n^3) $
%	\item {\em Running time}. The running time of $\wcs$ is constant.
	%5 rounds are taken to finish the $\sssh$, then another 4 rounds are taken to send $\LOCK$, $\CONFIRM$, $\COMMIT$ and $\RECREQ$ messages, at most 2 rounds are taken to finish the $\ssrec$. Finally one round is taken to send  $\CANDIDATE$ messages.
%\end{itemize}
	
	%!TEX root = main.tex
 
\section{Backbone: Reasonably Fair Common Coin from   PKI}
\label{sec:coin}

This section presents a novel way to private-setup free $\ABA$.
At the core of the design, it is a new reasonably fair common coin ($\coin$) which can be instantiated by $\AVSS$, $\wcs$ along with using VRFs in the bulletin PKI setting. 
The $\coin$ protocol attains constant running time, $\bigO(n^3)$ messages and $\bigO(\lambda n^3)$ bits.
Thus, many existing $\ABA$ protocols \cite{crain2020two,mostefaoui2015signature} can  directly adopt it for reducing private setup,
and preserve other benefits such as expected constant rounds and optimal resilience, with incurring expected cubic communicated bits.

\subsection{Common coin without private setup}
The backbone of our results is an efficient private-setup free common coin ($\coin$) protocol that  costs only   $\bigO(\lambda n^3)$ communicated bits and terminate in constant rounds. Formally, we consider the common coin notion defined as follows: 
%design such an efficient protocol among $n$ parties that can tolerate up to $f$ static corruptions, thus  realizing the properties of common coin defined as follows, in the asynchronous message-passing model with the PKI assumption:

\begin{definition}[$(n,f, f+k,\alpha)$-Common Coin] \label{def:coin}
	%An instance of the $(n,f, k,\alpha)$-common coin protocol is among .
	A protocol realizes $(n,f,f+k, \beta)$-$\coin$, if it is executed among $n$ parties with up to $f$ static byzantine corruptions  and has syntax and properties   as follows.
	
	\smallskip
	{\textsc{Syntax}}.
	For all executions of each protocol instance with   session identifier $\ID$, 
	every party takes the system's public knowledge (i.e., $\lambda$ and all public keys) and its own private keys  as input, and outputs a single bit. 

	\smallskip
	{\textsc{Properties}}. 
	It  satisfies the next properties except with negligible probability:
	\begin{itemize}
		\item {\bf Termination}. If all honest parties are activated on  $\ID$, every honest party will output a bit for $\ID$. 
		%\item {\bf Correctness}. All honest parties will output the same value $\sigma$. Moreover, before $k$ honest parties ($1 \le k \le f+1$) are activated on  $\ID$, each value $\sigma \in \{0,1\}$ has least $\beta$ probability to be output.
		%\item {\bf Validity}. For all $\ID$, the probabilities of that all honest parties   output 0 for  $\ID$ and that all honest parties output 1 for  $\ID$ are both at least $\alpha/2$.
		\item {\bf Reasonably fair bit-tossing}. Prior to that $k$ honest parties ($1 \le k \le f+1$) are activated on $\ID$,
		the adversary $\adv$ cannot fully guess the output. More precisely,
		consider the   predication game:
		$\adv$ guesses a bit $b^*$ before $k$ honest parties activated on $\ID$,
		if $b^*$ equals to some honest party's output for $\ID$, we say that $\adv$ wins;
		we require $\Pr[\adv \textnormal{ wins}]\leq 1-\alpha/2$. 
		
		Here  $\alpha$ represents the lower-bound probability that all honest parties would output the same bit that is as if uniformly distributed over $\{0,1\}$,
		%. Note that $\alpha$   represents the lower-bound probability that a common bit is output as if uniformly distributed over $\{0,1\}$,
		while $1-\alpha$ captures the possibility that the adversary might predicate/bias the output (which also capture the case that the  honest parties output differently). 
	\end{itemize}
\end{definition}

%\smallskip
%{\textsc{Remarks}}. 
Intuitively, with at least $\alpha$ probability taken over all possible $\coin$  executions, the adversary cannot predict the output bit better than guessing. A $\coin$ protocol is said to be perfect, if $\alpha=1$. Nevertheless, %many (binary) Byzantine agreement protocols ($\ABA$) do not necessarily need perfect $\coin$. % to be efficient as well as optimal resilient. 
many $\ABA$ constructions \cite{canetti1993fast,mostefaoui2015signature} actually  do not necessarily need perfect $\coin$, and can terminated in expected constant rounds with optimal $n/3$ resilience, as long as using a $(n,n/3,n/3+k,\alpha)$-$\coin$ scheme, where $\alpha < 1$ is a certain constant. This is mainly because such constructions can tolerate the probable disagreement of such imperfect coins, and then   repeat by iterations to explore the $\alpha$-probability good case for terminating. 

\medskip
\noindent{\bf  Tackling the seed of  VRF}. As briefly mentioned in Introduction, our $\coin$ construction   relies on VRF to let each party evaluate an unbiased random output, thus reducing the number of needed $\AVSS$es.
Before elaborating our $\coin$ construction, it is worth mentioning an issue of this cryptographic primitive in the PKI setting.
Different from many studies that implicitly assume the private key of VRF is generated by a trusted third-party \cite{cohen2020not},
we aim to opt out of such private trusted setup for VRFs. 
So the VRF key generation is conducted by each participating party itself.
That means, if the  $\coin$  protocol also uses some deterministic seeds for VRF evaluations, a compromised party can register at PKI with some maliciously chosen VRF keys, and probably can bias the distribution of its VRF during the protocol execution. 
%Imagine that: if the VRF input (a.k.a. seed) used in the $\coin$ protocol are predictable by the adversary, it becomes feasible for a corrupted party to repeat its VRF's key generation for   several times to choose a private-public  key pair that is more  favorable (e.g., to bias the most significant bits of its VRF evaluation).

{\em Trusted nonce from ``genesis''}. In many settings, this might not be an issue, since there could be a trusted nonce generated after all parties have registered their VRF keys (e.g., the same rationale behind the   ``genesis block'' in some Proof-of-Stake blockchains \cite{david2018ouroboros}), which can be naturally used as the VRF seed. Such a functionality  was earlier formalized by David,  Gazi,  Kiayias and Russell  in  \cite{david2018ouroboros} as an initialization functionality to output an unpredictable VRF seed. %In practice, this   can be somewhat implemented, e.g., by specifying in the protocol to use the hash of a future Bitcoin block.

{\em Generating VRF seed on the fly}. Nevertheless, we might still expect less setup assumptions to get rid of the trusted ``genesis block''.   
%(e.g., we cannot tolerate the long duration to wait for generating it), 
and therefore have to handle the issue by generating unpredictable VRF seeds on the fly during the course of the coin protocol.
To this end, we put forth a new notion called {\em reliable broadcasted  seeding} ($\seedgen$) as amendment to patch VRF with unpredictable nonce (or called seed interchangeably). 
Intuitively, the notion can be understood as a   ``broadcast'' version of common coins, and might not  terminate if encountering malicious leader. 
$\seedgen$ is a   protocol   with two successive   committing   and   revealing phases, and can be defined as follows.  

\begin{definition}[Reliable Broadcasted Seeding]\label{def:seeding} $\seedgen$ is a protocol with
	two successive committing and revealing phases, and can be defined as follows.
	%An instance of the $(n,f, k,\alpha)$-common coin protocol is among .
	%$\seedgen$  is  defined as follows.
	
	\smallskip
	{\textsc{Syntax}}.
	For each      protocol instance with an identifier $\ID$,   
	it has a designated party called leader $\node_L$ and is executed among $n$ parties with up to $f$ byzantine corruptions.
	Each party   takes as input the system's public knowledge  and its   private keys, and then sequentially executes the committing phase and the revealing phase,
	at the end of which it outputs a $\lambda$-bit string $seed$. 

	\smallskip
	{\textsc{Properties}}. 
	It  satisfies the next properties   with all but negligible probability:
	\begin{itemize}
		%\item {\bf Termination}. Conditioned on that all honest parties are activated on  $\ID$, every honest party would output a bit. 
		
		\item {\bf Totality}.  If some honest party outputs in the $\seedgen$ instance associated to $\ID$, 
		then every honest party  activated to execute the $\seedgen$ instance would complete the execution and output.
		
		%\item {\bf Agreement}. For any two honest parties $\node_i$ and $\node_j$ that output $seed_i$ and $seed_j$ for $\ID$, respectively, then $seed_i = seed_j$. 
		
		\item {\bf Correctness}. For all $\ID$, if the leader $\node_L$ is honest and all honest parties are activated on $\ID$, all honest parties would  output for $\ID$.
		
		\item {\bf Committing}. Upon  any honest party completes the protocol's committing phase
		and starts to run the protocol's revealing phase on  session $\ID$, there exists a fixed  value $seed$, 
		such that if any honest party outputs for $\ID$, then it outputs value $seed$ .
		
		\item {\bf Unpredictability}. Prior to that $k$ honest parties ($1 \le k \le f+1$) are activated to run the protocol's revealing phase on  session $\ID$, 
		the adversary $\adv$ cannot predicate the output $seed$. 
		Namely, $\adv$ guesses a value $seed^*$ before $k$ honest parties are activated on $\ID$, 
		%if $seed^*$   coincides with some honest party's output $seed$ for $\ID$, $\adv$ succeeds in predicating the seed,
		then the probability that $seed^*=seed$ shall be negligible, where $seed$ is the output of some honest party for $\ID$.
		%the advantage of the adversary $\adv$ in the following Predication Game is at most $1-\alpha/2$.
	\end{itemize}
\end{definition}

\smallskip
{\textsc{Remark on committing and unpredictability}}. Combining the committing and unpredictability properties would ensure that no one can predicate the output, before the seed to output is already   fixed, which is critical to guarantee that VRFs evaluated on the output $seed$ cannot be biased by manipulating the seed generation. Intuitively, if the adversary still can bias its own VRF's output, it must can query the VRF oracle (which performs as random oracle \cite{david2018ouroboros}) with the right   $seed$ in a number of polynomial queries (before $seed$ is committed), which would raise contradiction to break the unpredictability property.

\smallskip
{\textsc{Remark on totality}}. The totality property ensures that no honest party would receive some output $seed$ solely. 
%I.e., the honest parties either eventually receive the common $seed$ (to allow them verify VRFs evaluated on $seed$),
%or receive nothing from the $\seedgen$ protocol, even if the leader $\node_L$ is malicious. 
I.e., if an honest party gets $seed$, we can assert that all honest parties will do so.
%on the other side, once an honest party knows at least $f+1$ parties $seed$

\begin{lemma}\label{lemma:seed}
	In the asynchronous message-passing model with bulletin PKI assumption,
	there exists a $\seedgen$ protocol   among $n$ parties that can tolerate up to $f < n/3$ static byzantine corruptions, terminate in constant asynchronous rounds, and cost  $\bigO(n^2)$ messages and $\bigO(\lambda n^2)$ bits, assuming \textsf{EUF-CMA} secure  digital  signature and SXDH assumption.
\end{lemma}
The recent elegant result of aggregatable public verifiable secret sharing ($\PVSS$) due to Gurkan et al. \cite{gurkan2021aggregatable} lift a $\PVSS$  scheme to enjoy aggregability. 
%Recall that $\PVSS$ is a secret sharing scheme in which any third party can verify the dealer's behavior. \cite{gurkan2021aggregatable} proposes an algorithm $\agg(\pvss_1, \pvss_2)$ to combine any two valid $\pvss$ script to a new one. Moreover, it adds $\weights(\pvss)$, which takes a valid $\pvss$ script as input and outputs an $n$-sized vector $\vec{w}$ to represent which parties'  $\pvss$ scripts  are in the aggregated script.
%We defer a formal definition of  aggregatable $\PVSS$ in  Appendix \ref{app:seed}.
Employing this aggregatable $\PVSS$, we can construct an exemplary $\seedgen$ protocol. Intuitively, it is simple to let each party send an aggregatable $\PVSS$ script carrying a random secret to the leader,  so the leading party can aggregate them to produce an aggregated  $\PVSS$ script committing an unpredictable nonce contributed by enough parties (e.g., $2f+1$).
%, thus making at least an honest party's random secret to be counted. 
Then, before recovering the unpredictable secret hidden behind the aggregated $\PVSS$ script, 
the leader must send it to at least $2f+1$  parties to collect enough digital signatures to form a ``certificate'' to prove that the nonce is fixed and committed to the $\PVSS$ script. Only after seeing such proof, each party would decrypt its corresponding share from the committed  $\PVSS$ script, thus ensuring the unpredictability and commitment properties.
We defer the proof for Lemma \ref{lemma:seed} along with the  exemplary  construction (cf. Alg. \ref{alg:seed}) in Appendix \ref{app:seed}.

\begin{figure}
	\begin{center}
		\vspace{-0.5cm}
		\captionsetup{font={footnotesize}}
		\includegraphics[width=14cm]{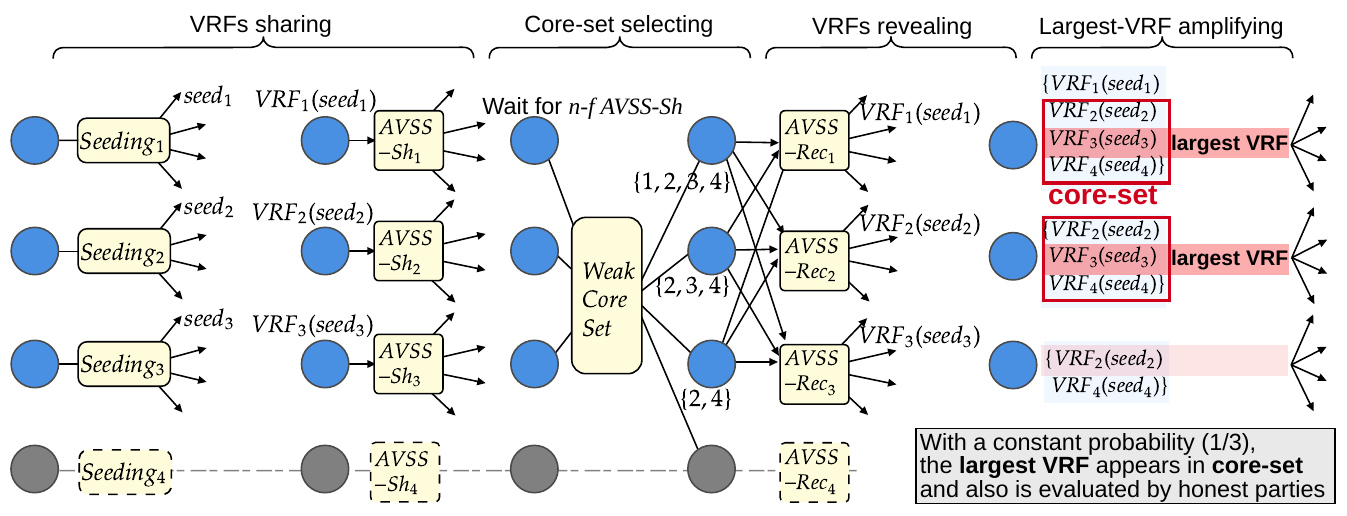}\\
		\vspace{-0.2cm}
		\caption{Overview of our $\coin$ protocol.  VRFs (patched by $\seedgen$) are shared via $\AVSS$es. Then, $f+1$ honest parties  get a   core-set   representing $n-f$ completed $\AVSS$es. After that, the VRFs are revealed, and each party multicasts the largest VRF seen by it.}	\label{fig:coin}
	\end{center}
	\vspace{-0.4cm}
\end{figure}

\medskip
\noindent
{\bf Constructing private-setup free $\coin$}. With $\seedgen$ at hand,  we are ready to present our $\coin$ construction (formally shown in Alg. \ref{alg:coin}), which has four main steps:  
\begin{enumerate}
	\item {\em VRFs sharing} (Line 1-8). Each party activates a $\seedgen$ process as leader and participants in all other $\seedgen$ processes  to get the seeds.
	A party $\node_i$ activates its own $\sssh$ instance as dealer to share its VRF evaluation-proof after obtaining its own VRF seed $seed_i$;
	once obtaining $seed_j$ besides $seed_i$,
	$\node_i$ also activates the corresponding $\sssh$ instance as a participant.
	%Then, each party invokes an $\Eval$ using its $seed_i$ to output a VRF evaluation $r$ and a proof $\pi$. Each party activates an $\sssh$ process as the dealer using its $(r, \pi)$ as input and activates all other $\sssh$ as a participant.
	
	\begin{algorithm}[h!]
		%\captionsetup{font={scriptsize}}
		\begin{footnotesize}
			
			\caption{$\coin$ protocol with identifier $\ID$, {\color{blue} for {each party $\node_i$}}}
			\label{alg:coin}
			\begin{algorithmic}[1]

				\State $\Si \leftarrow \emptyset$, $\Sigma \leftarrow \emptyset$,$\R \leftarrow \emptyset$,  $C \leftarrow \emptyset$,  $X \leftarrow 0$
				\State $\widehat{\Si} \leftarrow \bot$, $seed_j \leftarrow \bot$ for each $j$ in $[n]$
				%\If {$\ID$ is unpredictable by all parties} 
				%	\State $seed_j \leftarrow \ID$ for each $j$ in $[n]$ 
				%\Else
				%	\State $seed_j \leftarrow \seedgen[\langle  \ID, j \rangle]$ for each $j$ in $[n]$ without blocking for output
				%\EndIf
				
				%\State Suppose $\eta$ is a trusted nonce from ``genesis'' (provided by an initialization functionality upon all parties register at PKI), $seed_j \leftarrow \eta $ for each $j$ $\in$ $[n]$; if no such $\eta$, activate $\seedgen[\langle  \ID, j \rangle]$ for each $j$ $\in$ $[n]$, with being the leader in $\seedgen[\langle  \ID, i \rangle]$
				\State activate $\seedgen[\langle  \ID, j \rangle]$ for each $j \in [n]$ with being the leader in $\seedgen[\langle  \ID, i \rangle]$
				%
				%\Comment{If no such $\eta$, this line shall be replaced by  ``activate $\seedgen[\langle  \ID, j \rangle]$ for each $j$ in $[n]$ with being the leader in $\seedgen[\langle  \ID, i \rangle]$''}
				
				\Statex\Comment{If a trusted nonce $\eta$ is provided by ``genesis'' (generated by an initialization functionality after all %parties 
					have registered at PKI, cf. \cite{david2018ouroboros}), this line can be replaced by  ``$seed_j \leftarrow \eta $, for each $j$ $\in$ $[n]$''}
				
				\vspace*{0.15 cm}
				\Upon {$seed_j \leftarrow \seedgen[\langle  \ID, j \rangle]$}
				\IF {$j = i$}
				\State $(r, \pi) \leftarrow \Eval^\ID_i(seed_j)$, and activate $\sssh[\langle  \ID, j \rangle]$ as dealer with taking $(r, \pi)$ as input 
				\Else
				\State activate $\sssh[\langle  \ID, j \rangle]$ as a non-dealer participant
				\EndIf
				
				\EndUpon

				\vspace*{0.15 cm}
				\Upon {receiving output from $\sssh[\langle  \ID, j \rangle]$}
				\State $\Si \leftarrow \Si \cup \{ j \}$
				\IF {$|\Si| = n-f$}
				\State activate $\wcs$ taking $\Si$ as input
				\EndIf
				\EndUpon

				\vspace*{0.15 cm}
				\Wait { for $\wcs$ outputs $ \widehat{\Si}$}
				
				\State \textbf{send} $\RECREQ(\ID, k)$ to all parties for every $k \in \widehat{\Si}$
				\Wait { for that for every $k \in \widehat{\Si}$, $\ssrec[\langle  \ID, k \rangle]$ outputs $(r_k,\pi_k)$} {\bf do}
				\FOR {each $k \in \widehat{\Si}$}
				\If {$\VerifyVRF^\ID_k(seed_k, r_k,\pi_k) = 1$}
				\State $R \leftarrow R \cup (k, r_k, \pi_k)$
				\EndIf
				\ENDFOR
				\If {$R\neq \emptyset$} {$\ell \leftarrow \argmax_k \{r_k ~|~ (k, r_k, \pi_k) \in R \}$}  \Comment{Index of the largest VRF in $R$} \Else { $\ell \leftarrow \bot$, $r_{\ell} \leftarrow \bot$, $\pi_{\ell} \leftarrow \bot$} \EndIf
				\State \textbf{send} $\CANDIDATE(\ID, \ell, r_{\ell}, \pi_{\ell})$ to all parties
				\EndWait
				\EndWait
				
				\vspace*{0.15 cm}
				\Upon {receiving $\RECREQ(\ID, k)$ from any party for the first time}
				\Wait { for $\widehat{\Si} \ne \bot$ and that $\sssh[\langle  \ID, k \rangle]$ outputs $ss_k$} {\bf do}		\Comment{If $\widehat{\Si}$ becomes $\emptyset$, it is no longer $\bot$}	
				\State activate $\ssrec[\langle  \ID, k \rangle]$ with taking $ss_k$ as input 
				\EndWait
				\EndUpon

				\vspace*{0.15 cm}		
				\Upon {receiving  $\CANDIDATE(\ID, \ell', r_{\ell'}, \pi_{\ell'})$    from $\node_j$ for the first time}
				%\State {{\bf wait} for $seed_{l_j} \ne \bot$}
				\If {${\ell'} = \bot$} {$X \leftarrow X + 1$}
				\Else { {\bf if} $\VerifyVRF^\ID_{\ell'}(seed_{\ell'}, r_{\ell'}, \pi_{\ell'}) = 1$ {\bf then}}
				\Comment{Verifying VRF implicitly waits for $seed_{{\ell}'} \ne \bot$}
				\State $C \leftarrow C \cup (j, \ell', r_{\ell'}, \pi_{\ell'})$
				
				\EndIf

				\IF {$|C| + X =n-f$}
				\State $\hat{\ell} \leftarrow \argmax_{\hat{\ell}}  \{r_{\hat{\ell}}  | (j, \hat{\ell} , r_{\hat{\ell}} , \pi_{\hat{\ell} }) \in C \}$ \Comment{Index of the largest VRF in $C$}
				\State \textbf{output} the lowest bit of $r_{\hat{\ell}}$
				\EndIf
				
				\EndUpon
				
			\end{algorithmic}

		\end{footnotesize}
		
	\end{algorithm}
	
	\item {\em Core-set selecting} (Line 9-12). 
	Each party $\node_i$ records a local set $\Si$ of indices representing the completed $\sssh$ instances.
	Once the local $\Si$ of $\node_i$ contains $n-f$ indices, it activate $\wcs$ taking $\Si$ as input.
	
	%In this phase, each party stores the node $\ID$ in set $\Si$ after receiving its $\sssh$ output. When the size of set $\Si$ is $n-f$, fix a $\widetilde{\Si}$ and broadcast it using a $\LOCK$ message. If $\node_i$ receives a $\widetilde{\Si_j}$ which is is subset or equal to its $\Si$, which means that $\node_i$ has finished these $\sssh$ processes in $\widetilde{\Si_j}$, sign for $\widetilde{\Si_j}$ and send the signature back using a $\CONFIRM$ message. After receiving $n-f$ valid signatures of $\widetilde{\Si}$, $\node_i$ can commit $H(\widetilde{\Si})$ with a proof $\Sigma$ containing $n-f$ signatures.
	\item {\em VRFs revealing} (Line 13-24). 
	Once    $\wcs$ outputs $\hat{\Si}$, an honest party $\node_i$   starts to reconstruct $\AVSS$ associated to the indices in $\hat{\Si}$. 
	After reconstructions, it might get some valid VRF evaluation-proof, and then multicasts the VRF with maximum evaluation to all parties. 
	% out of the VRF evaluation-proof pairs,and multicasts this  to all parties via $\CANDIDATE$ messages. 
	%can realize that there exists a core set which is included in at least $f+1$ honest parties' set $\Si$.
	%Fix a $\widehat{\Si}$ and send $\RECREQ$ messages to activates $\ssrec[\langle  \ID, k \rangle]$ processes for every k in $\widehat{\Si}$. Wait for all of these $\ssrec$ to output a pair $(r_k, \pi_k)$, Select the maximum value $r_l$ and send a $\CANDIDATE$ message to broadcast it with the corresponding node $\ID$ and proof $\pi$. Note that an honest party will only activate a $\ssrec$ process after fixing its $\widehat{\Si}$.
	\item {\em Largest-VRF amplifying} (Line 25-31). After receiving $n-f$ $\CANDIDATE$ messages encapsulating valid VRF evaluation-proof pairs,  $\node_i$  selects the index  $\hat{l}$  of the largest evaluation and outputs the lowest bit of $r_{\hat{l}}$.
\end{enumerate}

 %We present the following lemmas to describe the properties of $\coin$. 
 %With a constant probability $\alpha$, all honest parties could  receive the same VRF's evaluation $r$ which is largest and computed by some honest party. All honest parties will output the lowest bit of $r$. Moreover, the low bit of the evaluation $r$ is unpredictable to the adversary.  

 \medskip
 \noindent
 {\bf Security analysis of $\coin$}. The main security properties of $\coin$ can be bridged to the next two key lemmas.
 Intuitively, Lemma \ref{lemma:terminate} states that every party can choose a speculative largest VRF  evaluation. This is   because at least $n-f$ $\seedgen$s and  $n-f$ $\AVSS$es must complete to ensure the completeness of $\wcs$, and implies the termination of $\coin$. While, Lemma \ref{lemma:core} bounds the probability of good case, and states that with a constant probability $\alpha$, all honest parties could  decide the same speculative largest VRF  evaluation $r_{\hat{l}}$ that is also evaluated by some honest party. This implies that with constant probability,  the  output bit  is as uniformly flipped. 
 
\begin{lemma}\label{lemma:terminate}{\em \bf (Termination)} 
	If all honest parties are activated on $\ID$, every honest party will decide a speculative largest VRF evaluation $r_{\hat{l}}$ with a valid proof $\pi_{\hat{l}}$, and all honest parties can eventually receive     the same $seed_{\hat{l}}$ that   $r_{\hat{l}}$ is evaluated on.
\end{lemma}
\begin{proof}
	 According to the correctness and commitment of $\seedgen$, 
	if all honest parties participate in $\seedgen[\langle \ID, j \rangle]$, every party will get the same $seed_j$  
	%and then output a pair $(r, \pi)$ from a VRF instance, 
	regarding an honest $\node_j$,
	which means  that all honest parties can complete their $\seedgen$s, compute their VRFs
	and activate their $\sssh$ instances that would finally joined by all honest parties.
	So every honest party will eventually complete at least $n-f$ $\sssh$ instances and record a $(n-f)$-sized set $\Si$ including the indexes of which $\sssh$ instances it participants in. Then every honest party activates $\wcs$ taking its set $\Si$.

	From the totality of $\AVSS$, if some honest party completes $\sssh$ instance on $\ID$, all honest parties will complete. So any index in some honest party's input set can eventually appear in all honest parties' set. From the termination of $\wcs$, all honest parties will output a set $\hat{\Si}$ and send $\RECREQ$ messages.
	
	From the validity of $\wcs$ and the totality of $\AVSS$, all honest parties would complete all $\sssh$ instances corresponds to its $\hat{\Si}$ and then start $\ssrec$s. According to the totality and commitment of $\AVSS$, all secrets corresponds to its $\hat{\Si}$ can be reconstructed.
	Recall that  for each $k \in \hat{\Si_i}$, $\node_i$ participants in $\sssh[ID, k]$. An honest party activates an $\sssh[\ID, k]$ only if it receives a $seed_k$ from  $\seedgen[\langle \ID, k \rangle]$. 
	This means that for each $k \in \hat{\Si}$, 
	$\node_i$  can output in $\seedgen[\langle \ID, k \rangle]$ to get a common $seed_k$.
	So for each $k \in \widehat{\Si}$, it can check whether $(k, r_k, \pi_k)$ is a validated VRF result, and pick up the maximum $r_l$ among all valid $r_k$. 

	Finally, every honest party sends the picked $(l, r_l, \pi_l)$ using a $\CANDIDATE$ message to all parties. All honest parties can eventually receive at least $n-f$ valid $\CANDIDATE$ from different parties. According to the totality and commitment of $\seedgen$, if any honest party gets  $seed_j$ from $\seedgen[\langle \ID, j \rangle]$, all honest parties would obtain the same $seed_j$ so all honest parties can get common VRF seeds and mutually consider whether others' $\CANDIDATE$ messages are valid, and then output $(j, r, \pi)$ where $r$ is the maximum of all $r_l$ among valid $\CANDIDATE$ messages.
\end{proof}

\begin{lemma}\label{lemma:core} {\em \bf (Good-case bound)} 
	Let $\mathsf{Event_{good}}$ to denote the case in which the (f+1)-supporting core-set $\Si^*$ solicits an honest party's VRF evaluation $r$ that is also largest among all parties' VRF. The remaining case denoted by $\mathsf{Event_{bad}}$ to cover all other possible executions. Then,   $\Pr[\mathsf{Event_{good}}] \ge \alpha=1/3$, under the  ideal functionality of VRF   in \cite{david2018ouroboros} (which is realizable in the random oracle model with CDH assumption).
\end{lemma}
\begin{proof}
	From the ($f+1$)-support core set of $\wcs$, at the moment when the first honest party outputs from $\wcs$ and invokes any $\ssrec$ instance, there has existed a core set $\Si^*$ including $2n/3$ indices, each of which represents a shared VRF's evaluation and at most $f < n/3$ indices of which are shared by adversary. 
	Each VRF's evaluation has a $1/n$ probability to be the maximal. Otherwise, the adversary directly breaks the pseudorandomness of VRF by biasing the distribution of corrupted parties' VRF evaluations (which is infeasible because according to the commitment and unpredictability of $\seedgen$, VRF seeds generated by $\seedgen$ protocols are unpredictable before it is committed, and once the $\seedgen$ completing the committing phase, the VRF seeds are fixed).
	Thus, the probability that the $\mathsf{Event_{good}}$ occurs is at least ${{{{2n}\over 3} - {n\over 3}}\over{n}}={1 \over 3}$. 
\end{proof}

%\begin{lemma}\label{lemma:core-unpred}
%	If $\mathsf{Event_{good}}$ defined in Lemma \ref{lemma:core} occurs, when $\mathsf{Event_{good}}$ occurs, there does not exists a polynomial adversary which can predicate a value $l=r \mod n$ better than guess. 
%\end{lemma}

%\begin{lemma}\label{lemma:core-common}
%	If $\mathsf{Event_{good}}$ defined in Lemma \ref{lemma:core} occurs, all honest parties will output the same bit $b$. 
%\end{lemma}

\begin{lemma}\label{lemma:core-unpred}
	If $\mathsf{Event_{good}}$ defined in Lemma \ref{lemma:core} occurs, there does not exist a polynomial adversary which can predicate the lowest bits(e.g. lowest $\lambda/2 $ bits) of $r_{\hat{l}}$. 
\end{lemma}

\begin{proof}
	Recall that the honest dealers' $\ssrec$s  leak nothing about their VRF's evaluations so at the moment when $\Si^*$ is fixed, the adversary learns nothing about honest parties' VRF evaluations. Thus when $\mathsf{Event_{good}}$ defined in Lemma \ref{lemma:core} occurs, the adversary cannot predicate the lowest bits of output better than guessing. Otherwise it can break the pseudorandomness of VRF by predicating the lowest bits of honest parties VRF's evaluations without accessing their secret keys.   
\end{proof}

\begin{lemma}\label{lemma:core-common}
	If $\mathsf{Event_{good}}$ defined in Lemma \ref{lemma:core} occurs, all honest parties will output the same bit $b$. 
\end{lemma}
\begin{proof}
	When $\mathsf{Event_{good}}$ defined in Lemma \ref{lemma:core} occurs, at least $f+1$ honest parties will receive the largest VRF's evaluation $r$ from some honest party and multicast it with $\CANDIDATE$ messages. All honest parties can receive at least one $\CANDIDATE$ message containing $r$ so that all honest parties will output the lowest bit of $r$.
\end{proof}

\begin{theorem}\label{thm:coin}
	In the the bulletin PKI setting and the random oracle model, our $\coin$ protocol (formally described in  Alg. \ref{alg:coin}) realizes $(n,f, 2f+1,1/3)$-$\coin$ against $n/3$ static Byzantine corruptions in the asynchronous message-passing model, conditioned on that the underlying primitives are all secure.
\end{theorem}
\begin{proof}
	We prove that Alg.  \ref{alg:coin} realizes the  properties of $\coin$ in  Def. \ref{def:coin} one by one:
	\begin{itemize}
		\item {\em Termination}. Termination can be proved directly from Lemma \ref{lemma:terminate}, 				
		
		\item {\em Reasonably fair bit-tossing}. From Lemma \ref{lemma:core}, the $\mathsf{Event_{good}}$ occurs with a probability $\alpha = 1/3$. 
		Before $f+1$ honest parties are activated to run the protocol, the adversary cannot predicate the protocol execution will fall into which case because no one can predicate the VRF seeds and thus even the corrupted parties cannot compute their VRF evaluations. 
		Following the same argument, from Lemma \ref{lemma:core-unpred},  before $f+1$ honest parties run the protocol, and when $\mathsf{Event_{good}}$ occurs, the adversary cannot predicate the lowest bit of the largest VRF's evaluation better than guessing. Moreover, in this case, all honest parties output the same bit according to  Lemma \ref{lemma:core-common}. Therefore,  the adversary succeeds in predicating some honest party's output with $\alpha/2$ probability.
		When $\mathsf{Event_{bad}}$ occurs with $1- \alpha$ probability, honest parties may not be able to output the same value, i.e., some honest parties output $0$ and some may output $1$. In this case, the adversary can always guess a bit $b$ which is equal to some honest parties' output. Therefore, the probability that adversary wins in the predication game is $\Pr[\adv \textnormal{ wins}] \le 1-\alpha + \alpha/2 = 1 - \alpha/2$, where $\alpha = 1/3$.
	\end{itemize}
\end{proof}
{\textsc{Remark on static security}}. The reason why our common coin cannot tolerate adaptive  corruptions is that the $\seedgen$ protocol is static, which is further caused by using Gurkan et al.'s statically-secure non-interactive  $\PVSS$ instantiation \cite{gurkan2021aggregatable}. Note that   if there is a one-time common random string $\eta$ announced after PKI registration, we do not have to use   $\seedgen$ to pack VRF, and  it is fine to directly use  $\eta$ as VRF seed (cf. Line 3 in Algorithm ). In such setting, our protocol actually can be adaptive secure, because $\AVSS$ and $\wcs$ are adaptively secure to ensure that: 
there must exist $f+1$ {\em forever} honest parties that hold a set of VRF evaluation-proof pairs with an intersection consisting of at least $n-f$ common VRFs, 
so Lemmas \ref{lemma:terminate} and \ref{lemma:core} still hold against an adaptive adversary in such setting. Namely,
given the extra setup assumption of one-time common  random string, our   coin protocol  (and also our later results including $\elect$, $\ABA$ and $\VBA$ constructions) can be adaptively secure.

\medskip
\noindent
{\bf Complexities of $\coin$}. The $\coin$ protocol   incurs $\bigO(n^3)$ messages, because it activates one $\wcs$, $n$ $\AVSS$es and $n$ $\seedgen$s in addition to $n^2$  $\CANDIDATE$ messages and $n^3$ $\RECREQ$ messages.
%each instance of which incurs $\bigO(n^2)$ messages. 
The overall communication complexity     is $\bigO(\lambda n^3)$, because each $\AVSS$ and $\seedgen$  incurs $\bigO(\lambda n^2)$ bits,   the $\wcs$ protocol exchanges $\bigO(\lambda n^3)$ bits,
and the   size of each $\CANDIDATE$ and $\RECREQ$ messages is $\bigO(\lambda n)$-bit and $\bigO(n)$-bit, respectively. 
Also,  the $\coin$ can terminate in constant asynchronous rounds, mainly because all underlying building blocks would output in constant rounds deterministically.

\subsection{Resulting ABA without private setup}

Given the new private-setup free $\coin$ protocol, one can construct more efficient asynchronous binary agreement ($\ABA$) with expected constant running time and cubic communications with PKI only. %By contrast, prior to our result, the state-of-the-art private-setup free $\ABA$   with optimal resilience and constant running time would incur prohibitive quadruple communications \cite{cachin2002secure}.
%
%, which improves the state-of-the-art result in the same setting by an $\bigO(n)$ factor.
%
In particular, we primarily focus on $\ABA$ with the following standard  definition \cite{cachin2000random,mostefaoui2015signature}.

\begin{definition}[Asynchronous Binary Agreement] \label{def:aba}
	%An instance of the $(n,f, k,\alpha)$-common coin protocol is among .
	A protocol realizes  $\ABA$, if it has syntax and properties defined as follows.
	
	\smallskip
	{\textsc{Syntax}}.
	For all execution of each protocol instance   with an identifier $\ID$, 
	each party input a single bit besides some implicit input including all parties public keys and its private key, and   outputs a   bit $b$. 	%
	
	\smallskip
	{\textsc{Properties}}. 
	It  satisfies the next properties   with all but negligible probability:
	\begin{itemize}
		\item {\bf Termination}. If all honest parties are activated on $\ID$, then every honest party outputs for $\ID$.
		
		\item {\bf Agreement}. Any two honest parties  that output associated to $\ID$ would output  the same   bit.
		
		\item {\bf Validity}. If any honest party outputs $ b $ for $\ID$, then at least an honest party inputs $ b $  for $\ID$.
	\end{itemize}
\end{definition}

%%%%%%%%%%%%%%%%%%%%%%%%%%%%%%%%%%%%%%%%%%%%%%%%%%%%%%%%%%%%%%%%%%%%%%%%%%%%%%%%%%%%%%%%%%%%%%%%

\medskip
\noindent
{\bf Constructing $\ABA$}. 
We refrain from reintroducing the $\ABA$ protocols presented in \cite{mostefaoui2015signature} and \cite{crain2020two},
as we only need to plug in our $\coin$ primitive to instantiate their reasonably fair common coin abstraction. More formally,

%Here we present an exemplary $\ABA$  construction shown in Alg. \ref{alg:aba} to showcase the applicability of our $\coin$ protocol.
%The demonstrating $\ABA$ protocol slightly adapts \cite{mostefaoui2014signature},
%with two major revisions, which are:
%\begin{enumerate}
%	\item {\em Extra confirm round}. 
%	This additional confirm round  \cite{ababug,macbrough2018cobalt} fixes a liveness issue of the original $\ABA$ protocol in \cite{mostefaoui2014signature},
%	as   asynchronous adversary might cause \cite{mostefaoui2014signature}  grind to a halt because   coin becomes predictable once some honest parties activate the primitive \cite{mostefaoui2015signature}.
	%Using the extra confirm round as fix was firstly pointed out  and later formally presented by . 
	
%	\item {\em Bootstrapping $b$-biased validity}. 
%	This adaption uses the ``preferable'' bit $b$ to replace the output of coin   in the first iteration of the $\ABA$ construction \cite{cachin2001secure}. %This simple trick   ensures every honest party to   output with $b$ when more than $f+1$ honest parties input $b$.
%\end{enumerate}

%\noindent
%We can   state the security and complexity of the exemplary $\aba$ protocol  as follows:

\begin{theorem}\label{thm:aba} 
	Given our $(n,f, 2f+1,1/3)$-$\coin$ protocol, \cite{mostefaoui2015signature} and \cite{crain2020two} can  implement $\ABA$ in the asynchronous message-passing model with $n/3$ static Byzantine corruption and bulletin PKI assumption, and    cost expected constant asynchronous rounds, expected $\bigO(n^3)$ messages and expected $\bigO(\lambda n^3)$ bits. 
\end{theorem}

The proofs for termination, agreement and validity can be found in \cite{mostefaoui2015signature} and \cite{crain2020two}, respectively. 
The complexity of   resulting $\ABA$s is dominated by our $(n,f, 2f+1,1/3)$-$\coin$ protocol, because given costless $\coin$, both  protocols can terminate in expected constant   rounds and exchange only  $\bigO(n^2)$ bits. 
%The complexities of   resulting $\ABA$ implementations would be dominated by our $(n,f, 2f+1,1/3)$-$\coin$ protocol, because given costless $\coin$, both  \cite{mostefaoui2015signature} and \cite{crain2020two} would attain expected constant asynchronous rounds, expected $\bigO(n^2)$ messages and expected $\bigO(\lambda n^2)$ bits.

%is similar to that in  \cite{macbrough2018cobalt}.
	
	%!TEX root = main.tex

\section{Augment: Towards Leader Election with Agreement}
\label{sec:elect}

This Section presents our  efficient asynchronous leader election ($\elect$) protocol without relying on private setup. This is the key step to realize fast, efficient and private-setup free multi-valued validated byzantine agreement ($\VBA$) in the asynchronous network environment. Considering that $\VBA$ is the quintessential core building block for efficient asynchronous DKG \cite{abraham2021reaching}, our technique  essentially can be plugged in the existing fast-terminating AJM+21 ADKG protocol to reduce its communication complexity from $\bigO(\lambda n^3 \log n)$ to $\bigO(\lambda n^3)$,
thus initialing a   new path to   easy-to-deploy replicated services in the asynchronous  network.

\subsection{Leader election without private setup}
\label{sec:leaderelec}

The aim of the $\elect$ primitive, in our context, is to randomly elect someone of the participating parties \cite{abraham2018validated}. 
More importantly, the primitive shall prevent the adversary from fully predicating which party would be elected, 
otherwise, the adversary might schedule message deliveries and cause the higher level protocol to never stop (or at least cause slow termination).  
For example, in Abraham et al's $\VBA$ \cite{abraham2018validated}, $\elect$ is invoked after $n-f$ input broadcasts are completed, and the   termination of this $\VBA$ protocol requires   $\elect$  to luckily choose an indeed completed input broadcast. Clearly, if the adversary can always predicate the $\elect$ result in advance, it can always delay the to-be-elected broadcast to make it not appear in the $n-f$ completed broadcasts,   thus causing $\VBA$ not to terminate. %Clearly, if the adversary cannot predicate the elected index to always delay . 
%, which is the crux to clip the power of the adversary for implementing multi-valued asynchronous consensuses (e.g., multi-valued validated byzantine agreement  \cite{lu2020dumbo}) with fast expected constant-time termination.

{\em Necessity of Agreement}. In addition, different from $\coin$ that only ensures agreement with a constant probability (e.g., 1/3), $\elect$  always has agreement.  
This is particularly important in many  $\VBA$  constructions, because $\elect$ is usually used to   decide which party's input becomes the final output,  
so lacking agreement in $\elect$ might immediately break   $\VBA$'s agreement. 
Essentially, the main task of this Section is to lift our $\coin$ protocol to realize the necessary   agreement. 

Formally, we would design an $\elect$ protocol realizing the following properties in the asynchronous   network without private setups: 

\begin{definition} [$(n,f, f+k, \alpha)$-Leader Election] \label{def:elect}
	A protocol is said to be $(n,f, f+k,\beta)$-$\elect$, if it is   among $n$ parties with up to $f$ static byzantine corruptions, and has syntax and properties   as follows.
	
		\smallskip
	{\textsc{Syntax}}.
	For each protocol instance with   session identifier $\ID$, 
	every party   takes the system's public knowledge (i.e., $\lambda$ and all public keys) and its own private keys  as input, and outputs a value $\ell \in [n]$. 
	%
	
	%\smallskip
	%{\textsc{Syntax}}.
	%An instance of the $(n,f, f+k,\beta)$-$\elect$  protocol is . 
	%When a party is activated to participate in the protocol's instance with a session identifier $\ID$, the party inputs the system's public knowledge (i.e., $\lambda$ and all public keys) and its own private keys, and outputs a value $\ell \in [n]$.  
	%
	
	\smallskip
	{\textsc{Properties}}. 
	It  satisfies the next properties except with negligible probability:
	\begin{itemize}
		\item {\bf Termination}. Conditioned on that all honest parties are activated on  $\ID$, every honest party would output a value $\ell \in [n]$. 
		\item {\bf Agreement}. For any two honest parties $\node_i$ and $\node_j$ that output $\ell_i$ and $\ell_j$ for $\ID$, respectively, there is $\ell_i = \ell_j$. 
		%\item {\bf Validity}. For all $\ID$, the probability that some honest party output a value $\ell$ is at least $\beta/n$ for each $\ell \in [n]$.
		\item {\bf Reasonably fair leader-election}. Before $k$ honest parties ($1 \le k \le f+1$) are activated on  $\ID$, 
		the adversary $\adv$ cannot always predicate the elected leader. 
		More precisely, 
		consider the   predication game:
		 $\adv$ guesses  $\ell^*$ before $k$ honest parties   activated on $\ID$,
		if $\ell^*$   equals to some honest party's output for $\ID$, we say that $\adv$ wins;
		we require   $\Pr[\adv \textnormal{ wins}] \le  1 - \alpha + \alpha/n$.
		
		Here $\alpha$   represents the lower-bound probability that the output is as if uniformly distributed over $[n]$,
		while $1-\alpha$ captures the possibility that the adversary might predicate/bias the output.
		%More precisely, %It is said that $\adv$ succeeds in predicating the output, if 
		%$\adv$ guesses a value $\ell^* \in [n]$ before $k$ honest parties activated on  $\ID$,
		%if $\ell^*$ coincides with some honest party's output $\ell$ for $\ID$, 
		%$\adv$ succeeds in predicating the output, the probability  of which shall be at most $1-\beta+\beta/n$.
		%the advantage of the adversary $\adv$ in the following Predication Game is at most $1-\alpha/2$.
	\end{itemize}
\end{definition}

\noindent	 
{\textsc{Remarks}}. 
%The unpredictability states that with at least $\beta$ probability taken over all possible $\elect$  executions, the adversary cannot predict the output index better than guessing over $[1,n]$.
When plugging in  a $(n,f, f+k,\alpha)$-$\elect$   with agreement ($k \ge 1$), 
most $\VBA$ constructions \cite{cachin2001secure,abraham2018validated,lu2020dumbo} can   preserve their securities and still terminate in expected constant rounds, as long as $\alpha$ is a certain constant between $(0,1]$. 
%The unpredictability of $\elect$ can be analog to that of $\coin$ to guarantee that each possible output value shall hold a probability to be decided.
Also, sometimes it can be important to realize larger $k$ (e.g., $f+1$) to clip the power of   asynchronous adversary as  in the $\VBA$ constructions of \cite{abraham2018validated,lu2020dumbo}, 
and we present a $(n,f, 2f+1,1/3)$-$\elect$ protocol.\footnote{Note that $(n,f, 2f+1,1/3)$-$\elect$ can prevent an adaptive adversary from always proposing the output in some $\VBA$ constructions \cite{abraham2018validated,lu2020dumbo} (which give the honest parties a chance to output their proposals). Similar to AJM+21 \cite{abraham2021reaching},  our $\elect$ protocol has the  potential to realize adaptive security if the underlying aggregatable $\PVSS$  can be adaptively secure. So we lift $k=f+1$ to maximize the strength of results.}

\begin{figure}
	\begin{center}
		\vspace{-0.5cm}
		\captionsetup{font={footnotesize}}
		\includegraphics[width=16.2cm]{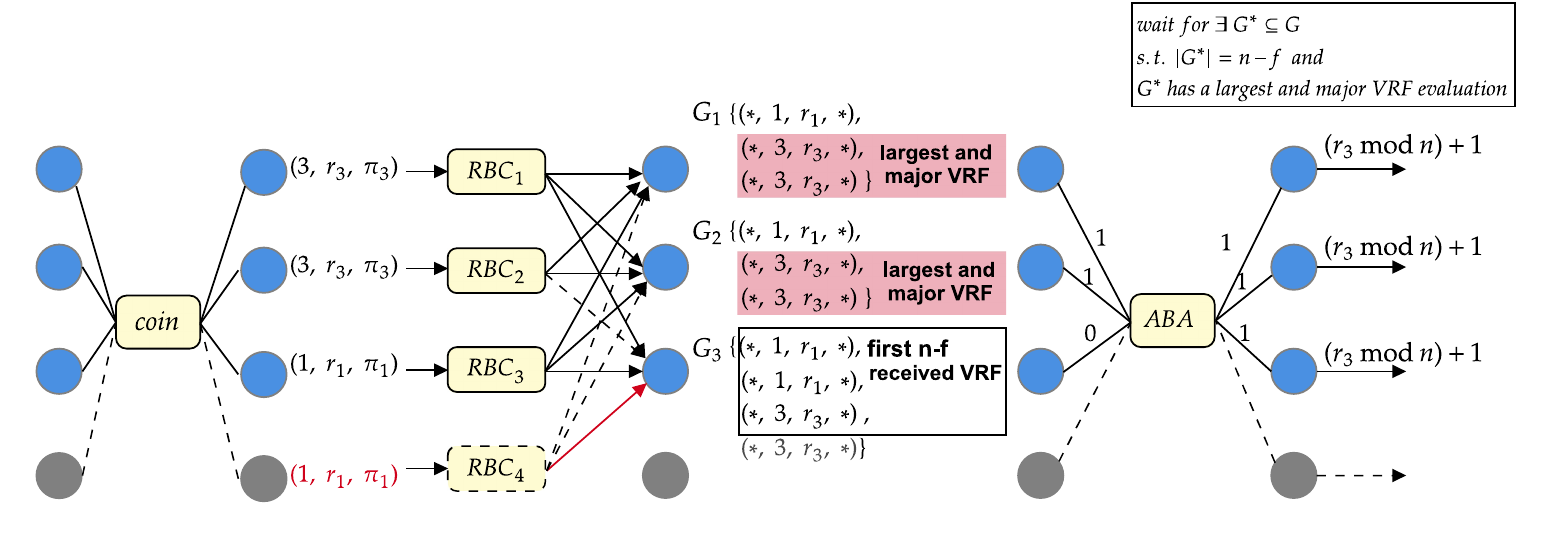}\\
		\vspace{-0.5cm}
		\caption{Overview of our $\elect$ protocol. Each party inputs $1$ to $\ABA$ if there exists a VRF evaluation which is the largest and majority among all $\RBC$ outputs. If $\ABA$ outputs $1$, wait for $G^* \subseteq G_i$ and output.} 
		\label{fig:elect}
	\end{center}
	\vspace{-0.5cm}
\end{figure}
 
\noindent
{\bf High-level rationale}. Our starting point of constructing $\elect$ is our $\coin$ protocol, at the end of which each party can have a speculative largest VRF. 
Recall Lemma \ref{lemma:terminate} and Lemma \ref{lemma:core} that reflect  the essential properties of   $\coin$.    

Lemma \ref{lemma:core} states that: with a constant probability $\alpha$, the speculative largest VRF of all honest parties is essentially same.
So our $\coin$ construction essentially can be thought of a leader election realizing   agreement with only $\alpha$ probability.
Thus, lifting perfect agreement   boils down to the problem of cleaning up the possible disagreement in the else $1-\alpha$ worse cases.

Furthermore, an in-depth thinking on Lemma \ref{lemma:terminate}   brings to light that all honest parties   can get the seeds needed to verify the speculative largest VRFs of all parties at the end of $\coin$ execution. This hints   a possibility that they can  cross-check the   largest VRFs for each other, and then vote on whether a common largest VRF indeed exists.
I.e., each party verifies some speculative largest VRFs from at least $n-f$ parties, and then votes 1 (resp. 0) to an asynchronous binary agreement ($\ABA$), according to whether there   exists   a largest and majority VRF out of the  $n-f$ verified VRFs (resp. or not). This efficient voting procedure through one single $\ABA$ can  resolve the possible disagreement on the speculative largest VRF at the end of coin execution, because the largest and majority VRF in the $n-f$ verified VRFs must be unique if  $\ABA$ returns 1.

\begin{algorithm}[H]
	%\captionsetup{font={scriptsize}}
	\begin{footnotesize}
		% \begin{scriptsize}
		
		\caption{$\elect$ protocol with identifier $\ID$,  {\color{blue} for {each party $\node_i$} } }\label{alg:elect}
		\begin{algorithmic}[1]

			\State $\G \leftarrow \emptyset$, $\G^* \leftarrow \emptyset$, %$\Sigma^* \leftarrow \emptyset$, 
			$ballot \leftarrow  0$, activate $\CBC[\langle \ID,j\rangle]$ for each $j\in [n]$
			\State run the code of $\coin$ in Alg. \ref{alg:coin} with replacing Line 31 by ``$\cmax \leftarrow (\hat{\ell}, r_{\hat{\ell}}, \pi_{\hat{\ell}})$" 
			%, $Flag1 \leftarrow 1$''
			%	\State {{\bf wait} for $\cmax$ assigned by $ (l^*, r^*, \pi^*)$}, i.e., Line 38 of modified Alg. \ref{alg:coin} is executed
			\Wait  { for $\cmax$ is assigned by $(\hat{\ell}, r_{\hat{\ell}}, \pi_{\hat{\ell}})$, i.e., Line 31 of modified $\coin$ algorithm is executed}
			%\State ${\sigma}^*_j\leftarrow \Sign^\ID_i(\langle \ell^*, r^* \rangle)$ and input $(\ell^*, r^*, seed^*, \pi^*, {\sigma}^*_j)$ to $\CBC[\langle \ID,i\rangle]$
			\State  input $\cmax$ to $\CBC[\langle \ID,i\rangle]$
			\EndWait

			\vspace*{0.13 cm}
			\Upon {$\CBC[\langle \ID,j\rangle]$ outputs $\cmax_{,j} = (\ell_j, r_{\ell_j}, \pi_{\ell_j})$} %$(\ell^*_j, seed^*_j, r^*_j, \pi^*_j,  {\sigma}^*_j)$}
			%\State {{\bf wait} for $seed_{l^*_j} \ne \bot$}
			\If {$\VerifyVRF^\ID_{\ell_j}(seed_{\ell_j}, r_{\ell_j}, \pi_{\ell_j}) = 1$}  		 \Comment{Verifying VRF implicitly waits for $seed_{\ell_j} \ne \bot$} 
			%	\Statex \Comment{Here the VRF validation implicitly waits for $seed_{{l}^*_j} \ne \bot$}
			\State $G \leftarrow G \cup (j, \ell_j, r_{\ell_j}, \pi_{\ell_j})$
			\If {$|G|=n-f$}
			\If {exist $(\cdot, \ell^*, r^*, \cdot)$ matching the majority elements in $G$}
			
			\If {$r^*$ is the largest VRF evaluation among all elements in $G$}
			\Statex {\ \ \ \ \ \ \ \ \ \ \ \ } \Comment{Namely, there exists $r^*$ that is the largest and majority VRF   among   $G$}
			\State   %$\G^* \leftarrow G$, 
			$ballot \leftarrow 1$ %and \textbf{send} $\VOTE(\ID, G^*)$ to all parties 				
			%\State let $c^* \leftarrow (\ell^*, r^*)$ and $\pi^*$ be any  valid  VRF proof $\pi^*_j$ for $r^*$ in $G$
			%\State let $\Sigma^* \leftarrow \{(j,  {\sigma}^*_j) ~|~ (j, \ell^*_j, r^*_j, \pi^*_j, {\sigma}^*_j) \in G ~\wedge~ \ell^*_j=\ell^* ~\wedge~ r^*_j=r^* \}$												
			
			\EndIf
			%\Else	 { \textbf{send} $\VOTE(\ID, No, \bot)$ to all parties}
			\EndIf

			%\Wait { for $ballot = 1$  or receiving $2f+1$ $\VOTE(\ID, No, \bot)$ messages from distinct parties}
			\State  activate $\ABA[\ID]$ with  $ballot$ as input
			\Wait { for that $\ABA[\ID]$ outputs $b$}
			\If {$b=1$}
			%\State  \textbf{send} $\VOTE(\ID, Yes, G^*)$ to all  if $ballot=1$ and  $\VOTE(\ID,Yes, G^*)$ not sent yet
			\Wait { for $\exists$ $G^* \subset G$ s.t. $|G^*|=n-f$ and $G^*$ has a largest and major VRF evaluation } %\Comment{Verifying VRF implicitly waits for $seed$}
			%\Statex   \Comment{Valid $G^*$ shall contain $n-f$ valid VRF evaluations signed by $n-f$ distinct parties, respectively; and there exists $r^*$ that is the largest and majority VRF evaluation  among  these $n-f$ elements in $G^*$}
			\State \textbf{output} $(r^* \bmod  n) +1$, where $r^*$ is the largest and majority VRF   among   $G^*$
			\EndWait
			\Else
			{  \bf output} the default index, i.e., 1
			\EndIf	
			\EndWait
			%\EndWait
			
			\EndIf
			\EndIf							
			\EndUpon
			
			%\vspace*{0.15 cm}
			%\Upon {receiving  valid $\VOTE(\ID,Yes, G_j^*)$ from $\node_j$ and $ballot = 0$}
			%\State$\G^* \leftarrow G$  and $ballot \leftarrow 1$ % record $G^*$
			%\EndUpon

			\vspace*{0.15 cm}

		\end{algorithmic}
		
	\end{footnotesize}
	%\end{scriptsize}
	
\end{algorithm}

\medskip
\noindent
{\bf Constructing $\elect$}. Our   a reasonably fair random $\elect$ protocol is formally shown in Alg. \ref{alg:elect}, and it has three main steps that proceeds as follows:
\begin{enumerate}
	\item {\em Committing the  largest VRF} (Line 1-4). Each party firstly runs the code of $\coin$ protocol, and obtains the speculate largest VRF's evaluation seen in its view, i.e., get  $\cmax=(\hat{\ell}, r_{\hat{\ell}}, \pi_{\hat{\ell}})$, where $\hat{\ell}$ represents that this speculative largest VRF is evaluated by which party, and $r_{\hat{\ell}}$ and $\pi_{\hat{\ell}}$ are the VRF evaluation and proof, respectively. 
	%Then, the party runs the $\coin$ and set $\cmax$ to $(\hat{l}, r_{\hat{l}}, \pi_{\hat{l}})$ instead of outputting the lowest bit of  $r_{\hat{l}}$. 
	After that, the party %signs $(\ell^*,r^*)$ and
	broadcasts %the signature along with 
	$\cmax$ to commit this speculative largest  VRF evaluation. 
	%Then broadcast the $\cmax$ and the signature. 
	Here Bracha's reliable broadcasts ($\CBC$)  \cite{bracha1987asynchronous} are used to prevent the corrupted parties from committing different speculative largest VRF to distinct parties.
	\item {\em Voting on how to output} (Line 5-13). A party can eventually output in $n-f$ $\CBC$s, each of which can return a  valid VRF evaluation-proof pair (which can be verified because the needed $seed$ can be waited from $\seedgen$ protocols). Then, the party checks if there exists a $\CBC$ output s.t. (i) it carries a VRF evaluation same to the majority of $n-f$ $\CBC$ outputs', and (ii) it also carries the largest VRF evaluation among all $n-f$ VRF evaluations received from $\CBC$s.  
	If such an element exists, %store a quorum set $\Sigma^*$ containing all signatures for this VRF and activate a $\zeroaba$ with 1 as input, otherwise, input 0. 
	%the party takes a ``snapshot'' of these $n-f$ $\CBC$s outputs to form $G^*$,
	%then it sends $\VOTE(\ID, G^*)$ to all parties %otherwise, send $\VOTE(\ID,No, \bot)$ to all parties. Then, everyone waits for $2f+1$ valid $\VOTE$ messages from distinct parties, if receiving a valid $\VOTE(\ID,Yes)$ message, then it 
	the party activates $\ABA$ with input 1, otherwise, it activates $\ABA$ with input 0.
	%Note that since a majority element needs $f+1$ valid signatures, there is at most one element which can meet the requirements.
	\item {\em Output decision} (Line 14-17). If   $\ABA$ outputs 1, then each party waits for that there exists a $(n-f)$-sized subset $G^*$ of all valid speculative largest VRFs received from $\CBC$s,
	%any $\VOTE(\ID, G^*)$ message\footnote{Remark that we use the $\VOTE(\ID, G^*)$ message for presentation brevity. Actually, a party does not have to send and/or wait $\VOTE$ message, because we can specify some quite involved ``rules''  used by   honest parties to check their local set $G$ to decide output if $\ABA$ returns 1.}
	s.t. there exists $r^*$ that is the largest and majority VRF  evaluation among all VRFs in  $G^*$,
	then it outputs $(r^* \bmod  n) +1$. 
	%send $\VOTE(\ID, Yes, G^*)$   if $G^*$ was recorded and  $\VOTE(\ID,Yes, G^*)$ not sent yet, and then wait for any valid $\VOTE(\ID, Yes, G^*)$ to output $(r^* \bmod  n) +1$, where $r^*$ is the largest and majority VRF  evaluation among all elements in   $G^*$. 
	If   $\ABA$ outputs 0, all parties would output a default index, e.g., 1.

	%send $(\ell^*, r^*,\pi^*, \Sigma^*)$ to all parties using a $\FINAL$ message. Otherwise, output the default index. After receiving $n-f$ valid $\FINAL$ messages containing the same $(\ell^*, r^*)$, output the lowest $\log n$ bits of $r^*$ as the elected leader. Otherwise, if $\zeroaba$ outputs 0, then output a default index, e.g., 0.
\end{enumerate}

\medskip
\noindent
{\bf Security Analysis of $\elect$}. The random leader election protocol presented in Algorithm \ref{alg:elect} securely realizes $(n,f, 2f+1,1/3)$-$\elect$   in the asynchronous message-passing model.
Here we prove its securities.

	\begin{lemma}\label{lemma:same}
	For any two parties $\node_i$ and $\node_j$, if there exists $(\cdot, \ell, r, \cdot)$ matching the majority elements in $G^*_i$ and $r$ is the largest VRF evaluation among all elements in $G^*_i$, and there exists $(\cdot, \ell', r', \cdot)$ matching the majority elements in $G^*_j$ and $r'$ is the largest VRF evaluation among all elements in $G^*_j$,  then the  $( \ell, r)=( \ell', r')$. 
	
	%	If any two honest parties $\node_i$, $\node_j$ set $ballot=1$ and fix $\G_i^*$ and $G_j^*$, respectively, then $r_i^* = r_j^*$, which are the largest and majority VRF among $\G_i^*$ and $G_j^*$. 
	%If any two honest parties $\node_i$, $\node_j$ record $c_i^*$, $c_j^*$, respectively, %and both with $1$ as the input of $\zeroaba$, 
	%then $c_i^*=c_j^*$.
\end{lemma}
\begin{proof}
	We prove this by contradiction. Suppose $r\neq r'$. By the code, $(\cdot, \ell, r, \cdot)$  and $(\cdot, \ell', r', \cdot)$ match the majority of $G^*_i$ and $G^*_j$, respectively, which means that the number of their appearance in $G^*_i$ and $G^*_j$ are at least $f+1$, respectively.  
	Without loss of generality, we assume that $r > r'$. %and other $f-1$ $r_\cdot^*$ is smaller than $r_j^*$. 
	Note that there are $n-f$ elements in $G^*_j$, so at least one valid $(\cdot, \ell, r, \cdot)$  must be included in $G^*_j$, because all elements in $G^*_i$ and $G^*_j$ are obtained via reliable broadcast that ensures agreement. Since $r'$ is the largest VRF evaluation among all elements in $G^*_j$, it also means $r'> r$, which is a contradiction to the assumption. Hence, $( \ell, r)=( \ell', r')$. 
\end{proof}

\begin{lemma}\label{lemma:agreement}
	For any two honest parties output $b=1$ from $\ABA$, they will output the same elected leader.
\end{lemma}
\begin{proof}
	If  $b=1$, from the validity of $\ABA$, there is at least one honest party activates $\ABA$ with input $1$, which implies that at least one honest party record a $G^*$ where exists $(\cdot, \ell^*, r^*, \cdot)$ matching the majority elements in $G^*$ and $r^*$ is the largest VRF evaluation among all elements in $G^*$. Since all elements are the outputs of $\RBC$s. From the totality, all honest parties can receive all elements in this $G^*$.
	Hence, each honest party can wait for a $G^* \subseteq G $ and then output.
	%at least one valid $\VOTE$ message containing a $G^*$ has been sent and all honest parties will receive it.  
	According to the Lemma \ref{lemma:same}, any valid $G^*$ has the same $(\cdot, \ell, r, \cdot)$ which matches the majority elements in $G^*$ and $r$ is the largest VRF evaluation. 
	Hence, all honest parties output the same value $(r \bmod  n) +1$.	
\end{proof}

\begin{lemma}\label{lemma:fair-good}
	When the $\mathsf{Event_{good}}$ defined in Lemma \ref{lemma:core} occurs, the polynomial-time adversary cannot predicate the elected leader better than guess.
\end{lemma}
\begin{proof}
	From Lemma \ref{lemma:core-common}, when the $\mathsf{Event_{good}}$ occurs, all honest parties will output the same $\cmax=(\ell^*, r^*, \pi^*)$ after running the code of $\coin$. In this case, all honest parties have the same $(\ell^*, r^*, \pi^*)$ and send it by $\RBC$. 
	Following the validity of $\RBC$, each honest party can receive at least $n-f$ messages from distinct $\RBC$ instances, at least $n-2f \ge f+1$ of which are sent by distinct honest parties and contain the same $(\ell^*, r^*, \pi^*)$. So all honest parties can collect a $G^*$, in which $(\cdot, \ell^*, r^*, \cdot)$ matches the majority elements and $r^*$  is the largest VRF. Then all honest parties activate $\ABA$ with $1$ as input. According to the validity of $\ABA$, all honest parties will output $1$ from $\ABA$, then all honest parties output the same value $( r^* \bmod  n) +1$ according to Lemma \ref{lemma:agreement}.
	
	Following   Lemma \ref{lemma:core-unpred},  with a probability $\alpha = \Pr[\mathsf{Event_{good}}]=1/3$, the adversary cannot predict   $\ell = ( r^* \bmod  n) +1$ better than guessing.
	Thus in this case, the probability that the adversary $\Adv$ succeeds in predicting   some honest party's output is no more than $\alpha \over n$. 
\end{proof}

\begin{theorem}\label{thm:elect}
	In the bulletin PKI setting,  Alg. \ref{alg:elect} realizes $(n,f, 2f+1,1/3)$-$\elect$   in the asynchronous message-passing model against $n/3$ static Byzantine corruptions, conditioned on that the underlying primitives are all secure.
\end{theorem} 
\begin{proof}
	Here we prove that Alg. \ref{alg:elect} satisfies the  properties of $\elect$ given in Def. \ref{def:elect} one by one:
	\begin{itemize}
		\item {\em Termination}. From Lemma \ref{lemma:terminate}, each honest party will output a $\cmax=(\ell^*, r^*, \pi^*)$, %then generate a signature $\sigma^*$ for $((\ell^*, r^*)$, after that every honest 
		then each honest party will broadcast its $(\ell^*, r^*, \pi^*)$ using $\CBC$. According to the validity of $\CBC$, each honest party can eventually collect a set $G$ containing at least $n-f$ $\CBC$ outputs. For an honest party, if there exists $(\cdot, \ell^*, r^*, \cdot)$ matching the majority elements in $G$ and $r^*$ is the largest VRF evaluation among all elements in $G$, activates the $\ABA[\ID]$ with $1$ as input, otherwise, inputs $0$ into $\ABA[\ID]$.

		%
		%Honest parties send either $\VOTE(\ID, G^*)$ or $\VOTE(\ID, No, \bot)$ to all parties.  All of them can wait for $ballot=1$ or $2f+1$ $\VOTE(\ID, No, \bot)$, and then activate the $\ABA[\ID]$ with $ballot$ as input.

		%according to whether there are any elements in its $G$ that meet the conditions. 
		According to the termination and agreement of $\ABA$, if all honest parties participate in the $\ABA$, then all of them will output the same bit $b$. If $b=0$, all honest parties output the default index, i.e., $1$. 
		If $b=1$, from Lemma \ref{lemma:agreement}, all honest parties will output a same value.

		%If $b=1$, from the validity of $\ABA$, at least one honest party inputs with $1$, by the code, it also implies this honest party sends a valid $\VOTE$ message carrying a $G^*$ to all. Due to that each element of $G^*$ is the output of $\RBC$, following the totality of $\RBC$, all honest parties can receive all elements in $G^*$. Hence, after receiving a $\VOTE$ message containing $G^*$, each honest parties can wait for $G^* \subset G $ and then output.
		%all honest parties can eventually wait for a valid $\VOTE$ message containing a $G^*$ which  $G^* \subset G $ and then output.

		\item {\em Agreement}. According to the termination and agreement of $\ABA$, if all honest parties participate in the $\ABA$, all of them would output from $\ABA$ with the same bit $b$. We analyze it in two cases: 
		(i), If $b=0$, it is obvious that all honest parties will output the default index, i.e., $1$.
		(ii), If $b=1$, from Lemma \ref{lemma:agreement}, all honest parties will output $(r \bmod  n) +1$, where $r$ is the largest VRF's evaluation.
		%there is only one $r^*$ that will be accepted by honest parties who set $ballot=1$, so the adversary $\Adv$ can not create a $G'^*$ which contain $n-f$ valid VRF evaluation signed by $n-f$ distinct parties, and $r'^* \neq r^*$ is the largest and majority VRF evaluation among $G'^*$. Therefore all honest parties will accept the same $r^*$ and output the same value.honest parties who set  input $ballot=1$ in $\ABA$(from validity of $\ABA$, the numbers of honest parties who input $1$ is at least $1$) send $(\ID, YES, G^*)$ by a $\VOTE$ or $\FINAL$ message.  
		%Each of these honest parties will send $\FINAL(\ID,c^*, \pi^*, \Sigma^*)$. %Other honest parties who input $0$ in the $\zeroaba$ will send $(\bot, \bot, \emptyset)$. 
		%So all honest parties can eventually wait for at least $1$ valid $\VOTE$ or $\FINAL$ messages. Note that according to the Lemma \ref{lemma:same}, there is only one $r^*$ that will be accepted by honest parties who set $ballot=1$, so the adversary $\Adv$ can not create a $G'^*$ which contain $n-f$ valid VRF evaluation signed by $n-f$ distinct parties, and $r'^* \neq r^*$ is the largest and majority VRF evaluation among $G'^*$. Therefore all honest parties will accept the same $r^*$ and output the same value.

		\item {\em Reasonably fair leader-election}. We use the $\mathsf{Event_{good}}$ and $\mathsf{Event_{bad}}$ defined in the Lemma \ref{lemma:core} to discuss this property by two cases.  
		Case (i): With probability $\alpha$, the $\mathsf{Event_{good}}$ occurs, then from Lemma \ref{lemma:fair-good}, the probability that the adversary $\Adv$ succeeds in predicting the output $\ell$ which coincides with some honest party's output is no more than $\alpha \over n$. 
		%
		%party receive $n-f$ $\VOTE(\ID, No, \bot)$ from distinct parties, it means at least one 
		% 
		%at least one honest party can receive $n-f $ message from distinct $\RBC$ instance to form $G$, where at least $f+1$ elements in $G$ with same $(\cdot, \ell, r, \cdot, \cdot, \cdot)$  and $r$ is the largest, then the honest node send $\VOTE(\ID, Yes, G)$ to all,	the other honest parties will also send  
		% 	
		%
		Case (ii): if the $\mathsf{Event_{bad}}$ occurs, the adversary $\Adv$ might lead up to different honest parties to obtain different $\cmax$, so that some honest parties would not be able to find the majority elements $(\cdot, \ell^*, r^*, \cdot)$, where $r^*$ is the largest VRF evaluation in $G$. 
		Nevertheless, it cannot be worse than that $\ABA$ always outputs 0 and the adversary always predicates the output.
		%In this case, it would input $0$ into $\ABA$ and the adversary is able to make all honest parties output $0$. Then all honest parties will output the default value (i.e. $1$) with probability $1-\beta$. So the adversary can obviously predicting the output for this case. 
		%
		
		In sum, the probability that adversary wins in the predication game is $\Pr[\adv \textnormal{ wins}] \le 1-\alpha +  \alpha /n $, where $\alpha = 1/3$.

	\end{itemize}
\end{proof}

%such that any honest party input 1 there exists a unique VRF evaluation.
%The biased validity of the primitive ensures that: if it returns 1, 
%when there are at least $f+1$ honest parties that pick up a common VRF evaluation-proof within $G$,
%so these parties can multicast $\VOTE(\ID,Yes)$ message such that all honest nodes can received one valid  $\VOTE(\ID,Yes)$ message, then all honest nodes will activate $\ABA$ with input 1. Besides, for any two valid $\VOTE(\ID,Yes, G)$ and $\VOTE(\ID,Yes, G')$, then  there exists $(\cdot, \ell, r, \cdot, \cdot, \cdot)$ matching the majority elements in $G$ and $G'$, meanwhile, $r$ is the largest VRF evaluation among all elements in $G$ and $G'$.
%Hence, with at least $1/3$ probability, not only $\ABA$ returns 1, but also the honest parties pick a common VRF evaluated by an honest party.
%In   other cases when  $\ABA$ returns 0,
%though the default index is the output, the probability of such worse cases is up to $2/3$.

\medskip
\noindent
{\bf Complexities of $\elect$}. %The complexities of the $\elect$ protocol can be easily seen as follows.
The overall message complexity is expected $\bigO(n^3)$, because the execution of $\coin$ part spends  $\bigO(n^3)$ messages, $n$ $\CBC$ incurs $\bigO(n^3)$ messages,  and the $\ABA$ instance costs expected $\bigO(n^3)$ messages. The overall exchanged bits are $\bigO(\lambda n^3)$ bits on average, because each $\CBC$ instance  costs $\bigO(\lambda n^2)$ bits and  the $\ABA$ instance incurs expected $\bigO(\lambda n^3)$ bits. Moreover, the   $\elect$ protocol can terminate in expected $\bigO(1)$ asynchronous rounds, which is mainly dominated by the underlying $\ABA$ instance.

\subsection{Resulting    VBA without private setup}
Given our $\elect$ protocol, we are ready to construct the private-setup free validated Byzantine agreement ($\VBA$), which is essentially a special Byzantine agreement  with external validity.
As aforementioned, it is     the core building block in many asynchronous protocols, such as atomic broadcast \cite{guo2020dumbo} and fast-terminating asynchronous DKG \cite{abraham2021reaching}. 
More formally, it can be defined as follows.

\begin{definition}[Asynchronous Validated Byzantine Agreement ($\VBA$)] \label{def:vba}
	%An instance of the $(n,f, k,\alpha)$-common coin protocol is among .
	%A protocol realizes  $\VBA$, if it has syntax and properties defined as follows.
	
	\smallskip
	{\textsc{Syntax}}.
	For each $\VBA$   instance   with an identifier $\ID$ and a  polynomial-time computable global predicate $Q_\ID$, 
	each party inputs a value  (besides the implicit inputs including all   public keys and its own private key), and   outputs a  value. 	%
	
	\smallskip
	{\textsc{Properties}}. 
	It  satisfies the next properties   with all but negligible probability:
	\begin{itemize}
		\item {\bf Termination}. If all honest parties activate  on $\ID$ with an input satisfying $Q_\ID$, then every honest party outputs for $\ID$.
		
		\item {\bf Agreement}. Any two honest parties  that output associated to $\ID$ would output  the same   value.
		
		\item {\bf External-Validity}. If any honest party outputs $ v $ for $\ID$, then $Q_\ID(v)=1$.
	\end{itemize}
\end{definition}

\medskip
\noindent
{\bf Constructing $\VBA$   without private setup}. 
Most existing     $\VBA$ constructions \cite{cachin2001secure,abraham2018validated,lu2020dumbo} 
rely on a pre-configured non-interactive threshold PRF (tPRF) \cite{cachin2000random}
to implement a Leader Election primitive that can uniformly elect a common party out of all parties.
Alternatively, our reasonably fair $\elect$ protocol is pluggable in all $\VBA$  implementations \cite{cachin2001secure,abraham2018validated,lu2020dumbo}
to replace   tPRF, thus removing the possible unpleasant private setup of it. More formally,

\begin{theorem}\label{thm:aba} 
	In the bulletin PKI setting, given our $(n,f, 2f+1,1/3)$-$\elect$ protocol, \cite{cachin2001secure,abraham2018validated,lu2020dumbo}    implement $\VBA$ in the asynchronous message-passing model with $n/3$ static Byzantine corruption, and         cost expected constant rounds, expected $\bigO(n^3)$ messages, and expected $\bigO(\lambda n^3)$ bits (w.r.t. $\lambda n$-bit or shorter inputs).
	%\footnote[4]{For   larger input having $\lambda n$ bits,    the  asymptotic complexities of the  privates-setup free $\VBA$ instantiations   can remain same to those of the $\lambda$-bit input cases.}. 
\end{theorem}

Proofs for $\VBA$ properties can be found in \cite{cachin2001secure,abraham2018validated,lu2020dumbo} with some trivial adaptions.
Here we briefly  discuss  the security intuition behind the securities and the resulting   complexities, with using Abraham et al.'s $\VBA$ construction  \cite{abraham2018validated} as an example.

%
%
%(or complex asynchronous distributed key generation (ADKG) which is due to the existing $\VBA$ construction \cite{abraham2018validated,cachin2001secure,lu2020dumbo} using Leader Election as a key underlying protocol. 
%
%However, the Leader Election primitive achieved by a private setup $(t,n)\textrm{-}\coin$ (the essence is a $(t,n)\textrm{-}$threshold signature), which is rely on a trust third party or ADKG. %There have been some works on constructing $\VBA$ \cite{abraham2018validated,cachin2001secure,lu2020dumbo} using leader election as a key step, however,  these works use a private setup $(t,n)\textrm{-}\coin$ (the essence is a $(t,n)\textrm{-}$threshold signature) to build the  Leader Election primitive.  
%Hence, how to achieve the atomic broadcast (or blockchain) with PKI is our goal.  
%
%This private-setup free Leader Election protocol can be applied to construct a private-setup free $\VBA$. In this section,  to simplify the presentation, we just take the $\VBA$  from \cite{abraham2018validated} as an example to show the application. 
%
%
%
Recall that  each    iteration of Abraham et al.'s $\VBA$ \cite{abraham2018validated} proceeds as follows: 
every party begins as a leader  to perform a 4-stage provable-broadcast (PB) protocol to broadcast a {\bf key}  proof (carrying a value as well), a {\bf lock} proof, and a {\bf commit} proof, where each proof is essentially a quorum certificate;
following  {\bf key}-{\bf lock}-{\bf commit} proofs, each party can further generate and multicast a completeness proof, 
attesting that it delivers these proofs to at least $f+1$ honest parties, which is called leader nomination; 
then, after at least $n-f$ 4-staged PBs are proven to complete leader nominations, 
a   $\elect$ primitive is needed to  sample a party called leader in a perfect fair (or reasonably fair) way; 
so with some constant probability, i.e., $2/9$ in case of plugging our $(n,f, 2f+1,1/3)$-$\elect$ protocol, 
the elected leader already finished its nomination and   delivered a {\bf commit} proof to at least $f+1$ honest parties, 
and these parties can output the value received from the leader's 4-stage PB;
and after one more round to multicast and amplify the proofs, all parties also output the same value; 
otherwise, it is a worse case with $7/9$ probability, in which no enough honest parties {\bf commit} regarding the elected leader, 
and the protocol   enters the next iteration; nonetheless, the nice properties of  the {\bf key} and {\bf lock}  proofs would ensure that the parties can luckily output in the next iteration with the same $2/9$ chance. % after one extra ``view-change'' round. %after  the view-change round to exchange proofs, 

Thus,   plugging our $\elect$ primitive into Abraham et al.'s $\VBA$ would preserve its constant running time. 
For message complexity, no extra cost is placed  except our $\elect$ primitive, so it becomes dominated by the $\bigO(n^3)$ messages of $\elect$.
For communication complexity, it is worth noticing that non-interactive threshold signature scheme is used to form short quorum certificates in the 4-staged provable-broadcast (PB) protocols; nevertheless, such instantiation of quorum certificate can be  replaced by trivially concatenating digital signatures from $n-f$ distinct parties in the bulletin PKI setting, which only adds an $\bigO(n)$ factor to the size of quorum certificates, thus causing   $\bigO(\lambda n^3)$  communication complexity to this private-setup free $\VBA$ instantiation (for $\lambda n$-bit input).

\subsection{Applications}\label{sec:application}

\noindent
{\bf Application to   asynchronous DKG}. The resulting $\VBA$ protocols can be  plugged in  AJM+21   $\ADKG$ \cite{abraham2021reaching} to  reduce  the communication 
to $\bigO(\lambda n^3)$ bits, with preserving fast termination in expected $\bigO(1)$   rounds and optimal $n/3$  resilience.

The basic idea (cf. Section 7.5 in \cite{abraham2021reaching})   lets each party multicast an aggregatable $\PVSS$ hiding a random secret. Then, everyone gathers and combines $n-f$ $\PVSS$ from distinct parties.
So they can input the aggregated $\PVSS$ to one $\VBA$ instance  (with external validity   specified to check the input is indeed $\PVSS$ aggregated by $n-f$ parties' contributions).
Finally, each party can get a consistent  $\PVSS$ script returned by $\VBA$, and therefore can decrypt it to get its     key share. 
The resulting      communication cost is $\bigO(\lambda n^3)$, because  all $\PVSS$ scripts are $\bigO(\lambda n)$-bit.
%Coupling the $\VBA$ protocol from \cite{abraham2018validated}  with our private-setup free Leader Election protocol ($\elect$) from Section \ref{sec:leaderelec}, each node with probability $1/3$ to be the leader due to the validity of $\elect$, hence, with a probability of $2/9$, the new $\VBA$ with setup free terminate in the current view, it clear that the new $\VBA$  terminate in constant time. Besides, we also replace all threshold signature with digital signature scheme, that is,  
%the original  quorum certificate in $\VBA$ \cite{abraham2018validated} is a threshold signature,  
%we use enough digital signature (i.e. $2f+1$) to form a quorum certificate to replace the threshold signature, it also did not change the security use different signature scheme. Hence, the new $\VBA$ protocol still satisfies all properties of $\VBA$. 
%Besides, the size of proof  is $\bigO(n\lambda)$ instead of $\bigO(\lambda)$,  
%hence, we can obtain a new $\VBA$ protocol with private-setup free that just incurs  $O(\lambda n^3)$-bit complexity.  
% 
%For the $\VBA$ from \cite{cachin2001secure} and \cite{lu2020dumbo}, it needs extra underlying protocols $\ABA$, in this papers, we also achieved a private-setup free $\ABA$, hence, our works can easily be adopted in these $\VBA$ such that to achieve an atomic broadcast (or blockchain) with PKI. 

\medskip
\noindent
{\bf Application to     random beacon w/o DKG}. Our $\elect$ protocol can be slightly adapted to realize an asynchronous random beacon service that all participating parties can proceed by consecutive epochs and continually output an unbiased and unpredictable value in each epoch. 
Here unbiased means the output is uniformly distributed~\cite{Bhat2021Rand,das2021spurt}; and the unpredictable means that the adversary cannot tell the   random output of next epoch
better than guessing, unless $f+1$ honest parties already output in the current epoch. 
%promises bias-resistance and 1-absolute unpredictability~\cite{Bhat2021Rand} with $\bigO(\lambda n^2)$ communication cost. 

To implement asynchronous random beacon, we can let all parties to   execute a sequence of $\elect$ protocols, with the following minor changes:
 (i) when $\ABA$ unluckily returns 0 and thus no largest VRF is agreed, the honest parties do not output the default value, and they directly move into the next $\elect$ instance;
 (ii) instead of returning a short index belong $[n]$, the parties can output the lowest $\bigO(\lambda)$ bits of the selected largest VRF, e.g., the half least-significant bits of the largest VRF evaluation.
%let all parties repeatedly run the $\elect$ protocol with only minor changes: In each epoch, if $b=1$, all honest parties will output $\hash(r^*)$ and move to the next epoch, or else they rerun $\elect$ from the beginning. 
The former adaption ensures that a non-default value will be output with a probability of $1-{(1-\alpha)}^k$ after sequentially running  $k$ $\elect$ instances, 
and thus a random value can always be output after expected constant rounds.
%This random beacon protocol can provide an output in $k$ times of running $\elect$ with a probability of $1-{(1-\alpha)}^k$.

For bias-resistance, according to the committing and unpredictability of $\seedgen$, the adversary cannot manipulate the generation of VRF seeds so that they cannot bias the $\VRF$s evaluated on the $seed$ or immediately break the unpredictability of $\VRF$. The unpredictability is similar, because before $f+1$ honest parties invoke $\seedgen$ protocols, the adversary cannot predicate the output VRF seeds, so all VRF evaluations in the next epoch would remain secret to the adversary.
	\bibliographystyle{splncs04}
	\bibliography{reference}
	%

	%!TEX root = main.tex

%\section*{Supplemental Materials}

%\newpage
\begin{subappendices}
	\renewcommand{\thesection}{\Alph{section}}%
	
	\section{Secrecy Game for AVSS}\label{app:sec_game}
	
	\noindent 
	{\bf Secrecy Game}. 
	{\em 
		The Secrecy game between an adversary $\adv$ and a challenger $\cha$ is defined as follows to capture the secrecy threat in the $\AVSS$ protocol among $n$ parties with up to $f$ static corruptions in the bulletin PKI setting:
		\begin{enumerate}
			\item  $\adv$ chooses a set $\bar{Q}$ of up to $f$ parties to corrupt, a dealer $\node_D$ ($\node_D\notin\bar{Q}$) and a session identifier $\ID$, and also  generates the secret-public key pairs for each corrupted party in $\bar{Q}$, and sends $\bar{Q}$, $\node_D$, $\ID$ and all relevant public keys to   $\cha$.
			\item The challenger $\cha$ generates the secret-public key pair  for every honest party in $[n] \setminus \bar{Q}$, and sends these public keys to the adversary $\adv$.
			
			\item $\adv$ chooses two secrets $s_0$ and $s_1$ with same length and send them to $\cha$. 
			
			\item The challenger  $\cha$ decides a hidden bit $b\in\{0,1\}$ randomly, executes the $\sssh$ protocol on $\ID$ (for all honest parties) to share $s_b$ via interacting with $\adv$ (that are on behalf of  the corrupted parties). During the execution, $\adv$ is consulted to schedule all message deliveries and would learn: (i) the protocol scripts sent to the ``corrupted parties'', and (ii) the length of all messages sent among the honest parties.
			
			\item The adversary $\adv$ guesses a bit $b'$. 
		\end{enumerate}
		The advantage of $\adv$ in the above Secrecy game ${\bf Adv_{sec}}$ is  $|\Pr[b=b'] -1/2 |$.
	}
	
	\medskip
	\noindent
	Recall that the secrecy requirement of $\AVSS$ requires that the adversary's advantage in the above game shall be  negligible.

	\section{An Implementation of Seed Generation}\label{app:seed}
	
	Here we give an exemplary construction for reliable   leaded seeding ($\seedgen$) through the elegant idea of   aggregatable public verifiable secret sharing ($\PVSS$)   in \cite{gurkan2021aggregatable}. The general idea of \cite{gurkan2021aggregatable} is to lift the beautiful $\scrape$ $\PVSS$ scheme \cite{cascudo2017scrape} to enable the aggregation of $\PVSS$ scripts from distinct participating parties, and along the way, it presents a way of using knowledge-of-signatures to allow each party to attach an aggregatable ``tag'' attesting its contribution in the aggregated $\PVSS$ script, thus ensuring that anyone can check the finally aggregated $\PVSS$ script $\pvss$ (and tag) to verify whether $\pvss$ indeed commits a polynomial collectively ``generated'' by more than $f$ parties while   preserving the size of communicated  scripts   minimal.
	
	It thus becomes immediate to follow the  nice idea   to construct an efficient $\seedgen$ protocol as shown in Alg. \ref{alg:seed} in the PKI setting,
	in which: (i) each party firstly generates a $\scrape$ $\PVSS$ script along with a knowledge-of-signature, such that the leader can collect and aggregate $2f+1$  $\scrape$ $\PVSS$ scripts and multicast the aggregated $\PVSS$ script  along with a vector of knowledge-of-signatures and some other metadata (called $\pvsstag$ by us) to the whole network, later (ii) each party returns to the leader a signature for the aggregated $\PVSS$ script, so the leader can collect and multicast a quorum certificate containing at least $2f+1$ signatures to attest that it has ``committed''  a consistent $\PVSS$ script collectively contributed by at least $2f+1$  parties across the whole network, finally (iii) it becomes simple for every party to multicast its own secret share regarding to the final $\PVSS$ script,
	thus reconstructing the secret collectively generated by at least $2f+1$ parties, which is naturally the random seed to output.

	For sake of completeness, we briefly review the cryptographic abstraction of the aggregatable   $\PVSS$  scheme (with unforgeable tags attesting contributions) among $n$ parties with a secrecy threshold $t$ due to Gurkan et al.  \cite{gurkan2021aggregatable}. Note that we mainly focus on abstracting the needed properties to construct and prove our $\seedgen$ protocol.

 %\vspace{-0.5cm}
 \begin{algorithm}[!h]
 	%\captionsetup{font={scriptsize}}
 	\begin{footnotesize}
 		%\begin{scriptsize}
 		
 		\caption{Implementation of Aggregatable PVSS   \cite{gurkan2021aggregatable}}
 		\label{alg:aggpvss}
 		\begin{algorithmic}
 			
 			\vspace*{0.1cm}
 			\State $\deal(\allek,  \sk_i, a_0) \rightarrow \pvss$
 			\begin{itemize}
 				\item \textbf{initialize} ($w_1$, .. , $w_n$) $\leftarrow$ ($0$, ..., $0$); ($C_1$, .. , $C_n$) $\leftarrow$ ($\bot$, ..., $\bot$); ($\sigma_1$, .. , $\sigma_n$) $\leftarrow$ ($\bot$, ..., $\bot$),
 				\item $w_i \leftarrow 1$; $C_i \leftarrow  {g_1}^{a_0}$; $\sigma_i \leftarrow \mathsf{SoK}.\Sign(C_i, sk_i, c_i)$
 				\item  randomly choose ($a_1$, ..., $a_t$) from $\mathbb{F}^t$
 				\item $F(X) \leftarrow \sum_{i=0}^{t}a_iX^i$
 				\item $F_0, ..., F_t \leftarrow {g_1}^{a_0}, ...,{g_1}^{a_t}$;
 				$\hat{u}_2 \leftarrow {\hat{u}_1}^{a_0}$;  $A_1$, ..., $A_n$ $\leftarrow$ ${g_1}^{f(\omega_1)}$, ...,  ${g_1}^{f(\omega_n)}$; $\hat{Y}_1$, ..., $\hat{Y}_n$ $\leftarrow$ ${ek_1}^{f(\omega_1)}$, ..., ${ek_n}^{f(\omega_n)}$
 				\item \textbf{return} $\mathbf{F}$, $\hat{u}_2$,  $\mathbf{A}$,  $\mathbf{\hat{Y}}$, $(C_1, ..., C_n)$, $(w_1, ..., w_n)$, $(\sigma_1, ..., \sigma_n)$
 			\end{itemize}			 	
 			
 			\vspace*{0.2 cm}
 			\State $\vrfyscript(\allek, \allvk, \pvss) \rightarrow 0/1$
 			\begin{itemize}
 				\item $\mathbf{F}$, $\hat{u}_2$,  $\mathbf{A}$,  $\mathbf{\hat{Y}}$, $\mathbf{C}$ $\boldsymbol{w}$, $\boldsymbol{\sigma}$ $\leftarrow$ \textbf{parse}($\pvss$)
 				\item  $\alpha \stackrel{\mathbb{\$}}{\leftarrow} \mathbb{F}$
 				\item \textbf{check} $\prod_{j=1}^{n}{{A_j}^{l_j(\alpha)}} = \prod_{j=0}^{t}F_j^{\alpha^j}$, $l_j(X)$ denotes the Lagrange polynomial equal to $1$ at $\omega_j$ and $0$ at $\omega_i \neq \omega_j$
 				\item \textbf{check} $e(F_0, \hat{u}_1) = e(g_1, \hat{u}_2)$
 				\item \textbf{check} $e(g_1, \hat{Y}_j) = e(A_j, ek_j)$ for each $1 \leq j \leq n$
 				\item \textbf{for} $1 \leq i \leq n$:
 				%\begin{itemize}
 				%\item 
 				\textbf{if} $w_i \neq 0$, then \textbf{check} $\mathsf{SoK}.\Vrfy(vk_i, C_i, \sigma_i)=1$
 				%\end{itemize}
 				\item \textbf{check} $C_1^{w_1} ...C_n^{w_n}=F_0$
 				\item \textbf{return} $1$ if all checks pass, else \textbf{return} $0 $
 			\end{itemize}
 			\vspace*{0.2  cm}
 			\State  $\agg(\pvss_1, \pvss_2) \rightarrow \pvss$
 			\begin{itemize}
 				\item $(F_{1,0}, ..., F_{1,t})$,~$\hat{u}_{1,2}$,~$(A_{1,1}, ..., A_{1,n})$,~$(\hat{Y}_{1,1}, ...,\hat{Y}_{1,n})$,~$(C_{1,1}, ..., C_{1,n})$,~$(w_{1,1}, ..., w_{1, n})$,~$(\sigma_{1,1}, ..., \sigma_{1,n})$~$\leftarrow$~\textbf{parse}($\pvss_1$)
 				\item $(F_{2,0}, ..., F_{2,t})$,~$\hat{u}_{2,2}$,~$(A_{2,1}, ..., A_{2,n})$,~$(\hat{Y}_{2,1}, ...,\hat{Y}_{2,n})$,~$(C_{2,1}, ..., C_{2,n})$,~$(w_{2,1}, ..., w_{2,n})$,~$(\sigma_{2,1}, ..., \sigma_{2,n})$~$\leftarrow$~\textbf{parse}($\pvss_2$)
 				\item \textbf{for} $0 \leq i \leq t$:
 				%\begin{itemize}
 				%\item 
 				$F_i \leftarrow F_{1,i}F_{2,i}$ 
 				%\end{itemize}
 				\item \textbf{for} $1 \leq i \leq n$:
 				\begin{itemize}
 					\item $A_i \leftarrow A_{1,i}A_{2,i}$; $\hat{Y}_i \leftarrow \hat{Y}_{1,i}\hat{Y}_{2,i}$; $w_i \leftarrow w_{1,i}+w_{2,i}$ 
 					\item \textbf{if} $\sigma_{1,i} \neq \bot$: $\sigma_i \leftarrow \sigma_{1,i}$, \textbf{else}: $\sigma_i \leftarrow \sigma_{2,i}$
 					\item \textbf{if} $\C_{1,i} \neq \bot$: $\C_i \leftarrow \sigma_{1,i}$, \textbf{else}: $\C_i \leftarrow \sigma_{2,i}$
 				\end{itemize}
 				\item $\hat{u}_2 \leftarrow \hat{u}_{1,2}\hat{u}_{2,2}$
 				\item \textbf{return} $\mathbf{F}$, $\hat{u}_2$,  $\mathbf{A}$,  $\mathbf{\hat{Y}}$, $\mathbf{C}$, $\boldsymbol{w}$, $\boldsymbol{\sigma}$ 
 			\end{itemize}
 			\vspace*{0.2 cm}
 			\State $\share(\dk_i, \pvss) \rightarrow \sh_i$
 			\begin{itemize}
 				\item $\mathbf{F}$, $\hat{u}_2$,  $\mathbf{A}$,  $\mathbf{\hat{Y}}$, $\mathbf{C}$, $\boldsymbol{w}$, $\boldsymbol{\sigma}$ $\leftarrow$ \textbf{parse}($\pvss$)
 				\item \textbf{return} $\mathbf{\hat{Y}_i}^{{dk_i}^{-1}}$ 
 			\end{itemize}
 			\vspace*{0.2 cm}
 			\State $\vrfyshare(j, \sh_j, \pvss) \rightarrow 0/1$.
 			\begin{itemize}
 				\item $\mathbf{F}$, $\hat{u}_2$,  $\mathbf{A}$,  $\mathbf{\hat{Y}}$, $\mathbf{C}$, $\boldsymbol{w}$, $\boldsymbol{\sigma}$ $\leftarrow$ \textbf{parse}($\pvss$)
 				\item \textbf{check} $e(A_j, \hat{h}_1) = e(g_1, \sh_j)$
 				\item \textbf{return} $1$ if the check pass, else \textbf{return} $0$
 			\end{itemize}
 			\vspace*{0.2 cm}
 			\State $\aggshares(\{(j, \sh_j)\}_{t}) \rightarrow a$
 			\begin{itemize}
 				\item $\Si \leftarrow \emptyset$
 				\item \textbf{for} all input $(j, \sh_j)$:
 				%\begin{itemize}
 				%\item 
 				\textbf{if} $\vrfyshare(j, \sh_j, \pvss) = 1$, then $\Si \leftarrow \Si \cup (j, \sh_j)$
 				%\end{itemize}
 				\item \textbf{return} $\prod_{i \in \Si}{\sh_i}^{l_{\Si, i}(0)}$	
 			\end{itemize}
 			
 			\vspace*{0.2 cm}
 			\State  $\vrfysecret(s, \pvss) \rightarrow 0/1$
 			\begin{itemize}
 				\item \textbf{check} 	$e(F_0, \hat{h}_1) = e(g_1, s))$
 				\item \textbf{return} $1$ if the check pass, else \textbf{return} $0$
 			\end{itemize}
 			\vspace*{0.2 cm}
 			\State $\weights(\pvss) \rightarrow  \boldsymbol{w}$
 			\begin{itemize}
 				\item $\mathbf{F}$, $\hat{u}_2$,  $\mathbf{A}$,  $\mathbf{\hat{Y}}$, $\mathbf{C}$, $\boldsymbol{w}$, $\boldsymbol{\sigma}$ $\leftarrow$ \textbf{parse}($\pvss$)
 				\item \textbf{return}  $\boldsymbol{w}$
 			\end{itemize}
 			
 		\end{algorithmic}
 		
 	\end{footnotesize}
 	%\end{scriptsize}
 	
 \end{algorithm}
 %\vspace{-0.4cm}
 
\noindent	
\subsection{Aggregatable PVSS}

 Gurkan et al.   \cite{gurkan2021aggregatable} constructed   aggregable $\PVSS$   by lifting Scrape  $\PVSS$ scheme from Cascudo et al. \cite{cascudo2017scrape}. 
We slightly  rephrase Gurkan et al.'s algorithm description, and present it in Alg. \ref{alg:aggpvss}.
Note that the participating parties use the same CRS consisting of (1) a bilinear group description $\mathsf{bp}$ (which fixes $g_1 \in \mathbb{G}_1$ and $\hat{h}_1 \in \mathbb{G}_2$), (2) a group element $\hat{u}_1 \in \mathbb{G}_2$ (3) encryption keys $ek_i \in \mathbb{G}_2$ for every party $\node_i$ with corresponding decryption keys $dk_i \in F$ known only to $\node_i$  such that $ek_i = h^{dk_i}$ and (4) verification keys $vk_i \in \mathbb{G}_2$ for every party $\node_i$ with corresponding secret keys $\sk_i \in F$ known only to $\node_i$  such that $\vk_i = {g_1}^{sk_i}$.

	Conditioned on SXDH assumption and PKI setup, 
	the above aggregatable $\PVSS$ scheme has a few nice security properties such as {\em verifiable commitment}, {\em verifiable aggregation} and {\em secrecy}:
	\begin{itemize}
%%%%%%%%%%%%%%%%%%%%%%%%%%%%%%%%%%%%%%%%%%%%%%%%%%%%%%%%%%%%%%%%%%%%%%%%%%%%%%%%%%%%%%%%%%%
		
		\item {\bf Verifiable commitment} indicates that any sharing script $\pvss$ can be verified by the public to tell whether it is valid to commit a fixed secret  $F^*(0)$ or not. Namely, if $\pvss$ is valid due to $\vrfyscript(\allek, \allvk, \pvss)=1$, it is guaranteed that: 
		\begin{itemize}
			\item There exists a fixed secret $F^*(0)$, such that $\vrfysecret(F^*(0), \pvss)=1$.
			%\item Any $f+1$  shares, that are   computed by distinct parties through using $\share$, can be combined by $\aggshares$ to yield the same fixed secret $F^*(0)$.
			\item Any $t$   decryption shares validated by $\vrfyshare$ with regard to the $\pvss$ script from   distinct parties (including up to $f$ malicious ones) can recover a      secret $F(0)$ same to $F^*(0)$.
		\end{itemize}
		
		%Since then, for any valid $\pvss$, we can assert that it must ``commit'' a fixed polynomial that is always the collectively reconstructed one despite the influence of adversary.
		
		%%%%%%%%%%%%%%%%%%%%%%%%%%%%%%%%%%%%%%%%%%%%%%%%%%%%%%%%%%%%%%%%%%%%%%%%%%%%%%%%%%%%%%
		\ignore{
			$$\Pr\Bigg[ 
			\begin{array}{lr}
				\aggshares({S}_1) &\\
				~~~~~~~ = &\\
				\aggshares({S}_2) 
			\end{array}
			~\Bigg|~  
			\begin{array}{lr}
				Q, Q_h, \pvss, \{(ek_j, dk_j)\}_{j \in \bar{Q}} \leftarrow  \adv(1^\lambda)    &\\
				Q \in \{P ~|~ P \in [n] \wedge |P|=f+1 \} &\\
				Q_h \in \{P ~|~ P \in [n] \setminus \bar{Q} \wedge |P|=f+1 \} &\\
				
				\vrfyscript(\allek, \pvss)=1 &\\
				
				%S_1 \leftarrow    \{(j,\hat{A}_j \leftarrow \share(\dk_j, \pvss) )\}_{j \in Q_1 \setminus \bar{Q} } &\\
				
				S_h \leftarrow    \{(j,\hat{A}_j \leftarrow \share(\dk_j, \pvss) )\}_{j \in {S}_h } &\\ 
				
				{S}_1 \leftarrow  \{(j,\hat{A}_j \leftarrow  \adv(1^\lambda)  )\}_{j \in {Q}_1 }  &\\
				{S}_2 \leftarrow  \{(j,\hat{A}_j \leftarrow  \adv(1^\lambda)  )\}_{j \in {S}_h }  &\\
				
				\forall (j, \hat{A}_j) \in  {S}_1, \vrfyshare(j, \hat{A}_j, \pvss)=1  &\\
				\forall (j, \hat{A}_j) \in  {S}_2, \vrfyshare(j, \hat{A}_j, \pvss)=1  &\\
				
			\end{array}
			\Bigg] \le \negl(\lambda)
			$$
		}
		%%%%%%%%%%%%%%%%%%%%%%%%%%%%%%%%%%%%%%%%%%%%%%%%%%%%%%%%%%%%%%%%%%%%%%%%%%%%%%%%%%%%%%
		
		\medskip
		\item {\bf Verifiable aggregation} means that any verified $\PVSS$ script $\pvss$ with weight tag $\boldsymbol{w}$ must indeed commit a secret that is a linear combination of participating parties' secrets due to $\boldsymbol{w}$. %Namely, if $\pvss$ is valid due to $\vrfyscript(\allek, \allvk, \pvss)=1$ and has weights $\boldsymbol{w}$ due to $\weights(\pvss)=\boldsymbol{w}$, then:
		\begin{itemize}
			\item If $\node_i$ is honest, 
			it is infeasible for the adversary to compute a $\PVSS$ script $\pvss_i$ by itself (without query $\node_i$ to get $\pvss_i$)  s.t., $\vrfyscript(\allek, \allvk, \pvss_i) = 1$  and $\weights(\pvss_{i})$ returns $\boldsymbol{w}$ with a non-zero $i^{th}$ position.
			%$F_i^*(0)$ must be the secret committed to a $\PVSS$ script $\pvss_i$ generated by $\node_i$.
			\item Moreover, $F^*(0) = \sum_{i=1}^{n} w_i F_i^*(0)$. Here $F^*(0)$ is the secret committed to a $\PVSS$ script $\pvss$, $w_i$ is the $i^{th}$ element in $\boldsymbol{w}=\weights(\pvss)$, 
			and $F_i^*(0)$ represents the secret that is committed to some  $\PVSS$ script $\pvss_i$ that is computed by the party $\node_i$.    
		\end{itemize}
		
		\medskip
		\item {\bf Secrecy} means that the adversary learns nothing besides the public knowledge through the $\PVSS$ script $\pvss$, unless the adversary gets the decrypted secret shares of $t-f$ honest parties. % send their decrypted secret  shares to the adversary.
		%in our setting,   the probability that the adversary  wins the  Weak Secrecy game down below is negligible. Note that the weak secrecy well captures the scenario that: the honest dealer uniformly chooses a secret over $\mathbb{Z}_q$ to share, and it is computationally infeasible for the adversary to compute the dealer's input secret on receiving the $\PVSS$ script carrying the secret, unless $t-f$ honest parties send their decrypted secret  shares to the adversary.
		%This is  at least not stronger than \cite{cascudo2017scrape}.
		%To formalize stronger secrecy to make  the adversary learn nothing besides the public knowledge through the $\PVSS$ script $\pvss$,
		%we actually can define {\em  Secrecy} through a distinguishing game where the adversary chooses two secrets $s_0$ and $s_1$ and lets the challenger randomly share one of them \cite{cascudo2017scrape}.
		%Nevertheless, we do not need such stronger secrecy property for $\PVSS$ in the paper.
		%Formally, this property states
		%To make it stronger for input over any distribution, both secrets $s_0$ and $s_1$ shall be chosen by the adversary, which we do not have to consider through the paper.
	\end{itemize} 
	
	\medskip
	With the above properties at hand, we actually see that the adversary to have negligible probability to win the following Prediction game. The property is also called   unpredictability by us, and captures the major threats in our $\seedgen$ protocol (soon to be explained). It intuitively states that if each honest party $\node_i$ randomly chooses the input secret $s_i$ committed to its $\PVSS$ script $\pvss_i$, then the adversary cannot compute the aggregated secret $s$ committed to any valid $\pvss$, as long as   $\pvss$  has a non-zero weight to reflect some honest party's $\PVSS$ script $\pvss_i$ is indeed aggregated to it. 
	
	\smallskip	
	\noindent
	{\bf Prediction game}. 
	{\em 
		The Prediction game between an adversary $\adv$ and a challenger $\cha$ is defined as follows for a $(n,t)$ aggregatable $\PVSS$ scheme with   weight tags in the presence of up to $f$ static corruptions ($f+1 \le t$):
		\begin{enumerate}
			\item The adversary chooses a set $\bar{Q}$ of   $f$ corrupted parties,  generates the public-private key pairs for each corrupted party in $\bar{Q}$, and sends all the public keys to the challenger $\cha$.
			\item The challenger $\cha$ generates all public-private key pairs for all honest parties, and sends these public keys to the adversary $\adv$.
			\item The adversary $\adv$ queries $\cha$ for each party $\node_i$ in $[n] \setminus \bar{Q}$, such that  $\cha$ randomly chooses $s_i \in \mathbb{Z}_q$, computes   $\pvss_i \leftarrow \deal(\allek,  \sk_i,s_i)$ and sends $\pvss_i$ to $\adv$.
			The adversary now produces a valid $\pvss$ script such that $\weights(\pvss)$ outputs $\vec{w}$ containing more than $f+k$ positions (where $1 \le k \le n-f$), and then sends $\pvss$ to $\cha$.
			\item The adversary asks the challenger to compute at most $t-f-1$ decryption shares on the received $\pvss$ and then send these shares back, and then
			the adversary guesses a value  $s^* \in \mathbb{Z}_q$. %; Finally, the adversary asks the challenger to compute  one more not queried decryption share on $\pvss$, and reconstructs the secret $s$ committed to $\pvss$.

		\end{enumerate}
		%The adversary $\adv$ wins in the above Weights-Forgery game if   $|\Pr[s^*=s]~|$.
		The adversary wins if $s^*=s$, where $s$ is the actual secret committed to $\pvss$. %(which is unique due to the commitment property).
	}
	
		\vspace{-0.5cm}
	\noindent
	\begin{lemma}[Unpredictability] 
		 There exists an  aggregatable $\PVSS$ construction \cite{gurkan2021aggregatable} based on SXDH assumption, such that no static adversary controlling up to $f$ parties can  win Prediction game with all but negligible probability.
	\end{lemma}
	
	\smallskip
	Note that Gurkan et al. \cite{gurkan2021aggregatable} did not formalize   unpredictability for their aggregatable $\PVSS$ scheme in our way. Nevertheless, as a trivial corollary, their construction  satisfies the unpredictability property (that can   simplify the  proof   for the security of our $\seedgen$ protocols), because failing to satisfy unpredictability  would directly break the authors' DKG scheme  (as the adversary can directly tell the   aggregated secret key to break DKG).

	\begin{algorithm}[!h]
		%\captionsetup{font={scriptsize}}
		\begin{footnotesize}
			%\begin{scriptsize}
			
			\caption{$\seedgen$ protocol with identifier $\ID$ and leader $\node_L$}
			\label{alg:seed}
			\begin{algorithmic}[1]
				
				\vspace*{0.1cm}
				\Statex {\color{blue}  /* Protocol for {each party $\node_i$} */ }
				\vspace*{0.1 cm}
				
				\Upon {being activated}
				\State randomly sample a secret $s$, $\pvss_i \leftarrow \deal(\allek,  \sk_i, s)$,
				and \textbf{send} $\SCRIPT(\ID,\pvss_i)$ to $\node_L$
				\EndUpon
				
				\vspace*{0.1 cm}
				\Upon {receiving $\LOCKPVSS(\ID,\pvss)$ message from $\node_L$ for the first time}
				\If {{$\vrfyscript(\allek, \allvk, \pvss) = 1$ $\wedge$ ($\weights(\pvss)$ contains $2f+1$ ones)}}
				\State record $\pvss$, $\sigma_i \leftarrow \Sign^\ID_{i}(\pvss)$
				and  \textbf{send} $\CONFIRMPVSS(\ID,\sigma_i)$ to $\node_L$
				\EndIf
				\EndUpon
				
				\vspace*{0.1 cm}
				\Upon {receiving $\COMMITPVSS(\ID,  \Sigma)$ from $\node_L$ for the first time}
				\If{ $\Sigma$ contains $2f+1$ valid signatures  for a recorded $\pvss$ from distinct parties}
				\State  $\sh_i \leftarrow \share(\pvss)$
				and  $ \textbf{send}$ $\SEEDSHARE(\ID,\sh_i)$ to $\node_L$
				\EndIf
				\EndUpon
				
				\vspace*{0.1 cm}
				\Upon {receiving  $\SEED(\ID,  \Sigma,seed)$ from $\node_L$ for the first time}
				\If {  $\vrfysecret(seed, \pvss)=1$ $\wedge$ ($\Sigma$ contains  $2f+1$ valid signatures  for $\pvss$ from distinct  parties)}
				\State  $\textbf{send}$ $\SEEDECHO(\ID, seed)$ to all parties
				\EndIf
				
				\EndUpon
				
				\vspace*{0.1 cm}
				\Upon {receiving $2f+1$ $\SEEDECHO(\ID, seed)$   from distinct parties}
				\State \textbf{send} $\SEEDREADY(\ID, seed)$  to all parties if $\SEEDREADY$ not sent yet
				\EndUpon
				
				\vspace*{0.1 cm}
				\Upon {receiving $f+1$ $\SEEDREADY(\ID, seed)$   from distinct parties}
				\State \textbf{send} $\SEEDREADY(\ID, seed)$  to all parties if $\SEEDREADY$ not sent yet
				\EndUpon
				
				\vspace*{0.1 cm}
				\Upon {receiving $2f+1$ $\SEEDREADY(\ID, seed)$    from distinct parties}
				\State \textbf{output} $seed$
				\EndUpon
				
				\vspace*{0.2 cm}
				\Statex {\color{blue}  /* Protocol for {the leader $\node_L$} */ }

				\vspace*{0.1 cm}
				\Upon {receiving $\SCRIPT(\ID,\pvss_j)$ from $\node_j$ for the first time}%$(\SCRIPT, \pvss_j, \pvsstag_j)$ 
				\If {$\vrfyscript(\allek,  \pvss_j) = 1$ $\wedge$ (the weights of $\pvss_j$ are all zeros but one at the $j^{th}$ position)}
				\State $\K \leftarrow \K \cup \{\pvss_j\}$
				\If {$|K|=2f+1$}
				\State $\pvss \leftarrow \agg(\K) $
				and \textbf{send} $\LOCKPVSS(\ID,\pvss)$ to all parties
				\EndIf
				\EndIf
				\EndUpon
				
				\vspace*{0.1 cm}
				\Upon  {receiving $\CONFIRMPVSS(\ID,\sigma_j)$ message from $\node_j$ for the first time}
				\Wait { for $\pvss$ is recorded} \EndWait
				\If {$\Verify^\ID_j(\pvss,\sigma_j)=1$}
				\State $\Sigma \leftarrow \Sigma \cup \{(j, \sigma_j)\}$
				\If {$|\Sigma|=2f+1$}
				\textbf{send} $\COMMITPVSS(\ID,  \Sigma)$ to all parties
				\EndIf
				\EndIf
				\EndUpon
				
				\vspace*{0.1 cm}
				\Upon  {receiving $\SEEDSHARE(\ID,\sh_j)$ message from $\node_j$ for the first time}
				\If {$\vrfyshare(\sh_j,\pvss)=1$}
				\State $S \leftarrow S \cup \{(j, \sh_j)\}$
				\If {$|S|=2f+1$}
				$seed \leftarrow \aggshares(S)$
				and \textbf{send} $\SEED(\ID,  \Sigma, seed)$ to all parties
				\EndIf
				\EndIf
				\EndUpon	
				
			\end{algorithmic}
			
		\end{footnotesize}
		%\end{scriptsize}
		
	\end{algorithm}

\subsection{Seeding from Aggregatable PVSS}
Here down below we describe an exemplary $\seedgen$ construction  from the $(n, 2f+1)$ aggregatable $\PVSS$ scheme (as detailed by Alg. \ref{alg:seed}). 
Recall that the protocol  is a two-phase protocol (including a committing phase and a revealing phase), and its execution can be briefly described as follows: %\yuan{add protocol description here...}
\begin{itemize}
	\item {\em Committing Phase I -- Seed aggregation} (Line 1-2, 18-22).In this phase, each party invokes $\deal$ to create a $\pvss$ script and send this script to the leader $\node_L$. When the leader receives $\pvss_j$ from $\node_j$, it verifies the $\pvss$ using $\ek$ and checks if the weights of $\pvss_j$ are all zeros but one at the $j^{th}$ position. Once collecting $2f+1$ valid $\pvss$ scripts, the leader aggregates them and send it using a $\LOCKPVSS $ message.
	\item {\em Committing Phase II -- Seed commitment} (Line 3-8, 23-27). After receiving a $\LOCKPVSS$ message with a valid $\pvss$. Each party sign for it and send the signature $\sigma_i$ to the leader to commit the same $\pvss$. Upon receiving $2f+1$ valid signature for the $\pvss$, the leader send a $\COMMITPVSS$ message containing a signature set $\Sigma$. After receiving a valid $\COMMITPVSS$ messages from the leader, each party confirms that the output is actually fixed, and it takes the $\pvss$ as a input in $\share $ to output a share $\sh_i$ regarding the secret committed to $\pvss$.
	\item {\em Revealing Phase -- Seed recovery} (Line 9-17, 28-31).  The leader collects $2f+1$ valid shares which are committed to $\pvss$ and aggregates them to a $seed$. Then the $seed$ is sent to all parties with  $\Sigma$. Once honest parties receiving a valid seed with a hash value of $\pvss$ and $2f+1$ valid signatures for $\pvss$, the execution is just like a reliable broadcast. Specifically, each party sends $\SEEDECHO$ to all parties. After receiving  $2f+1$  $\SEEDECHO$ messages, honest parties send $\SEEDREADY$ to all parties. If an honest party receives $f+1$ $\SEEDREADY$ messages of same value and $\SEEDREADY$ has not been sent, it sends $\SEEDREADY$. When receiving $2f+1$ $\SEEDREADY$ messages of same value, honest parties output the $seed$. 
\end{itemize}

	\begin{lemma}\label{lemma:seed-com1}
		Upon any honest party completes the   committing phase,  any two valid $\COMMITPVSS$ messages $\COMMITPVSS(\ID, \Sigma)$ and $\COMMITPVSS(\ID, \Sigma')$   must have $\Sigma$ and $\Sigma'$ containing $n-f$ valid signatures for the same $\pvss$, and   there exists a fixed value $seed$ associated   to this $\pvss$.
	\end{lemma}
	\begin{proof}
		Assume that $\node_i$ is the first party who starts to run the revealing phase of  $\seedgen$ protocol, it implies that $\node_i$ received a valid  $\COMMITPVSS(\ID, \Sigma)$ messag from leader $\node_L$. If another honest party $\node_j$ also received a valid message $\COMMITPVSS(\ID, \Sigma')$ from leader $\node_L$, where signatures in $\Sigma$ are signed for $\pvss$ while signatures in $\Sigma'$ are signed for $\pvss'$, since a valid $\Sigma$ contains $2f+1$ valid signatures for a same $\pvss$ from distinct  parties. According to the the unforgeability of digital signatures, it induces that at least one honest party signed for both $\pvss$ and $\pvss'$, which is impossible because each honest party signs at most once. Hence, when some honest party $\node_i$ starts to run the revealing phase, the $\Sigma$ from any valid $\COMMITPVSS(\ID, \Sigma)$ message is for the same $\pvss$. Following the commitment of the $\PVSS$ scheme, there exists a fixed value $seed$ corresponding to the $\pvss$.
	\end{proof}

	\begin{lemma}\label{lemma:seed-com2}
		If any twos honest party output for $\ID$, then they output the same value. 
	\end{lemma}
	\begin{proof}
		Suppose that some honest party outputs $seed'$ from the $\seedgen$. By the code, it receives $2f+1$ $\SEEDREADY$ messages containing $seed'$. Then at least one honest party received $2f+1$ valid $\SEEDECHO$ messages with the same $seed'$ from distinct parties, which means that at least $f+1$ honest parties received valid $\SEED(\ID, \Sigma,seed')$ message from the leader. From the previous analysis, no honest party will accept a $seed' \neq seed$ from $\node_L$ or multicast it. Thus, $seed' =seed$.
	\end{proof}

	\noindent
	{\bf Deferred proofs for Lemma \ref{lemma:seed}}. Here we prove that the protocol in Alg. \ref{alg:seed} satisfies all properties of Reliable Broadcasted Seeding given in Definition \ref{def:seeding} and analyze how does it incur only quadratic messages and communications.
	%\yuan{make up proofs here...}
	\begin{proof}
		Here prove that Alg. \ref{alg:seed} satisfy the $\seedgen$ properties one by one:
		\begin{itemize}
			\item {\em Totality}. Assume that an honest party outputs in the $\seedgen$, it must have received $2f+1$ $\SEEDREADY$ messages. At least $f+1$ of the messages are sent from honest parties. Therefore, all parties will eventually receive $f+1$ $\SEEDREADY$ messages from these honest parties and then send a $\SEEDREADY$ messages as well (if a $\SEEDREADY$ message has not been sent yet). So the totality property always holds.
			%\item {\em Agreement}. Suppose that $seed_i \neq seed_j$. $\node_i$ receives $2f+1$ $\SEEDREADY$ messages containing $seed_i$, the senders of which include at least $f+1$ honest parties. In the same way, $\node_j$ must have received  at least $f+1$ $\SEEDREADY$ messages containing $seed_j$ from honest parties. So it induces that at least one honest party sent two different messages, which is impossible. So there is a contradiction if $seed_i \neq seed_j$, implying $seed_i = seed_j$.
			\item {\em Correctness}. In the seed aggregation phase, the leader will collect $2f+1$ valid $\pvss_i$ scripts and aggregate them into a $\pvss$ whose weight has $2f+1$ positions to be 1. In the seed commitment phase, all honest parties sign for this $\pvss$ so that the leader can collect at least $n-f$ valid signatures for $\pvss$ to form valid $\Sigma$. In the seed recovery phase, the leader can collect at least $n-f$ valid shares which are committed to the $\pvss$ they have signed for. All honest parties can receive the same $\Sigma$ and $seed$ that pass verifications. So they would broadcast the same $\SEEDECHO$ message and the same $\SEEDREADY$ messages, thus finally outputting in the $\seedgen$ instance.
			
			\item {\em Commitment}. From Lemma \ref{lemma:seed-com1}, upon any honest party completes the protocol's committing phase, there exists a fixed value  $seed$ corresponding to the unique $\pvss$.
			%Assume that $\node_i$ is the first party who starts to run the revealing phase of  $\seedgen$ protocol, it implies that $\node_i$ received a valid message $\COMMITPVSS(\ID, \Sigma)$ from leader $\node_L$. If another honest party $\node_j$ also received a valid $\COMMITPVSS(\ID, \Sigma')$ message from leader $\node_L$, where $\Sigma' \neq \Sigma$, since a valid $\Sigma$ contains $2f+1$ valid signatures for a same $\pvss$ from distinct  parties, it induces that at least one honest party signed for both $\pvss$ and $\pvss'$, which is impossible. Hence, when some honest party $\node_i$ starts to run the revealing phase, the $\Sigma$ from any valid $\COMMITPVSS(\ID, \Sigma)$ message is for the same $\pvss$. Following the commitment of the $\PVSS$ scheme, there exists a fixed value $seed$ corresponding to the $\pvss$.
			From Lemma \ref{lemma:seed-com2}, if any honest party outputs for $\ID$, then it outputs value $seed$. Thus the commitment can be proved.

			%the leader of which include at least $f+1$ honest parties. In the same way, $\node_j$ must have received  at least $f+1$ $\SEEDREADY$ messages containing $seed_j$ from honest parties. So it induces that at least one honest party sent two different messages, which is impossible. So there is a contradiction if $seed_i \neq seed_j$, implying $seed_i = seed_j$.

			\item {\em Unpredictability}. Prior to   $f+1$  honest parties are activated to run the revealing phase of the $\seedgen$ protocol, the adversary can only collect at most $2f$ decryption shares for the committed $\pvss$ script. %Note that each honest party randomly samples a secret. 
			Trivially according to the unpredictability of $\PVSS$ with weight tags, since the aggregated $\pvss$ has a weight with $2f+1$ non-zero positions, it is infeasible for the adversary to compute a $seed^* = seed$ at the moment, where $seed$ is the actual  secret committed  to the aggregated $\pvss$ script. 
			
		\end{itemize}
		
		The complexities   can be easily seen as follows: The message complexity of  $\seedgen$ is $O(n^2)$, which is due to each party sends $n$ $\SEEDECHO$ and $\SEEDREADY$ messages; considering that the input secret $s$ and $\pvss$ both are $\bigO(\lambda)$ bits, and there are $O(n)$ messages with $\bigO(\lambda n)$ bits and $O(n^2)$ messages with $\bigO(\lambda)$ bits, thus the communication complexity of the protocol is of overall $O(\lambda n^2)$ bits.
		
	\end{proof}

	\end{subappendices}

\end{document}